\documentclass[a4paper, 11pt, one column]{article}

\usepackage[english]{babel}
\usepackage{latexsym,amssymb,amsmath,amsthm,amsfonts,enumerate,verbatim,xspace,
exscale}

\usepackage[top=1.3cm, bottom=2.0cm, outer=2.5cm, inner=2.5cm, heightrounded,
marginparwidth=1.5cm, marginparsep=0.4cm, margin=2.5cm]{geometry}

\usepackage{graphicx} 
\usepackage[colorlinks=False]{hyperref} 
\usepackage{amsmath}  
\usepackage{amsfonts} 
\usepackage{amssymb,amsthm}  
\usepackage{color}
\usepackage[all]{xy}

\usepackage{tikz} 

\newtheorem{theorem}{Theorem}[section]
\newtheorem{corollary}[theorem]{Corollary}
\newtheorem{lemma}[theorem]{Lemma}
\newtheorem{proposition}[theorem]{Proposition}
\newtheorem{remarkth}[theorem]{Remark}
\newenvironment{remark}{\begin{remarkth}\upshape}{\hfill$\diamond$\end{remarkth}}

\theoremstyle{definition}
\newtheorem{definition}[theorem]{Definition}

\newtheorem{example}[theorem]{Example}

\newcommand{\llbracket}{\lbrack\! \lbrack}
\newcommand{\rrbracket}{\rbrack\! \rbrack}
\def\subM{{\mbox{\tiny{$\mathcal{M}$}}}}

\def\subQ{{\mbox{\tiny{$Q$}}}}
\def\subP{{\mbox{\tiny{$P$}}}}
\def\red{{\mbox{\tiny{red}}}}
\def\Hor{{\mbox{H\tiny{or}}}}
\def\subB{{\mbox{\tiny{$B$}}}}
\def\nh{{\mbox{\tiny{nh}}}}
\def\g{\mathfrak{g}}
\def\M{\mathcal{M}}
\def\C{\mathcal{C}}
\def\vecOm{\boldsymbol{\Omega}}

\title{Nonholonomic momentum map reduction and a Chaplygin-type foliation} 
\author{Paula Balseiro\textsuperscript{1,a} and Danilo Machado Tereza\textsuperscript{1,b}}

\date{}

\begin{document}
\maketitle

\let\thefootnote\relax\footnotetext{\textsuperscript{1.} Universidade Federal Fluminense, Instituto de Matemática e Estatística, Rua Mario Santos Braga S/N, 24020-140, Niterói, Rio de Janeiro, Brasil. \\ \textsuperscript{a.} \href{mailto:pbalseiro@id.uff.br}{pbalseiro@id.uff.br}; and \textsuperscript{b.} \href{mailto:danilomt@id.uff.br}{danilomt@id.uff.br}}

\begin{abstract}    
    This paper presents a set-up for momentum map reduction of nonholonomic systems with symmetries, extending previous constructions in \cite{PaulaGarciaToriZuccalli, RecCortes2002}, based on the existence of certain conserved quantities and making essential use of the nonholonomic momentum bundle map of \cite{Bloch1996}. We show that the reductions of the momentum level sets carry an almost symplectic form codifying the reduced dynamics. These reduced manifolds are the leaves of the foliation associated with an almost Poisson bracket obtained by a (dynamically compatible) gauge transformation of the nonholonomic bracket. We show that each leaf is a {\em Chaplygin-type leaf}, in the sense that it is isomorphic to a cotangent bundle with the canonical symplectic form plus a ``magnetic term''.      
\end{abstract}

\setcounter{tocdepth}{2} 
\tableofcontents 

\newpage
\section{Introduction} \label{S:Intro}

In this paper we study a momentum map reduction for nonholonomic systems with symmetries. Such systems are described in terms of almost Poisson structures, so many of the usual properties of Hamiltonian actions do not hold in this setting. To overcome the difficulties posed by this fact, upon the assumption of existence of special conserved quantities we consider modifications of the nonholonomic brackets by suitable {\it gauge transformations by 2-forms} which allow us to effectively use the so-called ``nonholonomic momentum map'', originally defined in \cite{Bloch1996}, in the reduction procedure. The resulting reduced spaces are almost symplectic manifolds of ``Chaplygin-type'', in the sense that they can be identified with cotangent bundles equipped with the canonical symplectic form plus a (non-closed) magnetic term, and they can be seen as leaves of a foliation of an almost Poisson structure that describes the reduced nonholonomic dynamics. The results of the present paper allow us to understand the role of the nonholonomic momentum map in relating nonholonomic brackets to symmetries.
 
The usual setting for Hamiltonian reduction of Poisson manifolds
(see e.g. \cite{PoissonMarsdenRatiu}) is that of a $G$-invariant Poisson manifold $(P,\{\cdot,\cdot\}_{\subP})$ with a momentum map $J : P\rightarrow \mathfrak{g}^{\ast}$ \cite{MR1999}, i.e.,
$J$ is $Ad^{\ast}$-equivariant and, for all $\eta\in\mathfrak{g}$,  
\begin{equation} \label{Eq:Mom-Def}
    \{\cdot,J_\eta \}_\subP = \eta_\subP,
\end{equation}
where $\eta_\subP$ is the infinitesimal generator of the action at $\eta$. Then, with usual regularity conditions,
the quotient $J^{-1}(\mu)/G_{\mu}$ is a manifold that inherits a unique Poisson structure $\{\cdot,\cdot\}_{\mu}$ defined, for $f,h\in C^\infty(J^{-1}(\mu)/G_\mu)$, by 
\begin{equation}\label{Eq:Marsden-Ratiu-Red1}
   \{f,h\}_{{\mu}}\circ \rho_{\mu} =  \{F,H\}_{\subP}\circ \iota_{\mu},
\end{equation}
where $F,G$ are smooth extensions to $P$ of the functions $f\circ \rho_{\mu}$ and $g\circ \rho_{\mu}$ defined on $J^{-1}(\mu)$ (with differentials vanishing on tangent spaces to the $G$-orbits) and with $\iota_{\mu} : J^{-1}(\mu) \rightarrow P$ and $\rho_{\mu} : J^{-1}(\mu) \rightarrow J^{-1}(\mu)/G_{\mu}$ the inclusion map and the orbit projection, respectively.  

Note that condition \eqref{Eq:Mom-Def} implies the following (inter-related) properties:

\begin{enumerate}
    \item[$(i)$]  The vertical distribution (tangent to the $G$-orbits) belongs to the characteristic distribution of the Poisson bracket.  
    
    \item[$(ii)$] For a $G$-invariant Hamiltonian, the functions $J_\eta$ are conserved by the dynamics. 
    
    \item[$(iii)$] The infinitesimal generators of the $G$-action are Hamiltonian vector fields: $\eta_P = X_{J_\eta}$.  
\end{enumerate}

In the nonholonomic setting these three conditions are not necessarily satisfied and we have to understand how infinitesimal generators, conserved quantities and momentum map interact. 

Recall that a nonholonomic system on a manifold $Q$ is a mechanical system with constraints in the velocities, that is, the permitted velocities define a (nonintegrable) subbundle $D\subset TQ$ \cite{Bloch1996,Bloch2003,RecCortes2002,Cushman2015}. Classical examples are mechanical systems with rolling constraints such as a disk rolling on a plane or, more generally, solids --such as spheres, ellipsoids, snakeboards-- rolling without sliding on different smooth surfaces, see e.g., \cite{Bloch2003,BorisovMamaev2002}.  
From a geometric perspective, a nonholonomic system is determined by a triple $(\M,\{\cdot, \cdot\}_\nh,H_\subM)$ where $(\mathcal{M},\{\cdot, \cdot\}_{\nh})$ is an almost Poisson manifold, with $\M$ a submanifold of $T^*Q$ induced by $D$, and $H_\subM$ the corresponding Hamiltonian function (of mechanical type). The almost Poisson bracket $\{\cdot, \cdot\}_{\nh}$ is called the {\it nonholonomic bracket} \cite{IbLMM99, Marle98, Schaft-Maschke-1994} and it has a non-integrable characteristic distribution denoted by $\C$. The nonholonomic dynamics is given by the integral curves of the nonholonomic vector field $X_\nh$ on $\M$ given by $X_\nh = \{\cdot, H_\subM\}_{\nh}$.

We consider nonholonomic systems admitting symmetries given by the (free and proper) action of a Lie group $G$ on $Q$. The lifted action to $T^*Q$ leaves $\M$ and the bracket $\{\cdot, \cdot\}_\nh$ invariant.  The usual scenario for $G$-invariant nonholonomic system is to consider the {\it dimension assumption} \cite{Bloch1996}, that is, denoting by $\mathcal{V}$ the vertical distribution on $\M$, we assume that 
$$
    T\M= \mathcal{C} + \mathcal{V}.
$$
Note that the dimension assumption goes in the opposite direction of condition $(i)$ for Hamiltonian reduction above, since the vertical spaces complement the characteristic distribution.  
Following \cite{Bloch1996} we define the  (usually not integrable) distribution $\mathcal{S}:=\mathcal{C}\cap \mathcal{V}$ and the bundle $\mathfrak{g}_S\to Q$ with fibers given by $(\mathfrak{g}_S)|_q :=\{\xi_q\in \mathfrak{g}: (\xi_q)_\subM(m)\in \mathcal{S}_{m} \ \mbox{for} \ \tau_\subM(m)=q\}$. The nonholonomic momentum bundle map $J_{\nh} :\M\to \mathfrak{g}^*_S$ is the restriction of the canonical momentum map to $\M$.   

The novelty of the present paper is to address the issue of reduction of $G$-invariant nonholonomic system making use of the nonholonomic momentum bundle map.   

In contrast with properties $(i)$, $(ii)$ and $(iii)$ for Poisson reduction,  we have the following.

\begin{enumerate}
    \item[$(i')$] Property $(i)$ is replaced by the fact that only $\mathcal{S}\subset \mathcal{V}$ belongs to the characteristic distribution of $\{\cdot, \cdot\}_\nh$.
    
    \item[$(ii')$]  Instead of considering momentum map components $J_\eta$ corresponding to Lie algebra elements $\eta\in \mathfrak{g}$ as in $(ii)$, one should consider functions $J_\xi = \langle J_\nh, \xi\rangle$, for $\xi$ a section of the bundle $\mathfrak{g}_S$. A key difference is that such functions are not necessarily first integrals of the nonholonomic system (in spite of the $G$-invariance of $H$).  The existence of special sections $\xi$ for which $J_\xi$ is a first integral is in general an open problem, and first integrals of this form are called {\it horizontal gauge momenta} \cite{BatesGraumann96,FGS2005,Fasso2008}. 
        
    \item[$(iii')$] The fact that there is a section $\xi$ that generates a first integral $J_\xi$ does not necessarily mean that the infinitesimal generator $\xi_\subM$ is a Hamiltonian vector field of $\{\cdot, \cdot\}_\nh$. Following \cite{Balseiro2021,LuisG2016}, and assuming the existence of $k:=\textup{rank}(\mathcal{S})$ horizontal gauge momenta, we may consider a new bracket $\{\cdot, \cdot\}_\subB$ through a {\it (dynamical) gauge transformation} by a 2-form $B$. With respect to this new bracket, infinitesimal generators (defining first integrals) are Hamiltonian vector fields $\{\cdot , J_\xi \}_\subB= \xi_\subM$,  and the new bracket still codifies the nonholonomic dynamics in the sense that $\{\cdot , H_\subM \}_\subB= X_\nh$.  
\end{enumerate}

Throughout this work, as in \cite{Balseiro2021}, {\em we assume the existence of exactly $k:=\textup{rank}(\mathcal{S})$ $G$-invariant and functionally independent horizontal gauge momenta} (see $(ii')$). We stress that this assumption is satisfied in many examples of interest.
Once we fix $Ad$-invariant sections $\{\xi_1,\dots,\xi_k\}$ of $\mathfrak{g}_S$ generating the horizontal gauge momenta, we take
the dual $Ad^*$-invariant frame  $\{\mu^1, \dots , \mu^k\}$ of $\mathfrak{g}^*_S$ and consider momentum levels with respect to $\mu = c_i\mu^i$, for $c_i \in \mathbb{R}$ (rather than just elements in $\mathfrak{g}^*$).
Theorem~\ref{T:NH_momentum_red} states that the submanifold $J_\nh^{-1}(\mu)$ is $G$-invariant and $J_\nh^{-1}(\mu)/G$ admits an almost Poisson bracket  $\{\cdot, \cdot \}_\mu$ analogous to \eqref{Eq:Marsden-Ratiu-Red1}, that is, for $f,g\in C^\infty(J_\nh^{-1}(\mu)/G)$, 
\begin{equation*}
    \{f,h\}_{{\mu}}\circ \rho_{\mu} =  \{F,H\}_{\subB}\circ \iota_{\mu},
\end{equation*}
where $F,H$ are appropriate extensions of $\rho_\mu^*f$ and $\rho^*_\mu h$.  
In this case, we show that the almost Poisson bracket $\{\cdot, \cdot \}_\mu$ is nondegenerate, so it is equivalent to an almost symplectic 2-form $\omega_\mu^\subB$. It follows that the reduced spaces $(J_\nh^{-1}(\mu)/G, \omega_\mu^\subB)$ give rise to an almost symplectic foliation on $\M/G$, whose leaves integrate the characteristic distribution of the reduced bracket on $\M/G$ induced by $\{\cdot , \cdot\}_\subB$ (this gives an alternative proof to the fact that this bracket is twisted Poisson, which was the main result in \cite{Balseiro2021}). 

We verify that each reduced space $(J_{\nh}^{-1}(\mu)/G, \omega_\mu^\subB)$ is diffeomorphic to $(T^{\ast}(Q/G),\omega_{\mbox{\tiny{can}}}+\hat{\mathcal{B}}_{\mu})$, where 
$\hat{\mathcal{B}}_{\mu}$ is a semi-basic 2-form (with respect to the bundle $T^*(Q/G)\to Q/G$) that can be thought of as a ``{\it magnetic term}''. This result, found Theorem~\ref{T:PB-mu-identification}, bears resemblance with the one in \cite[Sec.~2.3]{MarMiPerRat:Book} for symplectic manifolds but in our case the magnetic term $\mathcal{B}_\mu$ is not necessarily closed and it depends exclusively on the 2-form $B$ defining the gauge transformation.  
Therefore, we conclude that a nonholonomic system with $k$ ($G$-invariant) horizontal gauge momenta can be reduced to a {\it Chaplygin-type} foliation, i.e., a foliation whose leaves look like Chaplygin systems, even if the original system is not Chaplygin.  



We remark that we express our reduction procedure in terms of (almost) Dirac structures \cite{DiracCourant1990}, even though the results can be exposed without mentioning them.  The issue is that when we want to ``pull back'' the almost Poisson bracket $\{\cdot, \cdot\}_\subB$ to the submanifold $J_\nh^{-1}(\mu)$,  we lose the bracket structure but obtain an almost Dirac structure $L_\mu^\subB$ which can be then ``pushed forward'' to the 2-form $\omega^\subB_\mu$ on $J_\nh^{-1}(\mu)/G$:

\begin{equation*}
    \xymatrix{ (J^{-1}_\nh(\mu), L_\mu^\subB, H_{\mu}) \ar[r]^{\ \ \iota_\mu} \ar[d]^{\rho_\mu} & (\M,\{\cdot, \cdot\}_\subB, H_{\subM}) \\
(J^{-1}_\nh(\mu)/G, \omega_{\mu}^\subB,H_\red^{\mu}) & }
\end{equation*}

We also initiate the study of systems with conserved quantities that are not horizontal gauge momenta.  In Theorem \ref{T:general0-level}, we observe that the common 0-level set of the conserved quantities is also identified with a Chaplygin system.  As a direct application of this result, we recover, from an intrinsic geometric view point, that the 0-level set of the conserved quantity of the mechanical system  described by a (inhomogeneus) ball rolling on a fixed sphere \cite{BorisovFedorov1995, BorisovMamaev2002,Jovanovic2010} is diffeomorphic to a cotangent bundle endowed with the canonical form plus a ``magnetic'' semi-basic 2-form.  

Comparing with previous versions of nonholonomic reduction in the literature, we note the following:
\begin{itemize}
\item The systems studied here are not of Chaplygin-type, not even decomposed into a Chaplygin system with an extra symmetry, as the ones treated in \cite{PaulaGarciaToriZuccalli}. In fact, with our more general framework we study the homogeneous ball rolling on a surface of revolution (see e.g. \cite{Balseiro2021, BorisovKilinMamaev2002}), which cannot be handled in \cite{PaulaGarciaToriZuccalli}.  

\item In the nonholonomic momentum map reduction in \cite{RecCortes2002}, the relation between infinitesimal generators and Hamiltonian vector fields is not taken into account, and no (almost) symplectic/Poisson manifold are present in the reduced spaces. 

\end{itemize}

The paper is organized as follows: In Section \ref{S:Dirac} we present the basic formalism for Dirac structures and the concept of a {\it momentum bundle map} with a suitable reduction.  In Section \ref{S:NHSystems} we study nonholonomic systems and the nonholonomic momentum bundle map reduction. The identification with a Chaplygin-type foliation is in Section \ref{S:Identification}. Section \ref{S:general-0level} contains the study of the 0-level set of first integrals that are {\it not} horizontal gauge momenta, while Section \ref{S:Examples} has the examples.

\noindent {\bf Acknowledgements:} The authors would like to thank to David Iglesias Ponte and Nicola Sansonetto for useful conversations. P.B. thanks CNPq (Brasil) and D.M.T. thanks
Faperj (Brasil) and CAPES (Brasil) for financial support.

\section{Dirac structures with symmetries} \label{S:Dirac}

In this section, we present the framework of Dirac structures \cite{DiracCourant1990} and {\it momentum bundle maps} in order to proceed with a  momentum map reduction. We will recover, in particular, the case studied in \cite{PoissonMarsdenRatiu}. 

The starting geometric structure, when studying nonholonomic systems is an almost Poisson bracket. However in the reduction process using a momentum map we have to pass through (almost) Dirac structures. This is our main motivation to study such general structures and the interaction with momentum bundle maps. In this section we follow the framework and notation of \cite{DiracIntro2011,BursztynCrainic2005,DiracMorita2004}. 

\subsection{Dirac structures} \label{Ss:DiracStructures}

Consider a smooth manifold $P$ and the bundle $\mathbb{T}P := TP\oplus T^{\ast}P$ on $P$ with two additional structures 
\begin{enumerate}
    \item[$(i)$] a nondegenerate, symmetric fibrewise bilinear form $\langle\cdot,\cdot\rangle : \mathbb{T}P\times \mathbb{T}P \rightarrow \mathbb{R}$ given, at each $x\in P$, $X,Y\in T_{x}P$ and $\alpha,\beta\in T_{x}^{\ast}P$, by
    \begin{equation*}
        \langle (X,\alpha),(Y,\beta)\rangle = \beta(X)+\alpha(Y);
    \end{equation*}
    
    \item[$(ii)$] a bracket $\llbracket \cdot,\cdot\rrbracket : \Gamma(\mathbb{T}P)\times\Gamma(\mathbb{T}P)\rightarrow \Gamma(\mathbb{T}P)$ called the {\it Courant bracket} \cite{DiracCourant1990}, defined by
    \begin{equation*}
        \llbracket (X,\alpha),(Y,\beta)\rrbracket = ([X,Y], \mathcal{L}_{X}\beta-{\bf i}_{Y}d\alpha).
    \end{equation*}
\end{enumerate}
An \textit{almost Dirac structure} on $P$ is a maximal isotropic subbundle $L$ of $\mathbb{T}P$ with respect to $\langle\cdot,\cdot\rangle$; if in addition $L$ is involutive with respect to the Courant bracket (i.e., $\llbracket \Gamma(L),\Gamma(L)\rrbracket \subset \Gamma(L)$), $L$ is referred as a \textit{Dirac structure} on $P$, \cite{DiracCourant1990}. Let us equip the bundle $\mathbb{T}P$ with their natural projections
$$
    pr_{T} : \mathbb{T}P\rightarrow TP\quad\textup{and}\quad pr_{T^{\ast}} : \mathbb{T}P\rightarrow T^*P. 
$$ 
An almost Dirac structure $L$ is \textit{regular} if the so-called {\it characteristic distribution} $pr_{T}(L)\subset TP$ on $P$ is regular. If $L$ is Dirac, then the distribution $pr_{T}(L)$ is integrable and it defines a presymplectic foliation. 
The {\it null distribution} of $L$ is the distribution $K_L$ on $P$ defined by 
\begin{equation*}
    K_{L} := pr_{T}((TP\oplus \{0\})\cap L)\subset TP.  
\end{equation*}

\begin{remark}\label{R:Dirac:BasicExamples}
    Given a $2$-form $\omega \in \Omega^{2}(P)$ and a bivector field $\pi\in \Gamma(\Lambda^{2}TP)$ with their respective bundle maps $\omega^{\flat} : TP\rightarrow T^{\ast}P$ such that $\omega^\flat(X)= {\bf i}_{X} \omega$ and $\pi^{\sharp} : T^{\ast}P\rightarrow TP$ so that $\pi^\sharp(\alpha)={\bf i}_{\alpha} \pi$, then $L_{\omega} := \textup{graph}(\omega^{\flat})$  and $L_{\pi} := \textup{graph}(\pi^{\sharp})$ are almost Dirac structures, which become Dirac when $\omega$ is presymplectic and $\pi$ Poisson. 
    Conversely, if $L$ is an almost Dirac structure so that $L\cap (\{0\}\oplus T^{\ast}P) = \{0\}$, then $L$ is the graph of a 2-form and $K_L=\textup{Ker}(\omega)$.  Respectively, if $K_L = \{0\}$, then $L$ is the graph of a bivector field. For details see e.g., \cite{DiracIntro2011}. 
\end{remark}

\begin{remark}\label{R:Bracket-Bivector}
    We recall here that there is a one-to-one correspondence between almost Poisson brackets $\{\cdot, \cdot\}$ on $C^\infty(P)$ and bivector fields $\pi$ on $P$ given, at each $f,g\in C^\infty(P)$, by $\{f,g\} = \pi(df,dg)$.  The Jacobi identity of $\{\cdot, \cdot\}$ is equivalent to ask that $[\pi,\pi]=0$, where $[\cdot, \cdot]$ is the Schouten bracket, see e.g. \cite{Marsden_1992}.
\end{remark}
    
Following \cite{LuisN2012,BursztynCrainic2005}, there is a one-to-one correspondence between regular almost Dirac structures $L$ on $P$ and pairs $(F,\omega_{F})$, where $F=pr_T(L)$ is a regular distribution on $P$ and $\omega_F$ is a 2-section on $F$ so that 
\begin{equation}\label{Eq:L_(F,w)}  
    L = \left\{(X,\alpha)\in \mathbb{T}P\,:\, X\in F,\, {\bf i}_{X}\omega_{F} = -\left.\alpha\right|_{F}\right\}.
\end{equation} 
We may also associate to $L$ a pair $(F, \omega)$ where $\omega$ is a 2-form on $P$ such that $\omega|_F = \omega_F$ (but, in this case, $\omega$ is not unique).  
If the 2-section $\omega_F$ is nondegenerate then $L$ is the graph of a bivector field $\pi$ defined by $\pi^{\sharp}(\alpha) = -X$ if and only if ${\bf i}_{X}\omega_{F} = \alpha|_{F}$ for $X\in \Gamma(F),\,\alpha\in T^*P$ and $F$ is the distribution generated by Hamiltonian vector fields.  

\subsection{Backward and Forward images of Dirac structures} \label{Ss:DiracBF}

In this section we formalize the notions of {\it pushforward} and {\it pullback} of almost Dirac structures, having in mind that these structures simultaneously encode almost Poisson structures and 2-forms; here we follow \cite{DiracIntro2011}.
In the case of nonholonomic systems we are particular interested in pulling back (almost) Poisson brackets.

\medskip

Let $P_1, P_2$ smooth manifolds and $\varphi : P_{1}\rightarrow P_{2}$ be a smooth map. 
First recall that the pull back of the bundle $\tau: TP_2\to P_2$ by $\varphi$ is the bundle $P_{1}\times_{P_{2}}TP_2$ on $P_1$. 
Analogously, the \textit{pullback of $\mathbb{T}P_{2}$ by $\varphi$} is defined as the vector bundle $P_{1}\times_{P_{2}}\mathbb{T}P_{2}$ on $P_{1}$ whose total space is
\begin{equation*}
	P_{1}\times_{P_{2}}\mathbb{T}P_{2} = \{(m,(X,\alpha))\in P_{1}\times \mathbb{T}P_{2}\,:\, \varphi(m) = \tau_{2}(X,\alpha)\},
\end{equation*}
where $\tau_{2} : \mathbb{T}P_{2}\rightarrow P_{2}$ is the bundle projection. 

Let us denote by  $T\varphi : TP_{1}\rightarrow P_{1}\times_{P_{2}}TP_{2}$ the tangent map (covering the identity). 
Following \cite{DiracIntro2011}, if $L_2 \subset \mathbb{T}P_2$ is an almost Dirac structure, then the \textit{backward image of $L_{2}$ by $\varphi$}, denoted by $\varphi^{\ast}L_2$, is the Lagrangian distribution on $\mathbb{T}P_1$ given, at each $x\in P_{1}$, by
\begin{equation*}
	({\varphi}^{\ast}L_{2})_{x} := \{(X, \varphi^{\ast}\beta)\in \mathbb{T}_{x}P_{1}\,:\, (T\varphi(X),\beta)\in (L_{2})_{\varphi(x)}\}.
\end{equation*}
If $L_1 \subset \mathbb{T}P_1$ is an almost Dirac structure, the \textit{forward image of $L_{1}$ by $\varphi$}, denoted by $\varphi_{\ast}L_1$, is the Lagrangian distribution on $P_{1}\times_{P_{2}}\mathbb{T}P_2$ given, at each $x\in P_{1}$, by
\begin{equation*}
	({\varphi}_{\ast}L_{1})_{x} := \{(T\varphi(Y), \alpha) \in \mathbb{T}_{\varphi(x)}P_2\,:\, (Y,\varphi^{\ast}\alpha)\in (L_{1})_{x}\}.
\end{equation*}
Moreover, $L_{1}$ is said to be \textit{$\varphi$-invariant} if $({\varphi}_{\ast}L_{1})_{x} = ({\varphi}_{\ast}L_{1})_{y}$ for all $x,y\in P_{1}$ such that $\varphi(x) = \varphi(y)$ (when $\varphi(x) = \varphi(y)$, $({\varphi}_{\ast}L)_{x}$ and $({\varphi}_{\ast}L)_{y}$ belong to $\mathbb{T}_{\varphi(x)}P_2 = \mathbb{T}_{\varphi(y)}P_2$).

In order to guarantee the smoothness of ${\varphi}^{\ast}L_{2}$ and ${\varphi}_{\ast}L_{1}$  and define them as almost Dirac structures on $P_{1}$ and $P_{2}$ respectively, we need to check the so-called {\it clean intersection condition} \cite[Sec.~5]{DiracBatesWeinstein1997} (see also \cite{DiracIntro2011})  for both structures and furthermore the $\varphi$-invariance of $L_1$. 

\begin{proposition}[\textit{Clean-intersection condition, \cite{DiracBatesWeinstein1997}}]\label{Prop:DiracCleanCond}
    Let $L_{i}\subset \mathbb{T}P_{i}$, for $i = 1,2$, be a (almost) Dirac structure and $\varphi : P_{1}\rightarrow P_{2}$ a smooth map with $(T\varphi)^{\ast} : P_{1}\times_{P_{2}}T^{\ast}P_{2}\rightarrow T^{\ast}P_{1}$ the dual map of the tangent map $T\varphi : TP_{1}\rightarrow P_{1}\times_{P_{2}}TP_{2}$. 
    \begin{enumerate}
            \item[$(i)$] \label{Prop:DiracCleanCondBW} If $(\{0\}\oplus \textup{Ker}((T\varphi)^{\ast})) \cap (P_{1}\times_{P_{2}}L_{2})$ has constant rank, then the backward image ${\varphi}^{\ast}L_{2} \subset \mathbb{T}P_{1}$ defines a (almost) Dirac structure on $P_{1}$. 
        
            \item[$(ii)$] \label{Prop:DiracCleanCondFW} If $\varphi : P_{1}\rightarrow P_{2}$ is a surjective submersion such that $\textup{Ker}(T\varphi)\cap K_{L_{1}}$ has constant rank and $L_1$ is $\varphi$-invariant, then the forward image ${\varphi}_{\ast}L_{1}$ is a (almost) Dirac structure on $P_{2}$. 
    \end{enumerate}
\end{proposition}

\begin{remark}\label{R:Dirac:BF-Examples}
    From the assumptions of Proposition \ref{Prop:DiracCleanCond} and using the Remark \ref{R:Dirac:BasicExamples}, it is straightforward to see that: 
    \begin{enumerate}
        \item[$(i)$] Let $\iota : M\to P$ be the inclusion map of a submanifold $M$ of $P$ and $L\subset \mathbb{T}P$ the graph of a 2-form $\omega$ on $P$, then ${\iota}^{\ast}L$ is the almost Dirac structure on $M$ given by the graph of the 2-form $\iota^{\ast}\omega$. However, if $L$ is the graph of a bivector field $\pi$ with nontrivial kernel, then the clean intersection condition is equivalent to have that $TM^\circ \cap \textup{Ker}(\pi^{\sharp})$ has constant rank, where $TM^{\circ}$ denotes the annihilator of $TM$.  In this case ${\iota}^{\ast}L$ is an almost Dirac structure that might not be associated to a bivector field nor a 2-form.  
        
        \item[$(ii)$] Let $\varphi : P_{1}\rightarrow P_{2}$ be a surjective submersion and $L_{1}\subset \mathbb{T}P_1$ be the graph of a $\varphi$-invariant bivector field $\pi_{1}$ on $P_{1}$, then ${\varphi}_{\ast}L_{1}$ represents the bivector field $\pi_2$ on $P_2$ which is the projection of $\pi_1$. However, if $L_{1}$ is the graph of a $\varphi$-invariant 2-form $\omega$ with nontrivial kernel, then the clean intersection condition means  that $\textup{Ker}(T\varphi) \cap \textup{Ker}(\omega)$ has constant rank. In this case, ${\varphi}_{\ast}L_{1}$ is an almost Dirac structure that might not be associated with a bivector nor a 2-form. 
    \end{enumerate}
\end{remark}

Let $L_{i}\subset \mathbb{T}P_{i}$, for $i = 1,2$, be a (almost) Dirac structure on the manifold $P_{i}$ and a smooth map $\varphi : P_{1}\rightarrow P_{2}$ such that $L_{1}$ is $\varphi$-invariant. The map $\varphi : P_{1}\rightarrow P_{2}$ is a {\it backward-Dirac map} (or a \textit{$b$-Dirac map}) if $L_{1} = {\varphi}^{\ast}L_{2}$, and it is a {\it forward-Dirac map} (or a \textit{$f$-Dirac map}) if ${\varphi}_{\ast}L_{1} = L_{2}$, see \cite{Bursztyn-Radko-2003}. Observe here that the forward and backward images are not inverse procedures, i.e.,  $\varphi^{\ast}(\varphi_{\ast}L_1)$ (resp. $\varphi_{\ast}(\varphi^{\ast}L_2)$) might not be $L_1$ (resp. $L_2$).  Furthermore, if $\varphi : P\rightarrow P$ is a diffeomorphism  and $L$ is an almost Dirac structure on $P$, then $\varphi$ is a $f$-Dirac map if and only if it is a $b$-Dirac map. In this case $\varphi$ is called a \textit{Dirac diffeomorphism} defining a diffeomorphism $\Phi := (T\varphi,(T\varphi^{-1})^{\ast}) : \mathbb{T}P\rightarrow \mathbb{T}P$ so that $\Phi(L) = L$.

We end this section with a Proposition that clarifies the relation between regular almost Dirac structures through backward and forward images. Let $L_{i}\subset \mathbb{T}P_{i}$, for $i = 1,2$, be regular (almost) Dirac structures on $P_{i}$ and let $\varphi : P_{1}\to P_{2}$ be a $f$-Dirac map with $L_1$ $\varphi$-invariant. Following \eqref{Eq:L_(F,w)} we identify each structure $L_{i}$ with the pair $(F_{i},\omega_{i})$, where $\omega_i$ is a 2-form on $P_i$ such that $\omega_i|_{F_i} = \omega_{F_i}$. Inspired in \cite{Bates1993}, we define the distribution $\mathcal{U}_\varphi$ on $P_1$ given by 
\begin{equation*}
    \mathcal{U}_{\varphi} := \textup{span}\left\{X\in F_{1} \,:\, (X,\alpha)\in L_{1},\, T\varphi(X)\in F_{2},\,\alpha|_{\textup{Ker}(T\varphi)} = 0 \right\}. 
\end{equation*}
In particular, $T\varphi(\mathcal{U}_{\varphi}) = F_{2}\subset T\varphi(F_{1})$ since $\varphi_*(L_1) = L_2$. 

\begin{remark}\label{R:UDist_of_Bates}
    When $\varphi:P\to P/G$ is an orbit projection with respect to a free and proper $G$-action and $L\subset \mathbb{T}P$ is the graph of a $G$-invariant bivector field $\pi$ with associated pair $(F,\omega_{F})$, the distribution $\mathcal{U}_\varphi$ is just $\mathcal{U}_\varphi := \{X\in F\,:\, {\bf i}_{X}\omega_{F}|_{\mathcal{S}} \equiv 0\}$ where $\mathcal{S}:= F\cap\textup{Ker}(T\varphi)$.  For this special case, this distribution was defined in \cite{Bates1993}, in the context of nonholonomic reduction independently of Dirac structures. 
\end{remark}

\begin{proposition}\label{Prop:DiracFormsRltm}
    For each $i=1,2,3$, let $L_{i}\subset \mathbb{T}P_{i}$ be a regular (almost) Dirac structure on the manifold $P_{i}$ described by the pair $(F_{i},\omega_i)$, where $\omega_i$ is a 2-form on $P_i$ such that $\omega_i|_{F_i} = \omega_{F_i}$. Let $\varrho : P_{1}\to P_{3}$ be a $b$-Dirac map and $\varphi : P_{1}\to P_{2}$ be a $f$-Dirac map with $L_1$ $\varphi$-invariant. If  $F_{1}\cap \textup{Ker}(T\varphi) \subset K_{L_{1}}$ for $K_{L_{1}}$ the null distribution of $L_{1}$, then
    \begin{equation}\label{Eq:DiracFormsRltm}
        \varphi^*\omega_2|_{\mathcal{U}_{\varphi}} = \varrho^*\omega_3|_{\mathcal{U}_{\varphi}}.
    \end{equation}
\end{proposition} 
\begin{proof}
    By assumption $L_{1} = \varrho^{\ast}L_{3}$ and $L_{2} = \varphi_{\ast}L_{1}$, then $T\varrho(\mathcal{U}_{\varphi}) \subset F_{3}$ and $T\varphi(\mathcal{U}_{\varphi}) = F_{2}$, which means that both hand sides of (\ref{Eq:DiracFormsRltm}) are well-defined. Moreover, if $X\in \mathcal{U}_{\varphi}$, then there exist $\alpha\in T^{\ast}P_{2}$ and $\beta\in T^{\ast}P_{1}$ such that: $(X,\beta)\in L_{1}$, $(T\varphi(X),\alpha)\in L_{2}$ and $(X+Z,\varphi^{\ast}\alpha)\in L_{1}$ for some $Z\in F_1\cap\textup{Ker}(T\varphi)$. Then $Z\in K_{L_1}$ and hence $(Z,0)\in L_1$. Thus $(X, \varphi^*\alpha)\in L_1$ and therefore we conclude that $\beta|_{F_1} = \varphi^{\ast}\alpha|_{F_1}$. On the other hand, there exists $\gamma\in T^{\ast}P_{3}$ such that $\varrho^{\ast}\gamma|_{F_{1}} = \varphi^{\ast}\alpha|_{F_{1}}$ and $(T\varrho(X),\gamma)\in L_{3}$. Finally, for $Y\in \mathcal{U}_{\varphi}$ we have: 
    \begin{align*}
        \varphi^*\omega_2(X,Y) & = \omega_{F_{2}}(T\varphi(X),T\varphi(Y)) = -\alpha(T\varphi(Y)) = -(\varphi^{\ast}\alpha)(Y) = -(\varrho^{\ast}\gamma)(Y) = \\
        & = -\gamma(T\varrho(Y)) = \omega_{F_{3}}(T\varrho(X),T\varrho(Y)) = \varrho^*\omega_3(X,Y).
    \end{align*}
\end{proof}
\begin{corollary}\label{cr:pullbackforms}
    If $L_{2}$ and $L_{3}$ are given, respectively, by the graph of 2-forms $\omega_{2}$ and $\omega_{3}$, then $\varphi^{\ast}\omega_{2} = \varrho^{\ast}\omega_{3}$ if and only if $\textup{Ker}(T\varphi) \subseteq K_{L_{1}}$. 
\end{corollary}

\subsection{Momentum map reduction} \label{Ss:MomMapDirac}

Let $P$ be a manifold with a Lie group $G$ acting freely and properly on $P$,  $\psi : G \times P \to P$. Let $L\subset \mathbb{T}P$ be a regular (almost) Dirac structure so that $L$ is invariant, i.e.,  for each $g\in G$, ${\psi}_{g} : P \to P$ is a Dirac diffeomorphism.  Let $\mathcal{V}$ be the vertical distribution on $P$ defined, at each $x\in P$, by $\mathcal{V}_x = T_x(Orb_G(x))$. During this article we will assume that $\mathcal{V}$ does not necessarily belong  to the characteristic distribution of $L$.  
Instead,  let us define the distribution $\mathcal{S}$ on $P$ by 
\begin{equation}\label{Eq:DefS}
    \mathcal{S} := pr_T(L) \cap \mathcal{V} \subset TP,
\end{equation}
which is $G$-invariant but, in general, not integrable.  
Inspired in the nonholonomic case, we assume that 
\begin{enumerate}
    \item[(A1)] $\mathcal{S}$ has constant rank and it admits a global basis of $G$-invariant vector fields on $P$; 
    \item[(A2)] $K_L \cap \mathcal{S} = \{0\}$, where $K_{L}$ is the null distribution of $L$.
\end{enumerate}

\begin{remark}
    Assumption (A2) ensures that $\rho : P\to P/G$ is a $f$-Dirac map. Furthermore, in the literature, condition (A2) is called the {\it non-degeneracy condition} and guarantees that $\rho : P\to P/G$ is a \textit{Dirac realization}. For details see, e.g. \cite{BursztynCrainic2005}. 
\end{remark}

As a consequence of assumption (A1), the distribution $\mathcal{S}$ defines a subbundle $\mathfrak{g}_{\mathcal{S}} \to P$ of the action bundle $P\times\mathfrak{g}\to P$ with fiber given, on each $p\in P$, by
\begin{equation}\label{Eq:gs}
    \mathfrak{g}_{\mathcal{S}}|_{x} :=\left\{(x,\xi)\in P\times\mathfrak{g} \ : \ \xi_{\subP}(x) \in \mathcal{S}_{x}\right\}, 
\end{equation}
where $\xi_\subP(x)$ denotes the infinitesimal generator of $\xi\in \mathfrak{g}$ at $x\in P$. A section $\xi$ of the bundle $\g_{\mathcal{S}} \to P$, denoted by $\Gamma(\g_{\mathcal{S}})$, can be thought as a $\g$-valued function on $P$, where $\xi(x) = (x,\xi) \in {{\mathfrak{g}}_{\mathcal{S}}}|_x$ and $\xi_\subP(x) = (\xi(x))_\subP(x)$. Observe that, $\textup{rank}(\mathfrak{g}_{\mathcal{S}}) = \textup{rank}(\mathcal{S}) =: k$. Furthermore, $\mathfrak{g}_{\mathcal{S}}$ admits a global basis of $Ad$-invariant sections
$$
\mathfrak{B}=\{\xi_1,\dots,\xi_k\}, 
$$
so that $\{\mathcal{Y}_1= (\xi_1)_\subP ,\dots,\mathcal{Y}_k= (\xi_k)_\subP\}$ is a global basis of $\Gamma(\mathcal{S})$ given by $G$-invariant vector fields on $P$.  

\begin{definition}\label{Def:MomentumBundle}
    A bundle map $J:P\to \mathfrak{g}_{\mathcal{S}}^{\ast}$, that covers the identity, is said to be a {\it momentum bundle map of $L$ associated to the basis $\mathfrak{B}=\{\xi_1,\dots,\xi_k\}$} if  the induced functions $J_i\in C^\infty(P)$ given by $J_{i}(x) := \left\langle J(x), \xi_{i}(x) \right\rangle$ satisfy that
    \begin{equation}\label{Eq:MomentumCond}
        \left(-{(\xi_{i})}_{\subP}, dJ_{i}\right)\in L\quad \mbox{for each}\quad i = 1,\dots, k. 
    \end{equation}
\end{definition}

\begin{example}[\textit{nonholonomic particle}]\label{ex:NHParticle}
    Consider \  the\  manifold \ $P = \mathbb{R}^{5}$\  with  \ coordinates \ $p=(x,y,z,p_{x},p_{y})$ and the bivector field $\pi = (\partial_{x}+y\partial_{z})\wedge \partial_{p_{x}}+\partial_{y}\wedge \partial_{p_{y}}+yf(y)p_{x}\partial_{p_{x}}\wedge\partial_{p_{y}}$, where $f(y) = (1+y^{2})^{-1}$. The action of $G = \mathbb{R}^{2}$ by translation on the first and third coordinates leaves $\pi$ invariant. Considering $L=\textup{graph}(\pi^\sharp)$, condition (A1) follows from the fact that $\mathcal{S} = span\{Y=\partial_{x}+y\partial_{z}\}$, and (A2) since $K_L=0$. 
    Observe that $\mathfrak{g}_{\mathcal{S}}|_p = \textup{span}\{(1,y)\in \mathfrak{g}\}$ since $(1,y)_P = Y$.  The bundle map $J : P\to \mathfrak{g}_{\mathcal{S}}^{\ast}$ defined at each $p\in \mathbb{R}^5$, by  $\langle J(p), (1,y)\rangle = p_x$ is not a momentum bundle map for $\pi$ associated to the basis $\mathfrak{B}=\{(1,y)\}$  since $\pi^\sharp(dp_x) \neq -Y$. However, it is a momentum bundle map for the basis $\mathfrak{B}=\{\xi := \sqrt{f(y)}(1,y)\}$ since $J_\xi =  \sqrt{f}p_x$ and $\pi^\sharp(dJ_\xi) = -\sqrt{f} Y = -\xi_\subP$.    
    In Sec.~\ref{S:Examples} we will see more examples of momentum maps associated to a basis.  
\end{example}

From now on, let $P$ be a manifold with a free and proper $G$-action and let $L$ a $G$-invariant regular (almost) Dirac structure on $P$  so that  conditions (A1) and (A2) are satisfied. Let us denote by  $\mathfrak{B} = \{\xi_{1},\dots,\xi_{k}\}$ a global basis of $Ad$-invariant sections of the bundle $\g_{\mathcal{S}}\to P$.  Let us assume that $J:P\to \mathfrak{g}_{\mathcal{S}}^{\ast}$ is a momentum bundle map of $L$ associated to $\mathfrak{B}$.
Condition (A2) guarantees that the independence of the sections $\{\xi_1,\dots,\xi_k\}$ implies that the corresponding functions $\{J_1,\dots,J_k\}$ are functionally independent. Moreover, 

\begin{lemma}\label{L:MomMap-Ad}
    The functions $J_{i}$ are $G$-invariant on $P$ for all $i=1,\dots,k$, if and only if the bundle map $J:P\to \g_{\mathcal{S}}^{\ast}$ is $Ad^{\ast}$-equivariant, that is, for $x\in P$ and $g\in G$, $J(\psi_g(x)) = Ad_{g^{-1}}^{\ast}J(x)$.
\end{lemma}

\begin{proof} 
    For $x\in P$ observe that $J_i(x) = \langle J(x),\xi_i(x)\rangle$.  Using that the basis $\mathfrak{B}$ is $Ad$-invariant, we also obtain, for $g\in G$, that
    $$
        J_i(\psi_g(x)) = \langle J(\psi_g(x)), \xi(\psi_g(x))\rangle = \langle J(\psi_g(x)),  Ad_g\,\xi_i(x)\rangle =\langle Ad_{g}^*J(\psi_g(x)), \xi_i(x)\rangle. 
    $$
    Then it is straightforward to see that if $J:\M\to \g_{\mathcal{S}}^*$ is $Ad^*$-invariant then the $J_i$ are $G$-invariant. The converse is proven using the linearity of the map $J(m):\g_{\mathcal S}|_m \to \mathbb{R}$.
\end{proof}

The $Ad$-invariant basis $\mathfrak{B}$ induces a dual $Ad^{\ast}$-invariant basis $\mathfrak{B}^{\ast}=\{\mu^{1},\dots,\mu^{k}\}$ of sections of $\mathfrak{g}_{\mathcal{S}}^{\ast}$, i.e., each $\mu^{i} \in \mathfrak{B}^{\ast}$ can be seen as $\mathfrak{g}^{\ast}_{\mathcal{S}}$-valued function on $P$ given, at each $x\in P$, by 
\begin{equation}\label{Eq:DefMu}
    \langle \mu^{i},\xi_{j}\rangle (x) = \langle \mu^{i}(x), \xi_{j}(x)\rangle = \delta_{ij}.
\end{equation}

For $\mu = c_{1}\mu^{1} + \dots + c_{k}\mu^{k}$ with $c_{i}$ constants in $\mathbb{R}$, we define the level set $J^{-1}(\mu)\subset P$ by
\begin{equation*}
    J^{-1}(\mu) := \left\{x\in P\, : \,J(x) = \mu(x)\right\}. 
\end{equation*}

\begin{proposition}\label{Prop:Jsubmfld} 
    If the momentum bundle map $J:P\to \g^*_{\mathcal{S}}$ is $Ad^*$-equivariant, then, for each $\mu = c_{1}\mu^{1} + \dots + c_{k}\mu^{k}$, $J^{-1}(\mu)$ is a $G$-invariant submanifold of codimension $k$ on $P$ and in particular, 
    \begin{equation}\label{Eq:LevelSetIdent}
        J^{-1}(\mu) = \bigcap_{i=1}^{k}J_{i}^{-1}(c_{i}).
    \end{equation}
    Moreover, the union of the connected components of the manifolds $J^{-1}(\mu)$ for $\mu \in span_{\mathbb{R}}\mathfrak{B}^{\ast}$ defines  a foliation of $P$.
\end{proposition}

\begin{proof} 
    It is straightforward to see that, for $x\in J^{-1}(\mu)$ means that $J(x)= c_i\mu^i(x)$ which is equivalent to $J_{i}(x) = \langle J(x), \xi_{i}(x)\rangle = c_{i}$ for all $i=1,\dots,k$, then (\ref{Eq:LevelSetIdent}) holds. Let $c := (c_{1},\dots,c_{k}) \in \mathbb{R}^{k}$ and $\mathcal{J} := (J_{1},\dots,J_{k}) : P\rightarrow \mathbb{R}^{k}$, since $\{J_{1},\dots,J_{k}\}$ are functionally independent, then $\mathcal{J}^{-1}(c) = \cap_{i}J_{i}^{-1}(c_{i})$ is a submanifold of codimension $k$ on $P$ that foliates $P$ through the connected components of $\mathcal{J}^{-1}(c)$ for all $c\in \mathbb{R}^{k}$. The $G$-invariance of $J^{-1}(\mu)$ follows directly from the $G$-invariance of $\{J_1,\dots,J_k\}$ by Lemma \ref{L:MomMap-Ad}. 
\end{proof}
 
Since $J^{-1}(\mu)$ is an invariant submanifold on $P$, we denote by $\mathcal{V}_{\mu}$ the vertical space of the induced $G$-action on $J^{-1}(\mu)$ and it is straightforward to see, for each $x\in J^{-1}(\mu)$, that ${\mathcal{V}_{\mu}}|_x = \mathcal{V}|_{x}$. We also denote by  $\rho_\mu: J^{-1}(\mu) \to J^{-1}(\mu)/G$ the orbit projection, which defines a principal bundle. 

Next, we study the (almost) Dirac structure obtained as the backward image of $L$ to the level set $J^{-1}(\mu)$, and subsequently its reduction to the orbit space $J^{-1}(\mu)/G$, which will be again a (almost) Dirac structure.

\begin{theorem}\label{T:momentum-red} 
    Let $P$ be a manifold with a free and proper $G$-action and let $L$ be a $G$-invariant regular (almost) Dirac structure on $P$ satisfying conditions \textup{(A1)} and \textup{(A2)}. If $J: P\to \mathfrak{g}_{\mathcal{S}}^{\ast}$ is an $Ad^*$-equivariant momentum bundle map of $L$ associated to a global basis of $Ad$-invariant sections $\mathfrak{B} = \{\xi_{1},\dots,\xi_{k}\}$ of $\mathfrak{g}_{\mathcal{S}}$, then, for $\mu= c_i\mu^i$ with $\mu^i\in \mathcal{B}^*$ and $c_{i}$ constants in $\mathbb{R}$, we have: 
    \begin{enumerate}
        \item[$(i)$] The backward image $L_{\mu}:={\iota_{\mu}^{\ast}}L$ of $L$ by the inclusion map $\iota_{\mu} : J^{-1}(\mu) \to P$ is a $G$-invariant (almost) Dirac structure on $J^{-1}(\mu)$. 
        
        \item[$(ii)$] The forward image $L_{\emph{\red}}^{\mu} :=(\rho_{\mu})_{\ast}L_{\mu}$ of $L_{\mu}$ by the orbit projection $\rho_{\mu}: J^{-1}(\mu) \to J^{-1}(\mu)/G$ is a (almost) Dirac structure on $J^{-1}(\mu)/G$: 
        \begin{equation}\label{Diag:DiracMomRed}
            \xymatrix{ (J^{-1}(\mu), L_{\mu}:={\iota_{\mu}^{\ast}}L) \ar[r]^{\qquad\iota_{\mu}} \ar[d]^{\rho_{\mu}} & (P,L)\\
        (J^{-1}(\mu)/G, L_{\emph{\red}}^{\mu} :=(\rho_{\mu})_{\ast}L_{\mu})}
        \end{equation}
        
        \item[$(iii)$] If $(F, \omega_L)$ and $(F_{\emph\red}^\mu, \omega_{\emph\red}^\mu)$ are the associated pairs of $L$ and $L_{\emph\red}^\mu$ respectively, we obtain that
        $$
            \iota_\mu^*\omega_L\,|_{\mathcal{U}_{\rho\!_\mu}} = \rho_\mu^* \omega_{\emph\red}^\mu\,|_{\mathcal{U}_{\rho\!_\mu}},
        $$
        where ${\mathcal{U}_{\rho_\mu}}= \{X\in F\cap TJ^{-1}(\mu) \ : \ (X,\alpha)\in L_{\mu}, \ T\rho_\mu(X)\in F_{\emph\red}^\mu, \ \alpha|_{\mathcal{V}_{\mu}} = 0\}$.
    \end{enumerate}
\end{theorem}
\begin{proof}
    $(i)$ In order to see that $L_\mu$ is an almost Dirac structure, we need to check the clean intersection condition (Prop. \ref{Prop:DiracCleanCond}).  In fact, observe from \eqref{Eq:LevelSetIdent}, that $\textup{Ker}((T\iota_{\mu})^{\ast}) = (T(J^{-1}(\mu)))^\circ = span\{dJ_{1},\dots,dJ_{k}\}$. Hence $(\{0\}\oplus \textup{Ker}((T\iota_{\mu})^{\ast})) \cap (J^{-1}(\mu)\times_{P}L) = \{0\}$ follows from the fact that $J:P\to\g_{\mathcal S}^*$ is a momentum bundle map and therefore the clean-intersection condition is verified.

    Recall that the $G$-invariance of $L$ implies that, for $g\in G$, $\Psi_g(L_{(p)}) = L_{\psi_g(p)}$ where $\Psi_g: \mathbb{T}P \to \mathbb{T}P$ is given by $\Psi_g = (T\psi_g, \psi^*_{g^{-1}})$.  
    To see that $L_\mu$ is $G$-invariant, first we denote by $\psi^\mu_g:J^{-1}(\mu)\to J^{-1}(\mu)$ the restriction of the $G$-action to $J^{-1}(\mu)$ and we will show that  $\Psi^\mu_g((L_\mu)_{(x)} = (L_\mu)_{\psi_g(x)}$. In fact, consider $({X}_x,  (\iota_\mu^*\alpha)_x)\in (L_\mu)_x$ where  $(T_{x}\iota_\mu({X}_x), \alpha_{\iota_\mu(x)})\in L_{\iota_\mu(x)}$.  Since $L$ is $G$-invariant, then $(T_{\iota_\mu(x)}\psi_g\circ T_{x}\iota_\mu ({X}_x), \psi_{g^{-1}}^*\alpha_{\iota_\mu(x)}) = (T_{\psi^{\mu}_{g}(x)}\iota_\mu \circ T_{x}\psi^\mu_g ({X_x}), \psi^*_{g^{-1}}\alpha_{\iota_\mu(x)}) \in L_{\psi_g \circ \iota_\mu(x)}$. Thus $(T_{x}\psi^\mu_g ({X_x}), \iota_{\mu}^{\ast}\psi^*_{g^{-1}}\alpha_{\iota_\mu(x)}) = (T_{x}\psi^{\mu}_{g} ({X_x}), (\psi^{\mu}_{g^{-1}})^{\ast}(\iota_{\mu}^{\ast}\alpha)_{x})\in (L_\mu)_{\psi^{\mu}_{g}(x)}$.  

    
    $(ii)$ The $\rho_{\mu}$-invariance of $L_\mu$ follows from its $G$-invariance, and then it remains to prove the clean-intersection condition for $L_{\mu}$ to show that the $L_{\red}^{\mu} := ({\rho_{\mu}})_{\ast}L_{\mu}$ is a (almost) Dirac structure on $J^{-1}(\mu)/G$. Since the action is free and proper on $J^{-1}(\mu)$, assumption (A1) implies that $\mathcal{S}_{\mu} = \mathcal{S}|_{J^{-1}(\mu)}$ has constant rank. Finally, since $(-(\xi_{i})_{\subP},dJ_{i})\in L$ for all $i = 1,\dots, k$, then $(-(\xi_{i})_{J^{-1}(\mu)},0)\in L_{\mu}$, which implies $\mathcal{S}_{\mu}\subset K_{L_{\mu}}$ and hence $\textup{Ker}(T\rho_{\mu})\cap K_{L_{\mu}} = \mathcal{S}_{\mu}$ has constant rank as well. 
 
    $(iii)$ As we have seen, $\mathcal{S}_{\mu}\subset K_{L_{\mu}}$, thus the desired result follows directly from Prop.~\ref{Prop:DiracFormsRltm}.
\end{proof}

When we say that $L\subset \mathbb{T}P$ is an (almost) Dirac structure associated to an (almost) Poisson bracket $\{\cdot , \cdot\}_\subP$ means that $L = \textup{graph}(\pi_\subP^\sharp)$, see Remark \ref{R:Dirac:BasicExamples}.

\begin{corollary}\label{C:MomReduction}
    Let $L$ be a (almost) Dirac structure on $P$ as in Theorem \ref{T:momentum-red}. 
    \begin{enumerate}
        \item[$(i)$] If $L$ is a (almost) Poisson structure $\{\cdot , \cdot\}_\subP$, then $L^\mu_{\emph{\red}}$ is a (almost) Poisson structure $\{\cdot, \cdot\}_\mu$ on $J^{-1}(\mu)/G$ satisfying 
        \begin{equation} \label{Eq:Marsden-Ratiu-Red}
            \{F,G\}_{\subP}\circ \iota_{\mu} = \{f,g\}_{{\mu}}\circ \rho_{\mu}
        \end{equation}
        for $f,g\in C^\infty(J^{-1}(\mu)/G)$ and $F,G\in C^\infty(P)$ ($G$-invariant) extensions of the functions $\rho_\mu^*f$ and $\rho_\mu^*g$ to $P$. 
        
        \item[$(ii)$] If $L$ corresponds to a (almost) symplectic form $\Omega$ on $P$, then $L^\mu_{\emph{\red}}$ is a (almost) symplectic form $\omega$ on $J^{-1}(\mu)/G$ such that $\iota_\mu^*\Omega=\rho_\mu^*\omega$. 
    \end{enumerate}
\end{corollary}

\begin{proof}
$(i)$ Let $X\in K_{(L_\red^\mu)}$, then there exists $Y\in K_\mu \subset T(J^{-1}(\mu))$ such that $T\rho_\mu(Y) =X$.  Since $L$ is the graph of a bivector field $\pi_\subP$, there exists a $\alpha\in T^*P$ such that $\pi^\sharp_\subP(\alpha) = T\iota_\mu(Y)$ and $\iota_\mu^*(\alpha) = 0$.  Therefore, $\alpha\in (T(J^{-1}(\mu)))^{\circ} = \textup{span}\{dJ_{1},\dots,dJ_{k}\}$ and hence $T\iota_\mu(Y)\in \mathcal{S}$. Since $T\rho_\mu(Y) = 0$ we conclude that $X=0$ implying that $L_\red^\mu$ is the graph of bivector field. Moreover, (\ref{Eq:Marsden-Ratiu-Red}) follows from the fact that $(X,df)\in L_\red^\mu$ if and only if $(T\iota_\mu Y, dF)\in L$ for $Y\in TJ^{-1}(\mu)$ such that $T\rho_\mu(Y) = X$ and $F\in C^\infty(P)$ is an extension of the function $\rho_\mu^*f$. 

$(ii)$ Note that $L_\mu=\textup{graph}(\iota_\mu^*\Omega)$ and, moreover, $K_\mu = \textup{Ker}(\iota_\mu^*\Omega) = \mathcal{V}_\mu$.  Then it is straightforward to see that $L^\mu_{\red}$ is given by the (almost) symplectic form $\omega$ on $J^{-1}(\mu)/G$ such that $\iota_\mu^*\Omega=\rho_\mu^*\omega$. 
\end{proof}

\begin{remark}
    Theorem \ref{T:momentum-red} (or more precisely Corollary \ref{C:MomReduction}) recovers (and generalizes) the Poisson reduction via a momentum map studied in \cite{PoissonMarsdenRatiu} by Marsden and Ratiu, see Sec.~\ref{S:Intro}. 
\end{remark}

\begin{example}[\textit{nonholonomic particle}]
    In order to complete Example \ref{ex:NHParticle}, consider the dual basis $\mathfrak{B}^* =\{\mu = \frac{1}{\sqrt{f(y)}}(1,0)\}$ of $\g_{\mathcal{S}}^*\to P$ and the momentum map $J : P\to \mathfrak{g}_{\mathcal{S}}^{\ast}$ defined by $J(p) = \sqrt{f} p_x. \mu$.  Then, for each $c\in \mathbb{R}$,  $ J^{-1}(c\mu) = J_{\xi}^{-1}(c) = \{p\in \mathbb{R}^5 \,:\, p_{x} = c\sqrt{1+y^{2}}\}$ is a $G$-invariant submanifold. 
    Then, $L_\mu = \textup{span}\{(-\sqrt{f}Y,0),(\partial_y,dp_y),(\partial_{p_y},dy), (0,dz-ydx)\}$ is an almost Dirac structure which is not the graph of a bivector field or a 2-form.
    Since our starting geometric structure is an almost Poisson bracket $\pi$ on $\mathbb{R}^5$, following Corollary \ref{C:MomReduction}, the almost Dirac structure $L_\red^{c\mu}$ on $J^{-1}(c\mu)/\mathbb{R}^2$, with coordinates $(y,p_y)$,  is associated to the bivector field $\pi_{c} = \partial_{y}\wedge \partial_{p_{y}}$ (that, in this case, is nondegenerate). 
\end{example}

\section{Nonholonomic systems} \label{S:NHSystems}

In this section, we study the reduction of nonholonomic systems via the nonholonomic momentum bundle map, which was originally defined in \cite{Bloch1996}. 
First we will recall the geometry underlying nonholonomic systems given by an almost Poisson bracket, following \cite{ Bates1993,Bloch1996,Bloch2003,RecCortes2002,Cushman2010,IbLMM99}. If the system is invariant by the action of a Lie group, the usual scenario for nonholonomic systems is that the vertical distribution does not necessary belong to the characteristic distribution of such bracket.  In this sense, the distribution $\mathcal{S}$ defined in \eqref{Eq:DefS} will play a fundamental role.

A nonholonomic mechanical system on a smooth manifold $Q$ is a classical mechanical system defined by a Lagrangian $L: TQ \rightarrow \mathbb{R}$ and a (constant rank) non-integrable distribution $D$ on $Q$ representing the permitted velocities. In this article, we consider Lagrangians of {\it mechanical-type}, that is, $L = \frac{1}{2}\kappa-\tau_{\mbox{\tiny{$TQ$}}}^*U$, where $\kappa$ is the \textit{kinetic energy metric} and $U: Q \rightarrow \mathbb{R}$ is the potential energy with $\tau_{\mbox{\tiny{$TQ$}}}:TQ\to Q$ the canonical projection, see \cite{Bloch2003}. Examples of nonholonomic systems are mainly mechanical systems with rolling constraints, for example solids -- as disks, spheres, ellipsoids, etc -- rolling without sliding on different smooth surfaces.  

\subsection{The nonholonomic bracket}\label{Ss:NHbracket}
The Legendre transformation $Leg = \kappa^{\flat}: TQ \rightarrow T^{\ast}Q$ (which, in this case, is a global diffeomorphism linear on the fibers) and the constraint distribution $D$ define the submanifold $\mathcal{M} := Leg(D) \subset T^{\ast}Q$ called the \textit{constraint manifold}.   Let us denote by $\tau_{\mbox{\tiny{$\mathcal{M}$}}} := \tau_{\mbox{\tiny{$T^*Q$}}} |_{\mbox{\tiny{$\mathcal{M}$}}} : \M \to Q$ the restriction to $\M$ of the canonical projection $\tau_{\mbox{\tiny{$T^*Q$}}}: T^*Q \to Q$.  Then the distribution $D$ induces a (constant rank) non-integrable distribution $\mathcal{C}$ on $\mathcal{M}$, given at each $m\in \mathcal{M}$,  by
\begin{equation*}
    \mathcal{C}_{m} := \left\{v_{m}\in T_{m}\mathcal{M}\,:\, T\tau_{\mbox{\tiny{$\mathcal{M}$}}}(v_{m})\in D_q,  q = {\tau_{\mbox{\tiny{$\mathcal{M}$}}}(m)}\right\} \subset T_{m}\mathcal{M}.
\end{equation*}

For $\iota_\subM: \M\to T^*Q$ the natural inclusion, we denote by 
 $\Omega_\subM:= \iota_\subM^*\Omega_\subQ$ the pull back of the canonical symplectic 2-form $\Omega_\subQ$ to $\M$ and $H_\subM:=\iota_\subM^*H$  the restriction of the Hamiltonian function $H$ to $\M$.
Due to the nondegenerancy of $\Omega_{\mathcal{C}}:= \left.\Omega_{\subM}\right|_{\mathcal{C}}$ (see \cite{Bates1993}), the {\it nonholonomic bracket} $\pi_{\nh}$ on $\mathcal{M}$ \cite{IbLMM99, Marle98, Schaft-Maschke-1994} is defined, at each $f\in C^\infty(\M)$, by 
\begin{equation} \label{Def:NH-Bracket} 
    \pi_\nh^\sharp(df) = -X_f \qquad \mbox{if and only if} \qquad {\bf i}_{X_{f}}\Omega_{\mathcal{C}} = \left.df\right|_{\mathcal{C}}.
\end{equation}
The {\it characteristic distribution} of $\pi_\nh$ is given by the nonintegrable distribution $\mathcal{C}$.  
The nonholonomic dynamics is described by the integral curves of the {\it nonholonomic vector field} $X_\nh$ on $\M$ defined by 
$$
    \pi_{\nh}^{\sharp}(dH_\subM) = -X_\nh,
$$
and that is why we say that the nonholonomic system is determined by the triple $(\mathcal{M},\pi_{\nh},H_\subM)$.   In the context of the previous section, we may see a nonholonomic system defined by the triple $(\M, L_\nh, H_\subM)$ where $L_\nh = \textup{graph}(\pi_\nh^\sharp)$ is an almost Dirac structure (with associated pair given by $(\mathcal{C}, \Omega_{\mathcal C})$, see \eqref{Eq:L_(F,w)}) and the dynamics is described by the pair $(-X_\nh, dH_\subM)\in L_\nh$.

\subsection{Symmetries and the nonholonomic momentum bundle map} \label{Ss:NHSymmetries}

The action of a Lie group $G$ on the configuration space $Q$ defines a \textit{$G$-symmetry} for the nonholonomic system if it is free and proper and the tangent lift leaves the Lagrangian $L$ and the constraint distribution $D$ invariant (or equivalently if the cotangent lift leaves the Hamiltonian $H$ and the constraint manifold $\mathcal{M}$ invariant). Therefore, there is a free and proper action $\psi:G\times \M\to \M$ defined on the manifold $\mathcal{M}$ and by construction the nonholonomic bivector field $\pi_\nh$ is $G$-invariant  as well.   

\subsubsection*{Reduction by symmetries} \label{Sss:NHSRed}
Since the action on $\M$ is free and proper, the orbit projection $\rho:\M\to \M/G$ defines a $G$-principal bundle.  
The nonholonomic vector field $X_{\nh}$ is $G$-invariant (i.e., $g\in G, T\psi_g(X_\nh(m)) = X_\nh(\psi_g(m))$) and hence it descends to the quotient manifold $\mathcal{M}/G$ defining the \textit{reduced nonholonomic vector field} given by $X_{\red}:= T\rho(X_{\nh})$. Moreover, the nonholonomic bivector field $\pi_\nh$ also descends to a reduced bivector field $\pi_\red$ on $\M/G$  defined, for each $f\in C^\infty(\M/G)$, as $\pi_\red^\sharp(df) = T\rho(\pi_\nh^\sharp(d\rho^*f))$ and the reduced nonholonomic vector field is also defined by 
\begin{equation}\label{Eq:RedDynamics}
    \pi_\red^\sharp(d H_\red) = -X_\red,
\end{equation}
where $H_\red\in C^\infty(\M/G)$ is the reduced Hamiltonian, $\rho^*H_\red = H_\subM$.

Consider now the $G$-action on $Q$ and denote by $V$ the {\it vertical distribution} on $Q$ given at each $q\in Q$ by $V_{q} := T_{q}(Orb_{G}(q))$.  Following \cite{Bloch1996} we say that the nonholonomic system defined by the distribution $D$ satisfies the \textit{dimension assumption} if
\begin{equation*}
    T_{q}Q = D_{q}+V_{q}, \quad \mbox{for all} \ q\in Q.
\end{equation*}
The distribution $S$ on $Q$ is given, at each $q\in Q$, by $S_{q} := D_{q}\cap V_{q}$ and observe that $V$ and $S$ have constant rank due to the freeness of the action and the dimension assumption. 

Equivalently, the dimension assumption can be  written on $\M$ (as in the Introduction), that is, for each $m\in \mathcal{M}$,  $T_{m}\mathcal{M} = \mathcal{C}_{m}+\mathcal{V}_{m}$ where $\mathcal{V}$ is the vertical distribution of the $G$-action on $\M$.  Analogously, for each $m\in \M$, we denote by $\mathcal{S}_{m} := \mathcal{C}_{m}\cap \mathcal{V}_{m}$ and note that $\textup{rank}(S)=\textup{rank}(\mathcal{S})$. 

Observe that the distribution $\mathcal{S}$ coincides with the one defined in \eqref{Eq:DefS}.  In this case, $\mathcal{S}$ satisfies automatically condition (A2), since the null distribution $K_{L_{\nh}} = \{0\}$, however condition (A1) is more subtle.

\subsubsection*{Conserved quantities induced by symmetries} \label{Sss:NHSConserved}

Following \cite{Bloch1996} we consider the subbundle $\mathfrak{g}_{S}\rightarrow Q$ of the trivial bundle $Q\times \mathfrak{g}\rightarrow Q$ with fiber, at $q\in Q$, given by $\mathfrak{g}_{{S}}|_{q} :=\left\{(q,\xi)\in Q\times\mathfrak{g} \ : \ \xi_{\subQ}(q) \in {S}_{q}\right\}$.  

\begin{remark}
    Let $\g_{\mathcal S}\to \M$ be the bundle defined in \eqref{Eq:gs}, i.e., for $m\in \M$, $\g_{\mathcal S}|_m =\{\xi_m\in \g : (\xi_m)_\subM(m) \in \mathcal{S}_m\}$, for $\mathcal{S} = \C \cap \mathcal{V}$. Then a section $\xi$ on  $\g_S\to Q$ induces a (basic) section $\widetilde \xi$  on $\g_{\mathcal S}\to \M$ given by $\widetilde \xi (m) = \xi(\tau_\subM(m))$.  
    Therefore, for a section $\xi$ on $\g_S\to Q$, we define the infinitesimal generator $\xi_\subM \in \Gamma(\mathcal{S})$ given by  $\xi_\subM(m):= (\xi_q)_\subM (m) \in \mathcal{S}_m$ for $\tau_\subM(m) = q$.  
    Since we work with lifted $G$-actions (and by Def.~\ref{Def:HGM}), from now on, we consider the bundle  $\mathfrak{g}_{S}\rightarrow Q$.
\end{remark}

The link between the reduction using a momentum map and mechanics is Noether theorem: we need to guarantee that $X_\nh$ is tangent to the level sets of the momentum map. That is why we are interested in conserved quantities of nonholonomic systems and their relation with a momentum.   However, $G$-invariant nonholonomic systems may not have conserved quantities; and even if the system admits conserved quantities arising from the symmetries, they may not be given by an element of the Lie algebra but by sections of the bundle $\g_S\to Q$, as it was observed in \cite{BatesGraumann96,FGS2005,Fasso2008}. More precisely: 

\begin{definition}\label{Def:HGM}\cite{BatesGraumann96}
    A function $J\in C^{\infty}(\mathcal{M})$ is a \textit{horizontal gauge momentum} of the vector field $X_{\nh}$ if
    \begin{enumerate}
        \item[$(i)$] $J$ is a conserved quantity of $X_{\nh}$, i.e., $X_{\nh}(J) = 0$; and
    
        \item[$(ii)$] there exists $\xi\in \Gamma(\mathfrak{g}_{S})$ such that $J = {\bf i}_{\xi_\subM} \Theta_\subM$,
    \end{enumerate}
    where $\Theta_{\subM} := \iota_{\subM}^{\ast}\Theta_{\subQ}$ and $\Theta_{\subQ}$ is the Liouville 1-form in $T^{\ast}Q$. In this case, the section $\xi\in \Gamma(\mathfrak{g}_{S})$ is a \textit{horizontal gauge symmetry} and we denote $J =: J_\xi$.
\end{definition}

Following \cite{Bloch1996}, we consider the \textit{nonholonomic momentum (bundle) map} $J_{\nh} : \mathcal{M} \rightarrow \mathfrak{g}_{S}^{\ast}$ which is the bundle map covering the identity on $Q$ defined, for each $m\in \mathcal{M}\subset T^*Q$ and $\xi\in \Gamma (\mathfrak{g}_{S})$, by
\begin{equation*} 
    \langle J_{\nh}(m), \xi(\tau_{\subM}(m))\rangle := {\bf i}_{\xi_{\mathcal{M}}}\Theta_{\subM}(m) = \langle m , \xi_{\subQ}(\tau_\subM(m))\rangle. 
\end{equation*}
Given $\xi\in \Gamma (\mathfrak{g}_{S})$, there is an induced function $J_{\xi}\in C^\infty(\mathcal{M})$ given by $J_{\xi}(m) := \langle J_\nh, \xi\rangle (m) = \langle J_{\nh}(m), \xi(\tau_\subM(m))\rangle$. In particular, the nonholonomic momentum map encodes the horizontal gauge momenta, in the sense that,  if   $J_\xi$ is a horizontal gauge momentum then $J_\xi = \langle J_{\nh} , \xi \rangle$. However for any $\zeta\in \Gamma(\g_S)$, the function $\langle J_\nh,\zeta\rangle$ may not be necessarily a horizontal gauge momentum.  As a consequence, we assume the following condition (already considered in \cite{Balseiro2021,LuisG2016}):

\medskip


\noindent {\bf Conserved quantities assumption:} A nonholonomic system $(\mathcal{M}, \pi_{\nh},H_{\subM})$ with a $G$-symmetry satisfying the dimension assumption,  verifies the \textit{conserved quantity assumption} if it admits $k = \textup{rank}(S)$ horizontal gauge momenta $\{J_{1},\dots,J_{k}\}$ that are $G$-invariant and functionally independent.

\medskip

Next we observe that the conserved quantity assumption implies condition (A1). 

\begin{lemma}\label{L:Ad-basis} 
    The conserved quantity assumption induces an $Ad$-invariant basis 
    $$
        \mathfrak{B}_{\mbox{\tiny\emph{HGS}}} := \left\{\xi_{1},\dots, \xi_{k}\right\},
    $$ 
    of sections of $\g_S\to Q$ of horizontal gauge symmetries.
\end{lemma}

\begin{proof}
    First, let us observe that if the horizontal gauge momenta $\{J_{1},\dots,J_{k}\}$ are functionally independent, then the corresponding horizontal gauge symmetries $\left\{\xi_{1},\dots, \xi_{k}\right\}$ are linearly independent (and globally defined). In fact, if $f\xi_1 = \xi_2$ for $f\in C^\infty(Q)$, then by Def.~\ref{Def:HGM}$(ii)$, $J_1= (\tau_\subM^*f)J_2$.  Using that $J_i$ are linear on the fibers and that $f$ does only depend on the basis $Q$, we see that $J_1= (\tau_\subM^*f)J_2$ contradicts the linear independence of $J_1$ and $J_2$.

    In order to prove the $Ad$-invariance of the sections $\xi_i$, recall that $Y_i = (\xi_i)_\subQ$ and $\mathcal{Y}_i = (\xi_i)_\subM$ are the vector fields on $Q$ and $\M$ respectively.  Using the $G$-invariance of the Liouville 1-form $\Theta_\subM$ and the function $J_i$, we have, for $g\in G$ and $m\in \M,$ that 
    $$
        \Theta_\subM(\psi_g(m))(T\psi_g(\mathcal{Y}_i(m))) = \Theta_\subM(m)(\mathcal{Y}_i(m)) = J_i(m)=J_i(\psi_g(m)) = \Theta_\subM(\psi_g(m))(\mathcal{Y}_i(\psi_g(m))).
    $$
    Using the definition of the Liouville 1-form, we obtain, for $q=\tau_\subM(m)\in Q$ that $Y_i(\psi_g(q)) = T\psi_g(Y_i(q)) = (Ad_g((\xi_i)_q))_\subQ(\psi_g(q))$. Finally we conclude that $\xi_i(q) = Ad_g((\xi_i)_q)$ by the freeness of the $G$-action. 
\end{proof}

Lemma \ref{L:Ad-basis} defines also a $G$-invariant global basis of $S$ and $\mathcal{S}$ given by 
\begin{equation}\label{Eq:BasisY}
    S= \textup{span}\{Y_1=(\xi_1)_\subQ,\dots, Y_k=(\xi_k)_\subQ\}\quad \mbox{ and } \quad \mathcal{S}= \textup{span}\{\mathcal{Y}_1=(\xi_1)_\subM,\dots, \mathcal{Y}_k=(\xi_k)_\subM\},
\end{equation}
respectively.  We also define the dual (global $Ad^*$-invariant) basis of sections
\begin{equation}\label{Eq:BasisHGS}
    \mathfrak{B}_{\mbox{\tiny{HGS}}}^{\ast} := \left\{ \mu^{1},\dots,\mu^{k}\right\},
\end{equation}
of the bundle $\g_S^*\to Q$ as in \eqref{Eq:DefMu}.  

Following Def.~\ref{Def:MomentumBundle} we observe that the nonholonomic momentum bundle map is not necessarily a momentum map for $\pi_\nh$ associated to the basis $\mathfrak{B}_{\mbox{\tiny{HGS}}}$ since, as it was observed in \cite{Balseiro2017,LuisG2016}, $\pi_\nh^\sharp(dJ_i)$ might not coincide with the infinitesimal generator $-(\xi_i)_{\subM}$ (in the context of Sec.~\ref{S:Dirac},  $(-(\xi_i)_{\subM}, dJ_i)$ might not be in $L_{\nh}= \textup{graph}(\pi_\nh^\sharp)$).

In the next section we will deal with this problem by replacing the nonholonomic bracket by another almost Poisson bracket following \cite{Balseiro2021,LuisG2010}.

\subsection{Gauge transformations} \label{Ss:GaugeTrans}

In what follows we study a new bivector field $\pi_\subB$ that describes the dynamics $\pi_\subB^\sharp(dH_\subM) = -X_\nh$ and for which $J_\nh:\M\to \g_S^*$ is a momentum map associated to the basis $\mathfrak{B}_{\mbox{\tiny{HGS}}}$.  In order to define such bracket we ``deform'' the nonholonomic bracket by means of a 2-form.  We follow \cite{LuisN2012} where it was introduced a {\it dynamical gauge transformation} of a bivector field by a 2-form, based in the concept of {\it gauge transformations} \cite{Severa2001}.

First, recall the definition of the nonholonomic bracket $\pi_\nh$ given in \eqref{Def:NH-Bracket}.  Now, if $B$ is a 2-form on $\M$ so that $\left.(\Omega_\subM+B)\right|_{\mathcal{C}}$ is a nondegenerate 2-section on $\mathcal{C}$, then the {\it gauge transformation} of $\pi_\nh$ by $B$ is a new bivector field $\pi_\subB$ such that
\begin{equation*}
    \pi_\subB^\sharp(\alpha) = X \iff {\bf i}_X (\Omega_\subM+B)|_{\mathcal{C}} = -\alpha|_{\mathcal{C}},
\end{equation*}
and we say that $\pi_\nh$ and $\pi_\subB$ are {\it gauge related} (by the 2-form $B$), \cite{Severa2001}.

In particular, gauge related (almost) Poisson brackets share the characteristic distribution, which in this case is $\mathcal{C}$ for $\pi_\nh$ and $\pi_\subB$.  If $B$ is a semi-basic 2-form with respect to the bundle $\tau_\subM:\M\to Q$ (i.e., ${\bf i}_XB=0$ for all $X\in T\M$ such that $T\tau_\subM(X)=0$), then the 2-section $(\Omega_\subM+B)|_{\mathcal{C}}$ is automatically nondegenerated.

\begin{remark}\label{R:GaugeConsiderations}
    \begin{enumerate}
        \item[$(i)$] If $(\Omega_\subM+B)|_{\mathcal{C}}$ is degenerate, then the remaining structure is a Dirac structure.  However, we will see that in our case of study, we always work with semi-basic 2-forms $B$.

  \item[$(ii)$] In general, the (almost) Dirac structures $L_1$ and $L_2$ on $P$ are {\it gauge related} by a 2-form $B$ if $L_2 = \{(X, \alpha + {\bf i}_XB) \in \mathbb{T}P: (X,\alpha)\in L_1 \}$.  In this case, if $(F,\omega)$ is the associated pair of $L_1$, then $(F, \omega+B)$ is the pair associated to $L_2$.  
    
        \item[$(iii)$] If $\pi$ and $\tilde \pi$ are gauge related  bivector fields by a 2-form $B$ with integrable characteristic distribution,  then the corresponding 2-forms on each leaf differ from $B_\mu$, i.e., the restriction of $B$ to the leaf $\mathcal{O}_\mu$.  Even though a gauge transformation keeps the characteristic distribution invariant, gauge related brackets might have different Hamiltonian vector fields, see \cite{Severa2001}.

    \end{enumerate}
\end{remark}

Following \cite{LuisN2012}, we say that a gauge transformation of $\pi_{\nh}$ by a 2-form $B$ satisfies the \textit{dynamical condition} if
\begin{equation} \label{Eq:DynCond}
    {\bf i}_{X_{\nh}}B = 0,
\end{equation} 
and the 2-form $B$ defines a {\it dynamical gauge transformation} of $\pi_{\nh}$. In this case, the gauge related bivector field $\pi_\subB$ describes also the nonholonomic dynamics: $\pi_\subB^\sharp (dH_\subM)=-X_\nh$. 

When the nonholonomic system admits a $G$-symmetry and the 2-form $B$ is $G$-invariant, then the gauge related bivector field $\pi_\subB$ is $G$-invariant as well, and it descends to the reduced quotient manifold $\M/G$ as a reduced bivector field $\pi_\red^\subB$ defined, at each $\alpha \in T^{\ast}(\mathcal{M}/G)$, by
$$
    (\pi_{\red}^{\subB})^{\sharp}(\alpha) := T\rho(\pi_{\subB}^{\sharp}(\rho^{\ast}\alpha)),
$$
where $\rho : \mathcal{M}\rightarrow \mathcal{M}/G$ is the orbit projection.  When $B$ satisfies the dynamical condition \eqref{Eq:DynCond} then  $(\pi_\red^\subB)^\sharp(dH_\red) = -X_\red$, c.f., \eqref{Eq:RedDynamics}.

\begin{remark}
    It is straightforward to check that the almost Dirac structures $L_\red= \textup{graph}(\pi_\red^\sharp)$ and $L_\red^\subB=\textup{graph}((\pi_\red^\subB)^\sharp)$ are the forward image by $\rho : \mathcal{M}\to \mathcal{M}/G$  of  the almost Dirac structures $L_\nh = \textup{graph}(\pi_\nh^\sharp)$ and $L_\subB = \textup{graph}(\pi_\subB^\sharp)$ respectively, see Remark \ref{R:Dirac:BF-Examples}.  
\end{remark}

As it was observed in \cite{LuisN2012} (also in \cite{LuisG2010}, but with more detail in \cite{Balseiro2021} and \cite{LuisG2016}), the reduced bivector fields $\pi_\red$ and $\pi_\red^\subB$ might have different properties, for example one might be Poisson while the other not.   In particular, in \cite{Balseiro2021} (see also \cite{LuisG2016}) it is defined a 2-form $B$ so that $\pi_\red^\subB$ has an integrable characteristic distribution. 

\subsubsection*{A particular dynamical gauge transformation}

Now, we define the 2-form $B$ presented in \cite{Balseiro2021}  and here we will show that $J_\nh:\M\to \g^*_S$ is a momentum map on the almost Poisson manifold $(\M,\pi_\subB)$ associated to the basis $\mathfrak{B}_{\mbox{\tiny{HGS}}}$.

Let $(\M, \pi_\nh, H_\subM)$ be a nonholonomic system  with a $G$-symmetry satisfying the dimension and the conserved quantity assumptions. Consider now a $G$-invariant distribution $H \subset D$ on $Q$ and a $G$-invariant vertical distribution $W\subset V$ on $Q$ -- called a {\it vertical complement of the constraints} \cite{Paula2014} --  so that
\begin{equation}\label{Eq:splittingTQ}
    TQ = H\oplus S\oplus W, \quad \mbox{where } D=H\oplus S \mbox{ and } V=S\oplus W.
\end{equation}

\begin{remark} 
    The distributions $H$ and $W$ can be chosen to be $H = D\cap S^\perp$ and $W= V\cap S^\perp$. However, we remark that they do not have to be necessarily chosen in this way \cite{Paula2014}.
\end{remark}

We have an analogous splitting on $T\M$ given by
\begin{equation} \label{Eq:splittingTM}
    T\M =\mathcal{H}\oplus \mathcal{S}\oplus \mathcal{W} \quad \mbox{where} \quad \mathcal{C} =\mathcal{H}\oplus \mathcal{S} \mbox{ and } \mathcal{V}=\mathcal{S}\oplus \mathcal{W},
\end{equation}
where $\mathcal{H}=\{v\in \mathcal{C} \,:\, T\tau_\subM(v)\in H\}$ and $\mathcal{W} =\{v\in \mathcal{V} \,:\, T\tau_\subM(v)\in W\}$.

Using that $TQ=D\oplus W$, we denote by $P_{\mbox{\tiny{$D$}}} :TQ \to D$ and $P_{\mbox{\tiny{$W$}}}:TQ\to W$ the projections to the first and second factor respectively. 
Then let $A_{\mbox{\tiny{$W$}}}:TQ\to \g$ be the map given, at each $v_q\in T_qQ$, by $A_{\mbox{\tiny{$W$}}}(v_q) = \eta\in \g$ if and only if $P_{\mbox{\tiny{$W$}}}(v_q) = \eta_\subQ(q)$ and we define the 2-form $K_{\mbox{\tiny{$W$}}}$ on $Q$ so that $K_{\mbox{\tiny{$W$}}}(X,Y) = dA_{\mbox{\tiny{$W$}}}(P_{\mbox{\tiny{$D$}}}(X),P_{\mbox{\tiny{$D$}}}(Y))$, for  $X,Y\in \mathcal{X}(Q)$.
Following \cite{Paula2014}, the $\mathcal{W}$\textit{-curvature} $\mathcal{K}_{\mbox{\tiny{${\mathcal W}$}}}$ is the $\mathfrak{g}$-valued $2$-form on $\mathcal{M}$ defined by $\mathcal{K}_{\mbox{\tiny{${\mathcal W}$}}} := \tau_{\subM}^{\ast} K_{\mbox{\tiny{$W$}}}$. On the other hand, let $J : \mathcal{M}\rightarrow \mathfrak{g}^{\ast}$ be the restriction of the canonical momentum map of $T^{\ast}Q$ to $\mathcal{M}$. 
We define the $2$-form $\langle J,\mathcal{K}_{\mbox{\tiny{${\mathcal W}$}}}\rangle$ on $\mathcal{M}$ as the pairing between $J(m)\in \mathfrak{g}^{\ast}$ and $\mathcal{K}_{\mbox{\tiny{${\mathcal W}$}}}(X_m,Y_m) \in \mathfrak{g}$, for $m\in \mathcal{M}$, and $X,Y\in T_m\mathcal{M}$.

Given the {\it conserved quantity assumption}, and from \eqref{Eq:BasisY} we denote by $\{Y^1,\dots, Y^k\}$ the 1-forms satisfying that $Y^i(Y_j)=\delta_{ij}$ and $Y^i|_H = Y^i|_W = 0$, for $i=1,\dots,k.$
We define the 2-form $B_1$ on $\M$ by
\begin{equation} \label{Eq:B1}
    B_{1} := \langle J, \mathcal{K}_{\mbox{\tiny{${\mathcal W}$}}}\rangle+J_{i}\, d^{\mathcal{C}}\mathcal{Y}^{i},
\end{equation}
where $\mathcal{Y}^{i}= \tau_\subM^*Y^i$, $i=1,\dots,k$ are 1-forms on $\M$ and, for $X, Y\in T\M$,  $d^{\mathcal{C}}\mathcal{Y}^{i}(X,Y) = d\mathcal{Y}^{i} (P_{\mbox{\tiny{${\mathcal C}$}}}(X),P_{\mbox{\tiny{${\mathcal C}$}}}(Y))$ for $P_{\mbox{\tiny{${\mathcal C}$}}}: \mathcal{C}\oplus \mathcal{W}\to {\mathcal{C}}$ the corresponding projection.

Next, consider the  principal curvature $\mathcal{K}_{\mbox{\tiny{${\mathcal V}$}}}$ associated with the principal connection $\mathcal{A}_{\mbox{\tiny{${\mathcal V}$}}} :T\M\to \g$ with horizontal space $\mathcal{H}$ given in \eqref{Eq:splittingTM}. We also define $\kappa_{\mathfrak{g}}$ as the $\mathfrak{g}^{\ast}$-valued 1-form on $\mathcal{M}$ given, at each  $X\in T\mathcal{M}$ and $\eta\in \mathfrak{g}$, by $\left\langle \kappa_{\mathfrak{g}}(X), \eta \right\rangle = \kappa\left(T\tau_\subM(X),\eta_{Q}\right)$.
Finally, we define the 2-form $\mathcal{B}$ on $\mathcal{M}$ given by
\begin{equation}\label{Eq:Bcal}
    \mathcal{B} := -\langle J,\mathcal{K}_{\mathcal{V}}\rangle-\frac{1}{2}(\kappa_{\mathfrak{g}}\wedge {\bf i}_{P_{\mathcal{V}}(X_{\nh})}[\mathcal{K}_{\mbox{\tiny{${\mathcal W}$}}} + d^{\mathcal{C}}\mathcal{Y}^{i}\otimes \xi_{i}])_{\mathcal{H}}.
\end{equation}
That is, for $X,Y\in \Gamma(\C)$,  
\begin{equation*}
    \begin{split}
        \mathcal{B}(X,Y) := -\langle J,\mathcal{K}_{\mathcal{V}}\rangle(X,Y) & -\tfrac{1}{2}\langle\kappa_{\mathfrak{g}} (P_{\mbox{\tiny{$\mathcal{H}$}}}(X)), [\mathcal{K}_{\mbox{\tiny{${\mathcal W}$}}} + d^{\mathcal{C}}\mathcal{Y}^{i}\otimes \xi_{i}](P_{\mbox{\tiny{$\mathcal{V}$}}}(X_\nh), P_{\mbox{\tiny{$\mathcal{H}$}}}(Y))\rangle \\
        & + \tfrac{1}{2}\langle\kappa_{\mathfrak{g}} (P_{\mbox{\tiny{$\mathcal{H}$}}}(Y)), [\mathcal{K}_{\mbox{\tiny{${\mathcal W}$}}} + d^{\mathcal{C}}\mathcal{Y}^{i}\otimes \xi_{i}](P_{\mbox{\tiny{$\mathcal{V}$}}}(X_\nh), P_{\mbox{\tiny{$\mathcal{H}$}}}(X))\rangle.
    \end{split}
\end{equation*}
For more details on the definitions of $B_1$ and $\mathcal{B}$ see \cite{Balseiro2021}. 
\begin{remark}\label{R:B1B}  Consider the 2-forms $B_1$ and $\mathcal{B}$ defined in \eqref{Eq:B1} and \eqref{Eq:Bcal} respectively.
    \begin{enumerate}
        \item[$(i)$] The $G$-invariance of the horizontal gauge momenta $\{J_1,\dots,J_k\}$ imply the $G$-invariance of the 2-forms $B_1$ and $\mathcal{B}$, \cite{Balseiro2021}.
        
        \item[$(ii)$]  Moreover, $\mathcal{B}$ is basic with respect to the principal bundle $\rho:\M\to \M/G$, that is, there is a 2-form $\bar{\mathcal{B}}$ on $\M/G$ such that $\rho^*{\bar{\mathcal B}} = \mathcal{B}$.
    \end{enumerate}
\end{remark}

By construction, $B_{1}$ and $\mathcal{B}$ are semi-basic with respect to the bundle $\tau_\subM : \mathcal{M} \rightarrow Q$ and hence

\begin{definition}\cite{Balseiro2021} \label{Def:PiPB} 
    Consider the nonholonomic system $(\M, \pi_\nh, H_\subM)$ with a $G$-symmetry satisfying the dimension and the conserved quantity assumptions.
    \begin{enumerate}
        \item[$(i)$] Let $\pi_{1}$ be the bivector field on $\M$ obtained by the gauge transformation of $\pi_\nh$ by the 2-form $B_1$.
        
        \item[$(ii)$] Let $\pi_{\subB}$ be the bivector field on $\M$ obtained by the gauge transformation of $\pi_\nh$ by the 2-form $B: = B_{1} + \mathcal{B}$. 
    \end{enumerate} 
\end{definition}

Observe that the bivector fields $\pi_1$ and $\pi_\subB$ are gauge related by the 2-form $\mathcal{B}$.  
However, as it was proven in \cite{Balseiro2021}, only $\pi_{\subB}$ is compatible with the dynamics:
\begin{equation}\label{Eq:NHB-dynamics}
    \pi_{\subB}^{\sharp}(dH_{\subB}) = -X_{\nh},
\end{equation}
that is, the 2-form $B$ satisfies the dynamical condition \eqref{Eq:DynCond}.

\begin{proposition}\label{Prop:Jnh-MomMap:MBM} 
    The nonholonomic momentum bundle map $J_{\emph\nh}:\M\to \g^*_S$ is a momentum bundle map for the bivector fields $\pi_{1}$ and $\pi_{\subB}$ associated to the basis $\mathfrak{B}_{\mbox{\tiny{\emph{HGS}}}}$.
\end{proposition}

\begin{proof}
    It is straightforward to prove this fact for $\pi_1$ using coordinates: first,  using the basis \eqref{Eq:BasisY}, we define a local basis of sections of $TQ$ adapted to the splitting \eqref{Eq:splittingTQ}, given by 
    \begin{equation} \label{Eq:BasisTQ}
        \mathfrak{B}_{\mbox{\tiny{$TQ$}}} = \{ X_{\alpha}, {Y}_{i} := (\xi_{i})_{\subQ}, Z_{a}\} \quad \mbox{with its dual basis} \quad \mathfrak{B}_{\mbox{\tiny{$T^*Q$}}} = \{X^{\alpha}, {Y}^{i}, \epsilon^{a}\},
    \end{equation} 
    where $\{X_{\alpha}\}$ and and $\{Z_a\}$ are $G$-invariant local basis of $H \subset D$ and $W\subset V$ respectively, and the 1-forms $\epsilon^{a}$ are the constraint 1-forms on $Q$. If we denote the coordinates on $T^{\ast}Q$ induced by $\mathfrak{B}_{\mbox{\tiny{$T^{\ast}Q$}}}$ as $(p_{\alpha},p_{i},p_{a})$, then we can use coordinates $(p_{\alpha},p_{i})$ on $\mathcal{M}$, since $p_{a}$ depends linearly on $p_{\alpha}$ and $p_{i}$. Therefore, we have a local basis $\mathfrak{B}_{\mbox{\tiny{$T\M$}}}$ of $T\mathcal{M}$ adapted to the splitting $T\M =\mathcal{H}\oplus \mathcal{S}\oplus \mathcal{W}$ with dual basis $\mathfrak{B}_{\mbox{\tiny{$T^*\M$}}}$ given by
    \begin{equation}\label{Eq:localCoordTM}
        \mathfrak{B}_{\mbox{\tiny{${T\M}$}}} = \{\mathcal{X}_{\alpha}, \mathcal{Y}_{i}, \mathcal{Z}_{a},\partial_{p_{\alpha}},\partial_{p_{i}}\}\quad \textup{and}\quad \mathfrak{B}_{\mbox{\tiny{${T^{\ast}\!\M}$}}} = \{\mathcal{X}^{\alpha}, \mathcal{Y}^{i}, \tilde{\epsilon}^{a},dp_{\alpha},dp_{i}\},
    \end{equation}
    where $\mathcal{X}^{\alpha} = \tau_{\subM}^{\ast}X^{\alpha}$, $\mathcal{Y}^{i} = \tau_{\subM}^{\ast}Y^{i}$ and $\tilde{\epsilon}^{a} = \tau_{\subM}^{\ast}\epsilon^{a}$ (in particular, $\mathfrak{B}_{\mbox{\tiny{${T^{\ast}\M}$}}}$ is defined first from $\mathfrak{B}_{\mbox{\tiny{${T^{\ast}Q}$}}}$ and then $\mathfrak{B}_{\mbox{\tiny{${T\M}$}}}$ is defined as the dual basis). 
    Since that $\mathcal{Y}_i=(\xi_i)_\subM$ then each horizontal gauge momenta is given by $J_{i} = p_{i}$ for $i=1,\dots,k$. Using that $\langle J, \mathcal{K}_{\mbox{\tiny{${\mathcal W}$}}}\rangle|_{\mathcal{C}} = \imath_{\subM}^{\ast}(p_{a})d\tilde{\epsilon}^{a}|_{\mathcal{C}}$, then
    \begin{equation}\label{Eq:Omega1-Coord}
        (\Omega_{\subM}+B_1)|_{\mathcal{C}} = \left(\mathcal{X}^{\alpha}\wedge dp_{\alpha} + \mathcal{Y}^{i}\wedge dp_{i}-p_{\alpha}d\mathcal{X}^{\alpha}\right)|_{\mathcal{C}}. 
    \end{equation}
    Recall that, if $X$ is an invariant vector field on $\M$ and $Y\in \Gamma(\mathcal{V})$, then $[X,Y]\in \Gamma(\mathcal{V})$ which implies that $d\mathcal{X}^\alpha|_{\mathcal{C}} = -\frac{1}{2}C_{\beta\gamma}^\alpha \mathcal{X}^\beta\wedge \mathcal{X}^\gamma|_{\mathcal{C}}$, where $C_{\alpha\beta}^{\gamma} := {X}^{\gamma}([{X}_{\alpha},{X}_{\beta}])$ denote the structure functions relative to the basis $\mathfrak{B}_{\mbox{\tiny{$TQ$}}}$. Thus ${\bf i}_{(\xi_i)_\subM}(\Omega_{\subM}+B_1)|_{\mathcal{C}} = dp_{i}|_{\mathcal{C}}$ and hence $\pi_1^\sharp(dJ_i) = - (\xi_i)_\subM$, for each $i = 1,\dots,k$.
    
    Finally, since $\mathcal{B}$ is basic (Remark \ref{R:B1B}$(ii)$) with respect to the bundle $\M\to \M/G$ then ${\bf i}_{(\xi_i)_\subM}(\Omega_{\subM}+B)|_{\mathcal{C}} = dp_{i}|_{\mathcal{C}}$, concluding that $\pi_\subB^\sharp(dJ_i) = - (\xi_i)_\subM$ for each $i = 1,\dots,k$. 
\end{proof}

The previous proposition shows the power of considering the gauge transformation by $B_{1}$ and $B$ since the nonholonomic momentum map might not be a momentum bundle map for $\pi_{\nh}$ associated to the basis $\mathfrak{B}_{\mbox{\tiny{HGS}}}$, but it is for $\pi_{1}$ and $\pi_{\subB}$.

Due to the $G$-invariance of the bivector fields $\pi_1$ and $\pi_\subB$ we can define the reduced bivectors $\pi_{\red}^{1}$ and $\pi_{\red}^{\subB}$ on $\M/G$ given respectively, at each $\alpha\in T^*(\M/G)$, by
\begin{equation} \label{Eq:NHRedBrackets}
    (\pi_\red^1)^\sharp(\alpha) = T\rho (\pi_1^\sharp(\rho^*\alpha)) \qquad \mbox{and} \qquad (\pi_\red^\subB)^\sharp(\alpha) = T\rho (\pi_\subB^\sharp(\rho^*\alpha)).
\end{equation}
As a consequence of \eqref{Eq:NHB-dynamics}, the reduced bivector field $\pi_{\red}^{\subB}$ is compatible with the reduced dynamics: 
\begin{equation*}
    (\pi_\red^\subB)^\sharp(dH_\red) = - X_\red.
\end{equation*}
Since $\mathcal{B}$ is basic then both bivector fields $\pi_\red^1$ and $\pi_\red^\subB$ are gauge related by the 2-form $\bar{\mathcal{B}}$ defined on $\M/G$ and share the characteristic distribution (see Remark \ref{R:B1B}$(ii)$).  As a consequence of Prop.~\ref{Prop:Jnh-MomMap:MBM} and the dimension assumption, we obtain:

\begin{proposition}\label{Prop:Jnh-MomMap:RED} 
    The bivector fields  $\pi_{\emph\red}^{1}$ and $\pi_{\emph\red}^{\subB}$ on $\M/G$ have an integrable characteristic distribution given by the common level sets of the (reduced) horizontal gauge momenta $\bar{J}_i$, for $\bar{J}_i$ the functions on $\M/G$ such that $\rho^*\bar{J}_i = J_i = \langle J_{\emph\nh}, \xi_{i}\rangle$ for $i=1,\dots,k$. 
\end{proposition}

\begin{proof}
    From Prop. \ref{Prop:Jnh-MomMap:MBM}, we conclude that $\pi_1^\sharp(dJ_i) = - (\xi_i)_\subM$ for each $i=1,\dots,k$. Then it is easily checked that  $\textup{span}\{d\bar{J}_1,\dots,d\bar{J}_k\} = \textup{Ker}(\pi_\red^1)$, since $\textup{Ker}(\pi_1)$ are the constraint 1-forms (which are not basic with respect to $\rho:\M\to \M/G$). Therefore the characteristic distribution of $\pi_\red^1$ is given by the common level set the horizontal gauge momenta $\bar{J}_i$. This is also true for $\pi_\red^\subB$ since it is gauge related to $\pi_\red^1$ by $\bar{\mathcal{B}}$.
\end{proof}

The technique used in the last proofs (Props. \ref{Prop:Jnh-MomMap:MBM} and \ref{Prop:Jnh-MomMap:RED})  will be a common strategy in this paper: first we prove a certain property for $\pi_1$ or $\pi_\red^1$. Second, using that $\pi_1$ and $\pi_\subB$ are gauge related (respectively $\pi_\red^1$ and $\pi_\red^\subB$), the property is easily checked to be valid on $\pi_\subB$ (respectively on $\pi_\red^\subB$).

\subsection{Nonholonomic momentum reduction} \label{Ss:Nonholonomic_momentum_reduction}

Now we arrive to the main result of this section. After the discussion in Sec.~\ref{Ss:GaugeTrans} (in particular in Prop.~\ref{Prop:Jnh-MomMap:MBM}), we observe that due to the interaction of the conserved quantities and the nonholonomic momentum bundle map,  the momentum map reduction --stated in Theorem \ref{T:momentum-red}-- can be done on $(\M,\pi_1)$ and $(\M,\pi_\subB)$ but not necessarily on $(\M,\pi_\nh)$.

\subsubsection*{Level sets of the nonholonomic momentum map} \label{Sss:LevelSets}
Let us start with a nonholonomic system $(\mathcal{M},\pi_{\nh},H_{\mbox{\tiny{${\mathcal M}$}}})$ with a $G$-symmetry satisfying the dimension and the conserved quantity assumptions. Denote by $\{J_{1},\dots,J_{k}\}$ the ($G$-invariant and functionally independent) horizontal gauge momenta, with associated horizontal gauge symmetries defining a global basis $\mathfrak{B}_{\mbox{\tiny{HGS}}} = \{\xi_{1},\dots,\xi_{k}\}$ of $Ad$-invariant section of $\g_{S} \to Q$.
Let $\mathfrak{B}_{\mbox{\tiny{HGS}}}^{\ast} := \left\{ \mu^{1},\dots,\mu^{k}\right\}$ be the dual $Ad^*$-invariant basis of sections on $\g_S^*\to Q$ of $\mathfrak{B}_{\mbox{\tiny{HGS}}}$ given in \eqref{Eq:BasisHGS}, and consider $\mu = c_{i}\mu^{i}\in C^{\infty}(Q,\mathfrak{g}_{S}^{\ast})$, where $c_{i}\in \mathbb{R}$ are constants.
Following Prop.~\ref{Prop:Jsubmfld}, we consider the $G$-invariant submanifold $J_\nh^{-1}(\mu) \subset \mathcal{M}$ given by
\begin{equation}\label{Eq:NHLevelSet}
    J_{\nh}^{-1}(\mu) := \left\{m\in\mathcal{M}\,:\, \,J_\nh(m) = \mu(\tau_{\subM}(m))\right\} = \bigcap_{i=1}^{k} \, J_{i}^{-1}(c_{i}),
\end{equation}
and the collection of (connected components of) the manifolds $J_\nh^{-1}(\mu)$ determines a foliation of $\mathcal{M}$ for $\mu \in \textup{span}_{\mathbb{R}}\mathfrak{B}_{\mbox{\tiny{HGS}}}^{\ast}$. 
Since the functions $J_i$, for $i=1,\dots,k$, are conserved quantities of the nonholonomic vector field $X_\nh$ it is straightforward to see that $X_\nh$ is tangent to the manifolds $J_{\nh}^{-1}(\mu)$.
Let us denote by $X_{\mu}$ the vector field restricted to $J_\nh^{-1}(\mu)$, that is, 
\begin{equation*}
    X_{\mu} := \left.X_{\nh}\right|_{J_\nh^{-1}(\mu)}. 
\end{equation*}
By the $G$-invariance of the horizontal gauge momenta $J_i$ we also have that $J_\nh^{-1}(\mu) /G =  \cap_{i}\, \bar{J}_{i}^{-1}(c_{i})$, where, as usual, $\bar{J}_i\in C^\infty(\M/G)$ are the reduced horizontal gauge momenta (i.e., $\rho^*\bar{J}_i = J_i$) for $i=1,\dots,k.$

\subsubsection*{The momentum map reduction of $\pi_1$} \label{Sss:MMReductionPi1}

Consider the bivector field $\pi_1$ obtained by the gauge transformation of $\pi_\nh$ by the 2-form $B_1$ defined in \eqref{Eq:B1} and recall that the nonholonomic momentum map is a momentum bundle map associated to $\mathfrak{B}_{\mbox{\tiny{HGS}}}$ (see Prop.~\ref{Prop:Jnh-MomMap:MBM}). Furthermore, observe that conditions (A1) and (A2) remain satisfied for the bivector field $\pi_1$. 

Following the ideas of Section \ref{Ss:MomMapDirac}, we define $L_\mu^1 \subset \mathbb{T}(J_\nh^{-1}(\mu))$ the backward image of the almost Poisson bracket $\pi_{1}$ (i.e., of the almost Dirac structure $L_{1} := \textup{graph}(\pi_{1}^{\sharp})$) under the inclusion map $\iota_{\mu} : J_{\nh}^{-1}(\mu) \rightarrow \mathcal{M}$, given by
\begin{equation}\label{Eq:Lmu1}
    L_{\mu}^{1} := {\iota_{\mu}^{\ast}}(L_{1}) = \left\{ (X,\iota_{\mu}^{\ast}\beta) \in \mathbb{T}(J_{\nh}^{-1}(\mu)) \,:\, \pi_1^\sharp(\beta) = T\iota_{\mu}(X)\right\}.
\end{equation}
Since the functions $J_i$ are $G$-invariant, from Theorem \ref{T:momentum-red}$(i)$ and Prop.~\ref{Prop:Jnh-MomMap:MBM}, $L_{\mu}^{1}$ is a $G$-invariant almost Dirac structure on $J_{\nh}^{-1}(\mu)$.  Moreover, recalling that the associated pair of $L_1$ is $(\mathcal{C}, \Omega_1)$ where $\Omega_1 = \Omega_\subM +B_1$, then the associated pair of $L_1^\mu$ is $(\mathcal{C}_\mu, \Omega_\mu^1)$ for $\mathcal{C}_\mu = TJ^{-1}_\nh(\mu) \cap \mathcal{C}$ and $\Omega_\mu^1 = \iota_\mu^* \Omega_1$.

Since $(T(J_{\nh}^{-1}(\mu)))^{\circ} = \textup{span}\{dJ_{1},\dots,dJ_{k}\}$, and $\mathcal{S} = \textup{span}\{(\xi_{1})_{\subM},\dots,(\xi_{k})_{\subM}\}$ we conclude that: 

\begin{lemma}\label{L:lemmaY}
    The null distribution $K_\mu$ of the almost Dirac structure $L_\mu^1$ is the vertical distribution $\mathcal{S}$ restricted to $J^{-1}_{\emph\nh}(\mu)$, that is, for $m \in J^{-1}_{\emph\nh}(\mu)$, $(K_\mu)_m = \mathcal{S}_m$.   
\end{lemma}

From Lemma \ref{L:lemmaY} and the fact that $\textup{Ker}(\pi_1^\sharp)$ are the constraints 1-forms $\epsilon^a$, it is straightforward to see that $L_\mu^1$ is a genuinely almost Dirac structure in the sense that it is not the graph of a bivector field nor of a 2-form.  

Next, denoting by $\rho_\mu :J_\nh^{-1}(\mu) \to J_\nh^{-1}(\mu)/G$  the orbit projection and following Theorem~\ref{T:momentum-red}$(ii)$ we conclude that
\begin{equation*}
    (L_\mu^1)_{\red} := ({\rho_{\mu}})_{\ast}L_{\mu}^{1} = \left\{(T\rho_{\mu}Y,\alpha)\in \mathbb{T}(J_{\nh}^{-1}(\mu)/G)\,: \, (Y,\rho_{\mu}^{\ast}\alpha)\in L_{\mu}^{1}\right\},
\end{equation*}
is an almost Dirac structure. Moreover: 

\begin{proposition}\label{Prop:L1red_simplectic}
    The structure $(L_\mu^1)_{\emph\red}$ is the graph of a symplectic 2-form $\omega_{\mu}^{1}$ on $J_{\emph\nh}^{-1}(\mu)/G$, that satisfies 
    \begin{equation}\label{Eq:PullBacks1}
        \rho_\mu^*\omega_\mu^1\, |_{\mathcal{C}_\mu} = \iota_\mu^*(\Omega_\subM+B_1) |_{\mathcal{C}_\mu}.
    \end{equation}
\end{proposition}
\begin{proof}
    First, we prove that the almost Dirac structure $(L_\mu^1)_{\red}$ on $J_\nh^{-1}(\mu)/G$ is associated to a 2-form. From Corollary~\ref{C:MomReduction}, we know that the structure $(L_\mu^1)_{\red}$ is the graph of a bivector field. Now, if $(0,\alpha)\in (L_\mu^1)_\red$, then there exists  $\mathcal{Y}\in \mathcal{V}_\mu\cap \mathcal{C}_\mu$ such that $(\mathcal{Y},\rho_\mu^*(\alpha))\in L_\mu^1$ which, by Lemma \ref{L:lemmaY}, means that $\mathcal{Y}\in K_\mu$. Therefore $\rho_{\mu}^{\ast}\alpha = 0$, which implies $\alpha = 0$. Let us denote by $\omega_\mu^{1}$ the (non-degenerate) 2-form associated to the almost Dirac structure $(L_\mu^1)_{\red}$.  

    Using  Prop.~\ref{Prop:DiracFormsRltm}, we have $\rho_\mu^* \omega_\mu^{1}|_{\mathcal{U}_{\rho_\mu}} = \Omega_\mu^1 |_{\mathcal{U}_{\rho_\mu}}$, where $(\mathcal{C}_\mu, \Omega_\mu^1)$ is the pair of $L_{\mu}^{1}$ and  $\mathcal{U}_{\rho_\mu} := \{X\in \mathcal{C}_\mu \,:\, (X,\alpha)\in L_{\mu}^{1} \mbox{ and } \alpha|_{\mathcal{S}_\mu} = 0 \}$.  From   Lemma~\ref{L:lemmaY} we conclude that $\mathcal{U}_{\rho_\mu} = \mathcal{C}_\mu$ since, if $X\in \mathcal{C}_\mu $ but not in $\mathcal{U}_{\rho_\mu}$,  means that there exists $\mathcal{Y}\in \mathcal{S}_\mu$ such that the corresponding 1-form $\alpha$  of $X$ satisfies that $\alpha(\mathcal{Y}) \neq 0.$  Then, following the definition of a (almost) Dirac structure we get that $\langle (X, \alpha), (\mathcal{Y}, 0)\rangle = \alpha(\mathcal{Y}) \neq 0$ which is a contradiction. 
    Then, we proved \eqref{Eq:PullBacks1}. In particular, $T\rho_{\mu}(\mathcal{C}_{\mu}) = T(J_{\nh}^{-1}(\mu)/G)$.

    Next, we prove that $d\omega_\mu^{1}=0$.  For each $A = 1,2,3$, let $X_A\in \mathfrak{X}(J_{\nh}^{-1}(\mu)/G)$ and $\tilde{X}_A\in \Gamma( \mathcal{C}_\mu)$ such that $T\rho_{\mu}(\tilde{X}_A) = X_A$.  Since $TJ_{\nh}^{-1}(\mu) = \mathcal{C}_\mu\oplus \mathcal{W}_\mu$ for $\mathcal{W}_\mu = TJ^{-1}_\nh(\mu) \cap \mathcal{W}$, then
    \begin{equation*}
        \begin{split}
            d\omega_{\mu}^{1} (X_{1},X_{2},X_{3}) & = \textup{cyclic}\left[X_1(\omega_{\mu}^{1}(X_2,X_3)) - \omega_{\mu}^{1}([X_1,X_2],X_3)\right] \\ & = \textup{cyclic}\left[\tilde{X}_1(\Omega_{\mu}^{1}(\tilde{X}_2,\tilde{X}_3))- \rho_{\mu}^{\ast}\omega_{\mu}^{1}([\tilde{X}_1,\tilde{X}_2],\tilde{X}_3)\right] \\ & =  \textup{cyclic}\left[\tilde{X}_1(\Omega_{\mu}^{1}(\tilde{X}_2,\tilde{X}_3))- \Omega_{\mu}^{1}(P_{\mathcal{C}_\mu}([ \tilde{X}_1,\tilde{X}_2]),\tilde{X}_3)  \right],
        \end{split} 
    \end{equation*}
    where we use that $\rho_\mu^*\omega_{\mu}^{1}(P_{\mbox{\tiny{$\mathcal{W}_\mu$}}}([\tilde{X}_A,\tilde{X}_B]),\tilde{X}_C) = 0$ because $\mathcal{W} \subset \mathcal{V}$. 
    Now, from \eqref{Eq:Omega1-Coord} and using that $dp_i=0$ on $J^{-1}_\nh(\mu)$, we have that 
    $$
        (\Omega_{\subM}+B_1)|_{\mathcal{C}_\mu} = \iota_\mu^*\left(\mathcal{X}^{\alpha}\wedge dp_{\alpha}-p_{\alpha}d\mathcal{X}^{\alpha}\right)|_{\mathcal{C}_\mu} = \iota_\mu^*(d(p_\alpha\mathcal{X}^\alpha))|_{\mathcal{C}_\mu}. 
    $$
    Let us denote by $\theta$ the closed 2-form on $J^{-1}_\nh(\mu)$ given by $d(\iota_\mu^*(p_\alpha\mathcal{X}^\alpha))$. Observe that ${\bf i}_Z\theta = 0$ for $Z\in \Gamma(\mathcal{W}_\mu)$ since $d\tilde{\mathcal{X}}^\alpha|_{\mathcal{C}_{\mu}} = -\frac{1}{2}C_{\beta\gamma}^\alpha \tilde{\mathcal{X}}^\beta\wedge \tilde{\mathcal{X}}^\gamma |_{\mathcal{C}_{\mu}}$.   Finally, using that $\theta|_{\mathcal{C}_\mu} = \Omega_\mu^1|_{\mathcal{C}_\mu}$, we obtain that 
    \begin{equation*}
        \begin{split}
            d\omega_{\mu}^{1}(X_{1},X_{2},X_{3})
            & = \textup{cyclic}\left[\tilde{X}_1(\theta(\tilde{X}_2,\tilde{X}_3))- \theta(P_{\mathcal{C}_\mu}([ \tilde{X}_1,\tilde{X}_2]),\tilde{X}_3) \right] = \\
            & = \textup{cyclic}\left[\tilde{X}_1(\theta(\tilde{X}_2,\tilde{X}_3))- \theta([ \tilde{X}_1,\tilde{X}_2],\tilde{X}_3) \right] = d\theta(X_{1},X_{2},X_{3}) = 0.
        \end{split} 
    \end{equation*}

\end{proof}

\begin{remark} 
    From the last proof, we obtain also that  $\mathcal{U}_{\rho_\mu} = \mathcal{C}_\mu$ due to the particular definition of the nonholonomic bracket. But in general, this is not true.
\end{remark}

From Prop.~\ref{Prop:L1red_simplectic}, we conclude that the forward image of the almost Dirac structure $L_\mu^1$ is not only an almost Poisson structure as Corollary \ref{C:MomReduction} suggest, but it has a much stronger property:  it is a symplectic 2-form.  

Moreover, recall that the bivector field $\pi_1$  on $\M$ descends to a bivector field $\pi_\red^1$ on $\M/G$ as in \eqref{Eq:NHRedBrackets}. 

\begin{proposition}\label{Prop:SymplFoliation}
    The connected components of the symplectic manifolds $(J_{\emph\nh}^{-1}(\mu)/G,\omega_\mu^{1})$, for $\mu\in \textup{span}_{\mathbb{R}}\mathfrak{B}^*_{\mbox{\tiny{\emph{HGS}}}}$,  are the leaves of the foliation associated to $\pi_{\emph\red}^1$ and hence $\pi_{\emph\red}^1$ is Poisson.  
\end{proposition}
\begin{proof}
    Let $\alpha$ be a 1-form on $\M/G$ and let $\bar{X} := (\pi_\red^1)^\sharp(\alpha)$ and ${X} := \pi_1^\sharp(\rho^*\alpha)$. It is straightforward to see that $X$ and $\bar{X}$ are tangent to the leaves $J_\nh^{-1}(\mu)$ and $J_\nh^{-1}(\mu)/G$ respectively, using that $\pi_1^\sharp(dJ_i) = -(\xi_i)_\subM$. Since ${\bf i}_X \Omega_1|_\C = \rho^*\alpha|_\C$ then ${\bf i}_{X_\mu} \iota_\mu^*\Omega_1|_{\C_\mu} = \iota_\mu^* \rho^*\alpha|_{\C_\mu}$, where $X_\mu$ is a vector field on $J_\nh^{-1}(\mu)$ such that $T\iota_\mu(X_\mu) = X$.  Using \eqref{Eq:PullBacks1}, ${\bf i}_{X_\mu} \rho_\mu^*\omega_\mu^1|_{\C_\mu} = \iota_\mu^* \rho^*\alpha|_{\C_\mu} = \rho_\mu^* (\iota_\red^\mu)^* \alpha|_{\C_\mu}$, where $\iota_\red^\mu:J_{\nh}^{-1}(\mu)/G \to \M/G$ is the natural inclusion.   
    Finally, using that $T\rho_\mu (\C_\mu) = T(J_{\nh}^{-1}(\mu)/G)$ we obtain that ${\bf i}_{\bar{X}_\mu} \omega_\mu^1 = (\iota_\red^\mu)^*\alpha$, where $\bar{X}_\mu$ is the vector field on $J^{-1}_\nh(\mu)/G$ such that $T\iota_\red^\mu(\bar{X}_\mu) = \bar{X}$.
\end{proof}

\begin{remark}\label{R:Poisson} 
    The conclusion that the reduced bivector field $\pi_\red^1$ is Poisson (obtained by Props.~\ref{Prop:L1red_simplectic} and \ref{Prop:SymplFoliation}) recovers one of the main results in \cite{Balseiro2021}. It is worth noticing that the proof presented here is different from the one in \cite{Balseiro2021}, since here we do not use the formulas characterizing the failure of the Jacobi identity of the bracket but instead we prove it through the nonholonomic momentum bundle map reduction. 
\end{remark}

In order to finish the section, we present following diagram illustrating the ideas (c.f., Diag. \eqref{Diag:DiracMomRed})
\begin{equation}\label{Diag:RedPi1}
    \xymatrix{ (J^{-1}_\nh(\mu), L_\mu^1) \ar[r]^{\ \ \ \ \ \ \ \iota_\mu} \ar[d]^{\rho_\mu} & (\M,\pi_1) \ar[d]^{\rho} \\
(J^{-1}_\nh(\mu)/G, \omega^1_\mu)\ar[r]^{\ \ \iota_\red^\mu} & (\M/G, \pi_\red^1) }
\end{equation}

\subsubsection*{The nonholonomic momentum map reduction of $\pi_\subB$} \label{Sss:MMReductionPiB}

This section is built on the momentum map reduction of $\pi_1$, using the fact that $\pi_1$ and $\pi_\subB$ are gauge related by the 2-form $\mathcal{B}$ defined in \eqref{Eq:Bcal}, see Def.~\ref{Def:PiPB}.  Recall, from Remark~\ref{R:B1B}$(ii)$, that $\mathcal{B}$ is basic with respect to the principal bundle $\rho:\M\to \M/G$.    

The almost Poisson manifold $(\M,\pi_\subB)$ does not only admit a nonholonomic momentum map reduction, but also it describes the nonholonomic dynamics.  Therefore, the conclusion of this section  enlightens the study of the nonholonomic systems.   

From Prop.~\ref{Prop:Jnh-MomMap:MBM}, the nonholonomic momentum bundle map is a momentum map for $\pi_\subB$ associated to the basis $\mathfrak{B}_{\mbox{\tiny{HGS}}}$. 
Fixing $\mu = c_1\mu^{1}+\cdots+c_k\mu^k$, the submanifold $J^{-1}_\nh(\mu)$ and the bivector field $\pi_\subB$ are $G$-invariant, and then the backward image of $\pi_\subB$  by the inclusion map $\iota_{\mu} : J_{\nh}^{-1}(\mu) \rightarrow \mathcal{M}$ is the almost Dirac structure on $J_{\nh}^{-1}(\mu)$ given by 
\begin{equation*}
    L_{\mu}^{\subB} := {\iota_{\mu}^{\ast}}(L_{\subB}) = \left\{ (X,\iota_{\mu}^{\ast}\beta) \in \mathbb{T}(J_{\nh}^{-1}(\mu)) \,:\, \pi_\subB^\sharp(\beta)=  T\iota_{\mu}(X) \right\}, 
\end{equation*}
where $L_\subB := \textup{graph}(\pi_\subB^{\sharp})$, c.f., \eqref{Eq:Lmu1}.
The associated pair of $L_\mu^\subB$ is $(\mathcal{C}_\mu, \Omega_\mu^\subB)$ but, in this case, $\Omega_\mu^\subB|_{\mathcal{C}_\mu} = \iota_\mu^*(\Omega_{\subM} + B)|_{\mathcal{C}_\mu}$.  From here it is straightforward to see that $L_\mu^1$ and $L_\mu^\subB$ are gauge related by the 2-form ${\mathcal{B}}_{\mu} := \iota_{\mu}^{\ast}{\mathcal{B}}$ (see Remark~\ref{R:GaugeConsiderations}) and that $L_\mu^\subB$ is a genuine almost Dirac structure not given by a 2-form nor a bivector field.

\begin{remark}
    Since $\pi_\subB^\sharp(dH_\subM) = -X_\nh$ and due to the {\it conserved quantity assumption},  $X_\nh$ is tangent to the submanifold $J_\nh^{-1}(\mu)$. Therefore $L_\mu^\subB$ describes the dynamics: $(-X_{\mu}, dH_{\mu})\in L_{\mu}^{\subB}$, where $H_{\mu} := \iota_{\mu}^{\ast}H_{\subM}$ and $X_\mu = X|_{J^{-1}_\nh(\mu)}$.
\end{remark}

By Theorem~\ref{T:momentum-red}$(ii)$, the forward image of $L_\mu^\subB$ by the projection map $\rho_{\mu} : J_{\nh}^{-1}(\mu) \rightarrow J_{\nh}^{-1}(\mu)/G$ defines now the almost Dirac structure $(L^\subB_\mu)_{\red} := ({\rho_{\mu}})_{\ast}(L_{\mu}^{\subB})$ on $J_\nh^{-1}(\mu)/G$.   

Based in the gauge relation of $\pi_1$ and $\pi_\subB$, and by Prop.~\ref{Prop:Jnh-MomMap:MBM} and Theorem~\ref{Prop:L1red_simplectic}, we conclude:
\begin{theorem}\label{T:NH_momentum_red}
    Consider a nonholonomic system $(\mathcal{M},\pi_{\emph\nh},H_{\subM})$ with a $G$-symmetry satisfying the dimension assumption and the conserved quantity assumption. Let us denote by $\mathfrak{B}_{\mbox{\tiny{\emph{HGS}}}} =\{\xi_1, \dots, \xi_k\}$ the $Ad$-invariant basis of sections of horizontal gauge symmetries and $\mathfrak{B}^*_{\mbox{\tiny{\emph{HGS}}}} =\{\mu^1, \dots, \mu^k\}$ its dual $Ad^*$-invariant basis. For $\mu=c_i\mu^i$, where $c_i$ ($i=1,\dots,k$) are constants in $\mathbb{R}$, then the manifold $J_{\emph\nh}^{-1}(\mu) := \left\{m\in\mathcal{M}\,:\, \,J_{\emph\nh}(m) = \mu(\tau_{\subM}(m))\right\}$ is $G$-invariant and 
    \begin{enumerate}
        \item[$(i)$] (Level set restriction) The nonholonomic momentum bundle map $J_{\emph\nh} : \mathcal{M} \rightarrow \mathfrak{g}_{S}^{\ast}$ is a $Ad^*$-equivariant momentum map of $\pi_\subB$ associated to the basis $\mathfrak{B}_{\mbox{\tiny{\emph{HGS}}}}$ and the backward image of $\pi_\subB$ to $J_{\emph\nh}^{-1}(\mu)$ is the (regular) almost Dirac structure $L_\mu^\subB$.
        
        \item[$(ii)$] (Nonholonomic momentum reduction)  The forward image of $L_\mu^\subB$ by $\rho_\mu: J^{-1}_{\emph\nh}(\mu) \to J^{-1}_{\emph\nh}(\mu)/G$ defines an almost Dirac structure $(L^\subB_\mu)_{\emph\red}$ which is the graph of the nondegenerate 2-form 
        $$
            \omega_{\mu}^\subB = \omega_\mu^{1} +\bar{\mathcal{B}}_\mu,
        $$ 
        where $\bar{\mathcal{B}}_\mu$ is the 2-form on $J^{-1}_{\emph\nh}(\mu)/G$ such that $\rho_{\mu}^{\ast}\bar{\mathcal{B}}_{\mu} = {\mathcal{B}}_{\mu} := \iota_\mu^*{\mathcal B}$. Moreover, 
        $$\rho_\mu^*\omega_\mu^\subB |_{\mathcal{C}_\mu} = \iota_\mu^*(\Omega_\subM+B) |_{\mathcal{C}_\mu}.$$
        
        \item[$(iii)$] (The characteristic foliation of $\pi_{\emph\red}^{\subB}$) The connected components of the almost symplectic manifolds $(J_{\emph\nh}^{-1}(\mu)/G, \omega_{\mu}^{\subB})$ define the foliation associated to the bivector field $\pi_{\emph\red}^{\subB}$ on $\mathcal{M}/G$ (showing that $\pi_{\emph\red}^\subB$ admits an almost symplectic foliation).
        
        \item[$(iv)$]  (Dynamics)  The reduced nonholonomic dynamics $X_{\emph\red}$ on $\M/G$ is tangent to the submanifolds $J_{\emph\nh}^{-1}(\mu)/G$ and $X_{\emph\red}^\mu:=X_{\emph\red}|_{J_{\emph\nh}^{-1}(\mu)/G}$ is the Hamiltonian vector field on the almost symplectic manifold $(J_{\emph\nh}^{-1}(\mu)/G, \omega_\mu^\subB)$ associated to the reduced restricted Hamiltonian $H_{\emph\red}^{\mu}$ to $J_{\emph\nh}^{-1}(\mu)/G$, i.e., the nonholonomic dynamics restricts to the reduced level sets and it is given by ${\bf i}_{X_{\emph\red}^\mu} \omega_\mu^\subB = dH_{\emph\red}^\mu$.  
    \end{enumerate}
\end{theorem}
The following diagram clarifies the results of Theorem~\ref{T:NH_momentum_red}. 
\begin{equation*}
    \xymatrix{ (J^{-1}_\nh(\mu), L_\mu^\subB, H_{\mu}) \ar[r]^{\ \ \ \ \ \ \ \iota_\mu} \ar[d]^{\rho_\mu} & (\M,\pi_\subB, H_{\subM}) \ar[d]^{\rho} \\
(J^{-1}_\nh(\mu)/G, \omega_{\mu}^\subB = \omega_\mu^1+ \bar{\mathcal B}_\mu,H_\red^{\mu})\ar[r]^{\ \ \ \ \ \ \ \ \iota_\red^\mu} & (\M/G, \pi_\red^\subB,H_\red) }
\end{equation*}

\begin{remark} \label{R:twisted}
    A bivector field $\pi$ on $P $ is twisted Poisson if there exists a closed 3-form defined on $P$ such that $\tfrac{1}{2}[\pi, \pi]= \pi^\sharp(\phi)$. Such bivector fields are characterized by the fact that they have an integrable characteristic distribution endowed with an almost symplectic foliation. 
    In our case, Theorem \ref{T:NH_momentum_red} $(ii)$ and $(iii)$ implies that the bracket $\pi_\red^\subB$ is {\it twisted Poisson} by the 3-form $-d\mathcal{B}$.  In the same line as in Remark~\ref{R:Poisson}, we enforce that Theorem \ref{T:NH_momentum_red} not only recovers the main result in \cite{Balseiro2021} (arriving to it using a different path) but also it enlightens the structure of the almost symplectic leaves of the reduced bracket $\pi_\red^\subB$.
\end{remark}

\begin{remark} 
    If we relax the dimension assumption, the conclusions are weaker (see e.g., \cite{RecCortes2002, LuisG2016}).  More precisely, if the dimension assumption is not satisfied, we can still consider the dynamical gauge transformation defined by $B$ in \eqref{Eq:B1} and \eqref{Eq:Bcal} and have that $\pi_\subB^\sharp(dJ_i)= -(\xi_i)_\subM$ for each $i=1,\dots,k$. However this condition is not sufficient to ensure that $\pi_\red^\subB$ (or $\pi_\red^1$) admit an integrable characteristic distribution.  Therefore, the reduction by a momentum map will not give the associated (almost) symplectic foliation.  Instead, we may conclude that our reduction induces, in the reduced space, an almost Poisson bracket (in concordance with Corollary \ref{C:MomReduction} $(i)$). 
\end{remark}

\section{The identification of the leaves with the canonical symplectic manifold} \label{S:Identification}

In this section we are going to identify each almost symplectic leaf $(J_\nh^{-1}(\mu)/G, \omega^\subB_\mu)$, defined in Theorem~\ref{T:NH_momentum_red}, with the canonical symplectic manifold plus the magnetic term $\hat{\mathcal{B}}_\mu$, 
obtaining a result that looks like the one in symplectic reduction \cite{MarMiPerRat:Book}. The main consequence of such identification is that the reduced bracket $\pi_\red^\subB$ is foliated by leaves of {\it Chaplygin-type}:  $(T^*(Q/G), \omega_{\mbox{\tiny{can}}} + \hat{\mathcal{B}}_\mu)$, showing that these nonholonomic systems behave as Chaplygin systems on each leaf. 

The core of this section resides in showing that the symplectic manifolds $(J^{-1}_\nh(\mu)/G, \omega_\mu^1)$ (obtained by the reduction of $(\M,\pi_1)$ in Prop.~\ref{Prop:L1red_simplectic}) and $(T^*(Q/G), \omega_{\mbox{\tiny{can}}})$ are symplectomorphic. This is done in two steps: first we study the zero level set of the nonholonomic momentum map, and afterwards --by means of the shift trick-- the $\mu$-level, always having in mind that $\mu$ is a section and not a constant element of the dual of the Lie algebra.   Finally, using our usual strategy (of gauge related bivectors), we conclude that $(J_\nh^{-1}(\mu)/G, \omega^\subB_\mu)$ and $(T^*(Q/G), \omega_{\mbox{\tiny{can}}} + \hat{\mathcal{B}}_\mu)$ are diffeomorphic.

\subsection{The leaves $(J^{-1}_{\nh}(\mu)/G, \omega_{\mu}^{1})$ of $\pi_{\red}^{1}$}

Let us consider a $G$-invariant nonholonomic system $(\M,\pi_{\nh},H_{\subM})$ satisfying the {\it dimension} and the {\it conserved quantity} assumptions.   Recall that $\M := \kappa^\flat(D)$, where $\kappa$ is the kinetic energy metric and $D$ is the constraint distribution.   

Consider also the 2-form $B_1$ defined in \eqref{Eq:B1} and the reduction Diagram~\eqref{Diag:RedPi1} for $(\M,\pi_1)$.  

\subsubsection*{The identification at the zero-level set of $\pi_\red^1$} \label{Ss:IdZeroLevel}

Consider the splitting of $TQ=H\oplus S \oplus W$ given in \eqref{Eq:splittingTQ}, but, where the distribution $H$ is defined by $H= \Hor = D \cap S^\perp$  for which the orthogonal complement is taken with respect to the kinetic energy metric. Now, consider the nonholonomic system given by the Lagrangian $L$ but with the constraint distribution $\Hor$. This new system is $G$-invariant and, moreover, it is a $G$-Chaplygin system since 
\begin{equation}\label{Eq:0-level-splitting}
    TQ = \Hor \oplus V.
\end{equation}
Let us define the submanifold $\M_0:=\kappa^\flat(\Hor)\subset T^*Q$ and observe that, according to \cite{Bates1993}, the reduced dynamics on $\M_0/G$ is described by a 2-form $\omega_0$ on $\M_0/G$ so that 
\begin{equation*}
    \rho_0^* \omega_0 |_{\mathcal{C}_{\mbox{\tiny{\Hor}}}} = \Omega_{\subM_0} |_{\mathcal{C}_{\mbox{\tiny{\Hor}}}},
\end{equation*}
where $\rho_0:\M_0\to \M_0/G$ is the orbit projection, $\Omega_{\subM_0} = \iota_0^*\Omega_\subQ$ for $\iota_0:\M_0\to T^*Q$ the inclusion map and  
\begin{equation}\label{Def:C_H}
    \mathcal{C}_{\mbox{\tiny{\Hor}}} := \{X\in T\mathcal{M}_{0}; T\tau_{\subM_0}X\in \Hor\}.
\end{equation}
Following \cite{Koiller92} there is a diffeomophism $\phi_0:\M_0/G\to T^*(Q/G)$ such that 
\begin{equation} \label{Eq:0-level-Diffeo}
    \phi_0^* \omega_{\mbox{\tiny{can}}}  = \omega_0 + B_{\mbox{\tiny{$\langle\! J\!, \!\mathcal{K}_0\!\rangle$}}} ,
\end{equation}
where $B_{\mbox{\tiny{$\langle\! J\!, \!\mathcal{K}_0\!\rangle$}}}$ is the 2-form on $\M_0/G$ such that $\rho_0^*B_{\mbox{\tiny{$\langle\! J\!, \!\mathcal{K}_0\!\rangle$}}} = \langle J, \mathcal{K}_0\rangle$ for  $J:\M_0\to \g^*$ the (pull back to $\M_0$ of) the canonical momentum map and $\mathcal{K}_0$  the $\g$-valued 2-form on $\M_0$ given by $\mathcal{K}_0 = \tau_{\subM_0}^*K_0$ with $K_0$ the curvature associated to the principal connection associated to the splitting \eqref{Eq:0-level-splitting}.

\begin{remark}
    For completeness we give the explicit definition of the diffeomorphism following \cite{MarMiPerRat:Book} with the appropriate modifications:  consider the map  $\widetilde\phi_{0} : \M_0 \rightarrow T^{\ast}(Q/G)$ defined, for each $m_{q}\in (\M_0)_{q}$ and $v_q \in D_q$, by
    \begin{equation*}
        \langle \widetilde\phi_{0}(m_{q}),T\rho(v_{q})\rangle := \langle m_{q},v_{q}\rangle.  
    \end{equation*}
    Since $\widetilde\phi_{0} : \M_0 \rightarrow T^{\ast}(Q/G)$ is $G$-invariant, it induces the diffeomorphism ${\phi}_{0} :\M_0/G \rightarrow T^{\ast}(Q/G)$ such that $\widetilde\phi_{0} = \phi_{0}\circ \rho_{0}$. 
\end{remark}

\begin{lemma}\label{Lemma:0-levelChap} For the nonholonomic momentum bundle map $J_{\emph\nh}:\M\to \g_S^*$, we have that
    \begin{enumerate}
        \item[$(i)$] $\M_0 = J_{\emph\nh}^{-1}(0)$;  
        
        \item[$(ii)$] $\mathcal{C}_{\emph{\mbox{\tiny{\Hor}}}} \subset TJ^{-1}_{\emph\nh}(0) \cap \C = : \C_0$; and
        
        \item[$(iii)$] $\rho_0^* B_{\mbox{\tiny{$\langle\! J\!, \!\mathcal{K}_0\!\rangle$}}}  \,|_{\mathcal{C}_{\emph{\mbox{\tiny{\Hor}}}}} = \iota_0^* B_1\,|_{\mathcal{C}_{\emph{\mbox{\tiny{\Hor}}}}}$.
    \end{enumerate}
\end{lemma} 

\begin{proof}
    $(i)$ For $m = \kappa^{\flat}(v)\in \mathcal{M}$, we have $\langle J_{\nh}(m),\xi\rangle = \kappa(v,\xi_{Q})$ for all $\xi\in\g_S$.  Thus $m\in J_{\nh}^{-1}(0)$ if, and only if, $v\in D\cap S^\perp$. 

    $(ii)$ For $m\in J^{-1}_\nh(0)$, we have $\mathcal{S}_m \cap (\mathcal{C}_{\mbox{\tiny{\Hor}}})_m =\{0\}$, but $\mathcal{S}_m \subset (\mathcal{C}_0)_m$. 

    $(iii)$ Denote by $A_0:TQ\to \g$ the principal connection associated to the splitting \eqref{Eq:0-level-splitting} and recall that $V = S\oplus W$.  For  $m\in J_\nh^{-1}(0)$ and $X,Y\in (\mathcal{C}_{\mbox{\tiny{\Hor}}})_m$, we see that
    \begin{align*}
        \rho_0^* B_{\mbox{\tiny{$\langle\! J\!, \!\mathcal{K}_0\!\rangle$}}} (m)(X,Y) & = \langle J(m), \mathcal{K}_0\rangle (X,Y) = - \langle J, A_0([T\tau_{\subM_0} X,T\tau_{\subM_0} Y])\rangle = \\ 
        & = - \langle J, A_{\mbox{\tiny{$W$}}}([T\tau_{\subM_0} X,T\tau_{\subM_0} Y])\rangle,
    \end{align*}
    where, as usual, $\tau_{\subM_0}:\M_0\to Q$ is the canonical projection.  
    On the other hand, using the definition of $B_1$ given in \eqref{Eq:B1} and by $(ii)$,   we get that 
    $\iota_0^*B_1 (X,Y) = \iota^*_0\langle J(m) , \mathcal{K}_{\mbox{\tiny{$\mathcal{W}$}}}(X,Y)\rangle  = - \langle J, A_{\mbox{\tiny{$W$}}}([T\tau_{\subM_0} X,T\tau_{\subM_0} Y])\rangle$. 
\end{proof}

On the one hand, we have that $(\M_0/G, \omega_0)$ and $(T^*(Q/G), \omega_{\mbox{\tiny{can}}} + (\phi_0^*)^{-1}B_{\mbox{\tiny{$\langle\! J\!, \!\mathcal{K}_0\!\rangle$}}} )$ are symplectomorphic. On the other hand, from  Prop.~\ref{Prop:L1red_simplectic}, we recall the symplectic manifold $(J^{-1}_\nh(0)/G, \omega_0^1)$ obtained by the nonholonomic momentum bundle map reduction of $\pi_1$ at $\mu=0$.  

\begin{proposition} \label{Prop:0-level-identification}
    The diffeomorphism ${\phi}_{0} : J_{\emph\nh}^{-1}(0)/G \rightarrow T^{\ast}(Q/G)$ satisfies
    $$	
        \omega_{0}^{1} = \phi_{0}^{\ast}\, \omega_{\mbox{\tiny{\emph{can}}}}. 
    $$
\end{proposition}
\begin{proof} 
    Since $\mathcal{C}_{\mbox{\tiny{\Hor}}} \subset \mathcal{C}_0$, from \eqref{Eq:PullBacks1} we obtain that $(\rho_{0}^{\ast}\, \omega_{0}^{1})|_{\mathcal{C}_{{\mbox{\tiny{\Hor}}}}}	 =  (\iota_{0}^{\ast}(\Omega_{\subM}+B_{1}))|_{\mathcal{C}_{{\mbox{\tiny{\Hor}}}}} $.  Using  Lemma~\ref{Lemma:0-levelChap}$(ii)$, and also \eqref{Eq:0-level-Diffeo} we get that $$(\rho_{0}^{\ast}\omega_{0}^{1})|_{\mathcal{C}_{\mbox{\tiny{\Hor}}}}	 = (\Omega_{\subM_{0}} + \rho_0^* \circ \phi_0^* \langle J,\mathcal{K}_0\rangle      )|_{\mathcal{C}_{\mbox{\tiny{\Hor}}}} = \rho_{0}^{\ast}(\omega_0 + \phi_0^* \langle J,\mathcal{K}_0\rangle      )|_{\mathcal{C}_{\mbox{\tiny{\Hor}}}}  = (\rho_{0}^{\ast} \circ \phi_0^* \, \omega_{\mbox{\tiny{can}}}  )|_{\mathcal{C}_{\mbox{\tiny{\Hor}}}}.$$
\end{proof}

Last proposition confirms that the zero-leaf $(J_{\nh}^{-1}(0)/G, \omega_{0}^{1} )$ of the bivector field $\pi_\red^1$ is symplectomorphic to the canonical symplectic manifold $(T^*(Q/G),\omega_{\mbox{\tiny{can}}})$.

\subsubsection*{The identification at the $\mu$-level and the Shift-trick} \label{Ss:IdMuLevel}

Now, for $\mu^i\in \mathfrak{B}_{\mbox{\tiny{HGS}}}^*$ and $\mu=c_i\mu^i$ we consider the symplectic leaves  $(J_{\nh}^{-1}(\mu)/G, \omega_{\mu}^{1})$ obtained through reduction of $(\M,\pi_1)$ in Prop.~\ref{Prop:L1red_simplectic}.
Next, we show that each manifold $(J_{\nh}^{-1}(\mu)/G, \omega_{\mu}^{1})$ is symplectomorphic to $(T^{\ast}(Q/G),\omega_{\mbox{\tiny{can}}})$.  On the one hand, the identification of the manifolds is similar to the Hamiltonian case using the shift trick \cite{MarMiPerRat:Book}, but having into account the fact that our starting manifold  $(\M, \pi_1)$  is an almost Poisson. On the other hand, in our case, $\mu$ is a section of the bundle $\g_S^*$ and not just an element of the dual of a Lie algebra.  The 2-form $B_1$ plays a fundamental role here in order to prove that  $(J_{\nh}^{-1}(\mu)/G, \omega_{\mu}^{1})$ is symplectomorphic to the canonical symplectic manifold (without any magnetic term). 

As usual we start with a $G$-invariant nonholonomic system $(\M, \pi_\nh, H_\subM)$ satisfying the dimension and the conserved quantity assumptions.  Recall that $\mathfrak{B}_{\mbox{\tiny{HGS}}}$ and $\mathfrak{B}_{\mbox{\tiny{HGS}}}^*$ are the basis of global sections of horizontal gauge symmetries and its dual basis respectively (as in \eqref{Eq:BasisHGS}).  
Consider $\mu = c_{i}\mu^{i}\in C^{\infty}(Q,\mathfrak{g}_{S}^{\ast})$, where $\mu^i\in \mathfrak{B}^*_{\mbox{\tiny{HGS}}}$ and $c_{i}$ are constants in $\mathbb{R}$. Inspired in \cite{MarMiPerRat:Book}, let us define the diffeomorphism $\textup{Shift}_\mu:T^*Q \to T^*Q$ given, for each $\alpha_q\in T_q^*Q$,  by
$$
    \textup{Shift}_\mu(\alpha_q) = \alpha_q - \alpha_q^\mu,
$$
where $\alpha_q^\mu = \langle \mu, A_{\mbox{\tiny{$S$}}} \rangle _q = \langle \mu(q), A_{\mbox{\tiny{$S$}}}(q)\rangle \in T_q^*Q$ for $A_{\mbox{\tiny{$S$}}}: TQ \to \g_S$ the bundle map that behaves as a connection but only counting the $S$-part of the vertical distribution: $A_{\mbox{\tiny{$S$}}}= Y^i\otimes \xi_i$ with $Y^i$ defined in \eqref{Eq:B1} and $\xi_i\in \mathfrak{B}_{\mbox{\tiny{HGS}}}$ ($A_{\mbox{\tiny{$S$}}}$ was defined already in the proof of Lemma~\ref{Lemma:0-levelChap}).
 
We consider now the usual splitting of $TQ=H\oplus S\oplus W$ as in \eqref{Eq:splittingTQ} but where 
$$
    H=D\cap S^\perp \qquad \mbox{and} \qquad W=V\cap S^\perp,
$$
where the perpendicular space is taken with respect to the kinetic energy metric. 

\begin{lemma} \label{Lemma:Shift}
    The bundle map $\textup{Shift}_\mu^\subM := \textup{Shift}_\mu|_\M : \M \to \M$ is a diffeomorphism that satisfies  
    $$      
        (\textup{Shift}_\mu^\subM)^* (\Omega_\subM + B_1)|_{\mathcal{C}}= (\Omega_\subM + B_1)|_{\mathcal{C}}.
    $$
\end{lemma}

\begin{proof} 
   For $\alpha_q\in \M$, $\textup{Shift}_\mu(\alpha_q) \in \M$ if and only if $\alpha_q^\mu\in \M$.  But this fact is true since using the definition of $H$ and $W$, and the corresponding basis \eqref{Eq:BasisTQ}, we get that $\alpha_q^\mu = c_i Y^i =  c_i \kappa^\flat (\kappa^{ij}Y_j)$, proving that $\alpha_q^\mu\in \kappa^\flat(D) = \M$.    

   Now, as it was observed in \cite{PaulaGarciaToriZuccalli}, for $\alpha\in \M$ and $X\in T\M$, $(\textup{Shift}_\mu^\subM)^*\Theta_\subM(\alpha)(X)= \langle \alpha, T\tau_\subM(X) \rangle - \langle \alpha_\mu^q,T\tau_\subM(X) \rangle = \Theta_\subM (X) - c_i \mathcal Y^i(X)$,  for $\mathcal{Y}^i = \tau_\subM^*Y^i$ and where we are using that $T\textup{Shift}_\mu^\subM(X) = X$. 
   Therefore $(\textup{Shift}_\mu^\subM)^*\Omega_\subM = \Omega_\subM +c_{i}d\mathcal{Y}^{i}$. Next, for $X,Y\in \Gamma(\mathcal{C}) $, we have 
    \begin{align*}
        (\textup{Shift}_\mu^\subM)^*B_1(X,Y) & = (\langle J(\textup{Shift}_\mu^\subM(\alpha)),\mathcal{K}_{\mbox{\tiny{$\mathcal{W}$}}} \rangle + J_{i}(\textup{Shift}_\mu^\subM(\alpha)) d\mathcal{Y}^{i}))(T\textup{Shift}_\mu^\subM(X),T\textup{Shift}_\mu^\subM(Y)) =	\\
        & = \langle J(\alpha),\mathcal{K}_{\mbox{\tiny{$\mathcal{W}$}}} \rangle(X,Y)+(J_{i}(\alpha) -c_{i})d\mathcal{Y}^{i}(X,Y) =  B_1(X,Y)-c_{i}d\mathcal{Y} ^{i}(X,Y),
    \end{align*}
    where we are using the fact that $W = S^{\perp}\cap V$.  Finally, combining the first and second step we obtain the desired result.  
\end{proof}

If $\alpha_q\in J_\nh^{-1}(\mu)$, then $\textup{Shift}_\mu^\subM(\alpha_q) \in J^{-1}_\nh(0)$ and therefore, we have the following well defined bundle map 
$$
    \textup{shift}_{\mu} := \left. \textup{Shift}^\subM_{\mu}\right|_{J_{\nh}^{-1}(\mu)} : J_{\nh}^{-1}(\mu) \rightarrow J_{\nh}^{-1}(0).  
$$
\begin{lemma}
    The diffeomorphism $\textup{shift}_{\mu} : J_{\emph\nh}^{-1}(\mu) \rightarrow J_{\emph\nh}^{-1}(0)$ is equivariant and therefore it induces the reduced diffeomorphism 
    $$
        \overline{\textup{shift}}_{\mu} : J_{\emph\nh}^{-1}(\mu)/G \rightarrow J_{\emph\nh}^{-1}(0)/G,\quad \textup{such that}\quad \overline{\textup{shift}}_{\mu}\circ \rho_{\mu} = \rho_{0}\circ \textup{shift}_{\mu}. 
    $$
\end{lemma}
\begin{proof}
    The equivariance of $\textup{shift}_{\mu} : J_{\nh}^{-1}(\mu) \rightarrow J_{\nh}^{-1}(0)$ means that for $\alpha\in \M$ and $g\in G$,
    \begin{equation*}
        T^{\ast}\psi_{g}(\textup{shift}_{\mu}(\alpha)) = \textup{shift}_{\mu}(T^{\ast}\psi_{g}(\alpha)), 
    \end{equation*}
    where $\psi_{g} : Q\rightarrow Q$ is the diffeomorphism associated with the $G$-action on $Q$.  By the definition of the map $\textup{shift}_\mu$ it remains to see that $T^{\ast}\psi_{g}(\alpha^\mu_{q}) = \alpha_q^\mu$. But, using that $\mathfrak{B}_{\mbox{\tiny{HGS}}}$ and $\mathfrak{B}^*_{\mbox{\tiny{HGS}}}$ are $Ad$ and $Ad^*$-invariant basis respectively, we obtain that 
    $
    T^{\ast}\psi_{g}(\alpha_{gq}^\mu) = T^{\ast}\psi_{g}(\langle \mu,A_{\mbox{\tiny{$S$}}} \rangle_{gq}) = \langle \mu(gq), (A_{\mbox{\tiny{$S$}}})_{gq}\circ T\psi_{g}\rangle = \langle \mu(gq), Ad_{g}((A_{S})_{q})\rangle = \langle Ad_{g^{-1}}^{\ast}(\mu(gq)), (A_{\mbox{\tiny{$S$}}})_{q}\rangle = \langle \mu, A_S\rangle_{q}$.
\end{proof}

The following diagram clarifies the maps: 
\begin{equation} \label{Diag:figuraShiftTrick}
    \xymatrix{\ar[d]^{\textup{Shift}^\subM_{\mu}} \mathcal{M}    &  \ar[l]_{\iota_{\mu}}\ar[d]^{\textup{shift}_{\mu}} J_{\nh}^{-1}(\mu)\ar[r]^{\rho_{\mu}} &   J_{\nh}^{-1}(\mu)/G \ar[d]^{\overline{\textup{shift}}_{\mu}} \ar[r]^{\iota_{\red}^\mu} & \M/G \\
    \mathcal{M}	 & \ar[l]_{\iota_{0}} J_{\nh}^{-1}(0)\ar[r]^{\rho_{0}}  &  J_{\nh}^{-1}(0)/G  \ar[r]^{\iota_{\red}^0} & \M/G}
\end{equation}

Finally from Prop.~\ref{Prop:0-level-identification} we recall the diffeomorphism $\phi_0 : J^{-1}_\nh(0)/G \to T^*(Q/G)$ and we obtain:   

\begin{theorem} \label{T:P1-mu-identification}  
    For each $\mu\in\g_S^*$ given by $\mu=c_i\mu^i$ with $c_i\in\mathbb{R}$, the symplectic manifolds $(J_{\emph\nh}^{-1}(\mu)/G, \omega_{\mu}^{1})$ are symplectomorphic to $(T^*(Q/G), \omega_{\emph{\mbox{\tiny{can}}}})$. More precisely,  the diffeomorphism ${\phi}_{\mu} := {\phi}_{0}\circ \overline{\textup{shift}}_{\mu} : J_{\emph\nh}^{-1}(\mu)/G \rightarrow T^{\ast}(Q/G)$  satisfies
    $$
        {\phi}_{\mu}^{\ast}\, \omega_{\emph{\mbox{\tiny{can}}}} = \omega_{\mu}^{1}.
    $$
\end{theorem}
\begin{proof} 
    Using Prop.~\ref{Prop:L1red_simplectic} and due to Diag.~\eqref{Diag:figuraShiftTrick} and Prop.~\ref{Prop:0-level-identification} we see that   
    \begin{align*}
        \rho_{\mu}^{\ast}\circ{\phi}_{\mu}^{\ast}\, \omega_{\mbox{\tiny{can}}} |_{\mathcal{C}_{\mu}}	& = \rho_{\mu}^{\ast} \circ \overline{\textup{shift}}_{\mu}^{\ast}\circ {\phi}_{0}^{\ast}\, \omega_{\mbox{\tiny{can}}}|_{\mathcal{C}_{\mu}} = (\overline{\textup{shift}}_{\mu}\circ \rho_{\mu})^{\ast}\, \omega_{0}^{1} |_{\mathcal{C}_{\mu}} = (\rho_{0}\circ \textup{shift}_{\mu})^{\ast}\, \omega_{0}^{1}|_{\mathcal{C}_{\mu}} = \\
        & = \textup{shift}_{\mu}^{\ast}\circ \iota_{0}^{\ast}(\Omega_\subM +B_1) |_{\mathcal{C}_{\mu}} = \iota_{\mu}^{\ast}\circ (\textup{Shift}^\subM_{\mu})^{\ast} (\Omega_\subM +B_1)  |_{\mathcal{C}_{\mu}} = \iota_{\mu}^{\ast}(\Omega_\subM +B_1) |_{\mathcal{C}_{\mu}}.	
    \end{align*}
\end{proof}

\begin{remark}(The canonical nature of $\pi_1$) Given a nonholonomic system as in Theorem \ref{T:P1-mu-identification}, then 
    \begin{enumerate}
        \item[$(i)$]  the Poisson bivector field $\pi_\red^1$ has a canonical nature since the symplectic foliation is (symplectomorphic to) the canonical symplectic manifolds $(T^*(Q/G), \omega_{\mbox{\tiny{can}}})$. 
        
        \item[$(ii)$] Consider now just the $G$-invariant almost Poisson bracket $(\M,\pi_\nh)$ with the dimension assumption and observe that if the bundle $\mathfrak{g}_S$ is trivial, then for any choice of an $Ad$-invariant global basis $\{\xi_1,\dots,\xi_k\}$ of sections of $\mathfrak{g}_S$ there exists a Poisson bivector field $\pi_\red^{1}$ whose symplectic leaves are given by the common level sets of the functions $J_i=\langle J_\nh, \xi_i\rangle$ and which are also symplectomorphic to the canonical symplectic manifold.
    \end{enumerate}
\end{remark}


\subsection{The leaves $(J^{-1}_\nh(\mu)/G, \omega^\subB_\mu)$ of $\pi_\red^\subB$ and the Chaplygin-type foliation}

Let $(\M,\pi_\nh, H_\subM)$ be a $G$-invariant nonholonomic system with the dimension and the conserved quantity assumption. Then we consider the almost symplectic leaves $(J^{-1}_\nh(\mu)/G, \omega^\subB_\mu)$ obtained by the reduction of $(\M,\pi_\subB)$ in Theorem~\ref{T:NH_momentum_red}.
Following the same strategy as in Section~\ref{Ss:Nonholonomic_momentum_reduction}, we use the fact that $(\M/G,\pi_\red^1)$ and $(\M/G,\pi_\red^\subB)$ are gauge related by the 2-form $\bar{\mathcal{B}}$ (Remark \ref{R:B1B}$(ii)$), 
to conclude that the corresponding leaves  $(J^{-1}_\nh(\mu)/G, \omega^1_\mu)$ and  $(J^{-1}_\nh(\mu)/G, \omega^\subB_\mu)$ differ from the 2-form $\bar{\mathcal{B}}_\mu$ where $\rho_\mu^*\bar{\mathcal{B}}_\mu= \iota_\mu^*{\mathcal B}$, see Remark~\ref{R:GaugeConsiderations}$(ii)$. 

\begin{theorem}\label{T:PB-mu-identification} 
    Consider the nonholonomic system with a $G$-symmetry, satisfying the dimension and the conserved quantity assumptions, described by the triple $(\M,\pi_\subB,H_\subM)$.   The almost symplectic leaf $(J^{-1}_{\emph\nh}(\mu)/G, \omega^\subB_\mu)$ obtained through a momentum map reduction in Theorem \ref{T:NH_momentum_red} is diffeomorphic to  
    $$
        (T^*(Q/G), \omega_{\emph{\mbox{\tiny{can}}}} + \hat{\mathcal{B}}_\mu),
    $$ 
    where $\hat{\mathcal{B}}_\mu$ is the 2-form on $T^*(Q/G)$ such that $\phi_\mu^*\hat{\mathcal{B}}_\mu = \bar{\mathcal{B}}_\mu$, where ${\phi}_{\mu} : J_{\emph\nh}^{-1}(\mu)/G \rightarrow T^{\ast}(Q/G)$  is the diffeomorphism given by ${\phi}_{\mu} := {\phi}_{0}\circ \overline{\textup{shift}}_{\mu}$. 
\end{theorem}

\begin{proof}
    On the one hand, by Theorem~\ref{T:P1-mu-identification}, we have that $\phi_\mu^*\omega_{\mbox{\tiny{can}}} = \omega_\mu^1$ and, on the other hand, by Theorem \ref{T:NH_momentum_red},  $\omega^\subB_\mu = \omega_\mu^1 +\bar{\mathcal{B}}_\mu$ for $\rho_\mu^*\bar{\mathcal{B}}_\mu= \iota_\mu^*{\mathcal B}$.  
    Therefore, we conclude that $\omega_{\mu}^{\subB} = {\phi}_{\mu}^{\ast}\, \omega_{\mbox{\tiny{can}}} + \bar{\mathcal{B}}_\mu.$
\end{proof}

The next Corollary shows the Chaplygin nature of nonholonomic systems verifying the dimension and conserved quantity assumptions. Note that these nonholonomic systems are not of Chaplygin-type necessarily.  

\begin{corollary} [{\it The Chaplygin-type foliation}]
    Consider a nonholonomic system with a $G$-symmetry, satisfying the dimension and the conserved quantity assumptions. 
    \begin{enumerate}
        \item[$(i)$] The leaves of the almost symplectic foliation associated to the bivector field $\pi_{\emph\red}^\subB$ are diffeomorphic to $(T^*(Q/G), \omega_{\emph{\mbox{\tiny{can}}}} + \hat{\mathcal{B}}_\mu)$.
        
        \item[$(ii)$] The reduced nonholonomic system is described by the foliation of ``Chaplygin-type'' given by $(T^*(Q/G), \omega_{\mbox{\tiny\emph{can}}} + \hat{\mathcal{B}}_\mu)$, behaving as a Chaplygin system on each leaf. 
    \end{enumerate}
\end{corollary}

\begin{remark}
    Our procedure generalizes the momentum map reduction proposed in \cite{PaulaGarciaToriZuccalli} that treats Chaplygin systems with an extra symmetry, since in the present paper, we do not consider Chaplygin systems. As a consequence, we are forced to work with the  nonholonomic momentum bundle map $J_\nh:\M\to \g_S^*$, which is not the case in \cite{PaulaGarciaToriZuccalli}.
\end{remark}

\paragraph{Coordinate expression for $\hat{\mathcal{B}}_\mu$.}

Next we will study the 2-form $\hat{\mathcal{B}}_\mu$ in coordinates adapted to the basis \eqref{Eq:BasisTQ} and \eqref{Eq:localCoordTM} where  $\Hor=S^\perp \cap D$ and $W = S^\perp\cap V$.  Therefore, for $\mu=c_i\mu^i$, on $J^{-1}(\mu)$ for each $i=1,\dots,k$, we have that $p_i= \kappa_{ij}v^j=c_i$ and $p_a = \kappa_{a\alpha}v^\alpha=\kappa_{a\alpha}\kappa^{\alpha\beta}p_\beta$.  

Then, on $T^*(Q/G)$ we consider coordinates  $p_\alpha$ on $T_{\bar{q}}^*(Q/G)$ ($\bar{q}\in Q/G$) and 
using the coordinate expression from \cite[Sec.3.4]{Balseiro2021}, we obtain that 
\begin{equation*}
    \hat{\mathcal{B}}_\mu = \tfrac{1}{2}\left[ p_\gamma \kappa^{\gamma\delta} \kappa_{\delta a} C^a_{\alpha\beta} + c_i(C^i_{\alpha\beta} + \kappa^{ij}\kappa_{\alpha A}C^A_{j\beta}) \right] \bar{\mathcal X}^\alpha \wedge \bar{\mathcal X}^\beta,
\end{equation*}
where $\bar{\mathcal{X}}^\alpha=\tau^*_{\mbox{\tiny{$T^*\!(Q/G)$}}} \bar{X}^\alpha$ are the 1-forms on $T^*(Q/G)$, so that $\bar{X}^\alpha$ are 1-forms on $Q/G$ such that $\rho_\subQ^*\bar{X}^\alpha = {X}^\alpha$
for $\rho_\subQ : Q\to Q/G$ the orbit projection and $\tau_{\mbox{\tiny{$T^*(Q/G)$}}} : T^*\!(Q/G) \to Q/G$ the canonical projection.

From here, we observe that $\hat{\mathcal{B}}_0= \tfrac{1}{2}\left[ p_\gamma \kappa^{\gamma\delta}\kappa_{\delta a} C^a_{\alpha\beta} \right] \bar{\mathcal X}^\alpha \wedge \bar{\mathcal X}^\beta$.

\section{Almost symplectic structure at the 0-level of linear first integrals}\label{S:general-0level}

Building on Section~\ref{Ss:IdZeroLevel}, we study the common 0-level sets of first integrals that are {\it not} necessary horizontal gauge momenta. This section is inspired in the example describing the nonholonomic system given by a ball rolling on a fixed sphere (studied in e.g., \cite{BorisovFedorov1995, BorisovMamaev2002,Jovanovic2010,Tsiganov2011})   that has a linear conserved quantity that is not a horizontal gauge momentum.
More precisely, consider a nonholonomic system with a $G$-symmetry satisfying the dimension assumption. In this section, we show that, if the system admits $k=\textup{rank}(S)$ $G$-invariant first integrals  (linear on the fibers),  the zero-level sets of these functions are also identified with a Chaplygin-type leaf (i.e., with the canonical symplectic manifold plus a magnetic term). 

We start this section with a proper definition of the first integrals that we are interested.

\begin{definition}\label{Def:D-momentum} \cite{Fasso2007}
    Let $(\M,\pi_\nh,H_\subM)$ be a nonholonomic system. The function $F$ on $\M$ is a {\it $D$-momentum} if 
    \begin{enumerate}
        \item[$(i)$] There exists a non-vanishing vector field $X$ on $Q$ taking values in $D$ (i.e., $X\in\Gamma(D)$), such that, for each $m\in \M \subset T^*Q$, $F(m) =  \langle m, X(q) \rangle$ where $\tau_\subM(m)=q$, for  $\tau_\subM:\M\to Q$. 
    
        \item[$(ii)$] $F$ is a first integral of the nonholonomic dynamics: $X_{\nh}(F) = 0$. 
    \end{enumerate}
    In this case, the vector field $X$ is called a \textit{generator} of $F$.    
\end{definition}

\begin{remark} 
    \begin{enumerate}
        \item[$(i)$] From Definition \ref{Def:D-momentum}, we see that a $D$-momentum $F$ on $\M\subset T^*Q$ is a linear function  on the fibers of the bundle $\tau_\subM:\M\to Q$.
        
        \item[$(ii)$] A $D$-momentum can be written also as $F(m)={\bf i}_{\mathcal{X}}\Theta_\subM(m)$ where $\mathcal{X}$ is a $\tau_\subM$-projectable vector field on $\M$ taking values in $\mathcal{C}$. In contrast with the definition of a {\it horizontal gauge momentum}, $\mathcal{X}$ is not necessarily vertical, i.e., $\mathcal{X}(m)$ might not be in $\mathcal{S}_m$. 
    \end{enumerate} 
\end{remark}

Analogously as Lemma~\ref{L:Ad-basis}, in this case we have that 

\begin{lemma}
    Given $F_1, F_2$ functionally independent $D$-momenta of the nonholonomic system $(\M, \{\cdot, \cdot\}_{\emph\nh}, H_\subM)$ then the associated generators $X_1$ and $X_2$ are independent vector fields on $Q$.
\end{lemma}

If the nonholonomic system $(\M, \{\cdot, \cdot\}_\nh, H_\subM)$ has $\{F_1,\dots,F_p\}$ functionally independent {\it $D$-momenta}, then the associated generators $\{X_1,\dots,X_p\}$ define a (constant rank) distribution $\widetilde S$ on $Q$ given by 
\begin{equation*}
    \widetilde{S} = \textup{span} \{X_1,\dots,X_p\}.
\end{equation*}
Moreover, we can define the smooth function $\mathcal{F} := (F_1,\dots,F_p) : \M\to \mathbb{R}^p$ and the common 0-level set of the $D$-momenta is represented as the smooth submanifold  $\mathcal{F}^{-1}(0)\subset \M$.

\begin{lemma}\label{L:tildeH}
    Let $(\M, \pi_{\emph\nh}, H_\subM)$ be a $G$-invariant nonholonomic system satisfying the dimension assumption and having $k:=\textup{rank}(S)$ functionally independent and $G$-invariant $D$-momenta. If $\widetilde S^\perp\cap S = \{0\}$ (where the orthogonal complement is taken with respect to the kinetic energy metric $\kappa$), then the constant rank distribution  $\widetilde{H} := D\cap \tilde{S}^{\perp}$ satisfies that 
    $$
        TQ= \widetilde{H}\oplus V.
    $$
    Moreover,  $\mathcal{F}^{-1}(0) = \kappa^{\flat}(\widetilde{H})\subset \M$ is a $G$-invariant submanifold.
\end{lemma}
\begin{proof}
    By definition of $D$-momentum,   $\widetilde{S}\subset D$ and therefore, $D=\widetilde{S} \oplus \widetilde{H}$ showing that $\widetilde{H}$ has constant rank.  Moreover, using that $\widetilde S^\perp\cap S = \{0\}$ we can check that $\widetilde{H}\cap V =\{0\}$ which implies that $TQ=\widetilde{H} \oplus V$.  

    Finally, we see that if $m \in \mathcal{F}^{-1}(0)$ then $F_i(m)=0$ for all $i=1,\dots,k$, which means that $0 = \langle m ,X_i\rangle = \kappa(v, X_i)$ for $m=\kappa^\flat(v)$.  Moreover  the submanifold $\mathcal{F}^{-1}(0)\subset \M$ is invariant whenever each of the ${D}$-momenta $F_i$ are $G$-invariant.  
\end{proof}

Therefore, under the hypothesis of Lemma~\ref{L:tildeH}, we can consider the $G$-Chaplygin system given by the Lagrangian $L$ and the constraint distribution $\widetilde{H}$. Following the theory of Chaplygin systems \cite{Koiller92} and Sec.~\ref{Ss:IdZeroLevel}, for $\M_0=\kappa^\flat(\widetilde{H})$ the reduced system takes place in  $\M_0/G$ which is diffeomorphic to $T^*(Q/G)$. 

Considering the principal bundle $\rho_0: \M_0\to \M_0/G$, and analogously as in \eqref{Eq:0-level-Diffeo}, we define the 2-form 
$B_{\mbox{\tiny{$\langle J\!,\! \!\widetilde{\mathcal{K}}\rangle$}}}$ on $\M_0/G$ such that \begin{equation*}
    \rho_0^*B_{\mbox{\tiny{$\langle J\!,\! \!\widetilde{\mathcal{K}}\rangle$}}} = \langle J, \widetilde{\mathcal{K}}\rangle,
\end{equation*}
for  $J:\M_0\to \g^*$ the (pull back to $\M_0$) of the canonical momentum map and $\widetilde{\mathcal{K}}$  the $\g$-valued 2-form on $\M_0$ given by $\widetilde{\mathcal{K}} = \tau_{\subM_0}^*\widetilde{K}$ with $\widetilde{K}$ the curvature associated to the principal connection $TQ=\widetilde{H}\oplus V$.

Let $\mathcal{C}_{\mbox{\tiny{$\widetilde{H}$}}}$ be the distribution on $\M_0$ given by the vector fields $X$ on $\M_0$ such that $T\tau_0(X)\in \Gamma(\widetilde{H})$, for $\tau_0:\M_0\to Q$, analogously as in \eqref{Def:C_H}.

\begin{theorem} \label{T:general0-level}
    Consider the nonholonomic system $(\M, \pi_{\emph\nh}, H_\subM)$ with a $G$-symmetry satisfying the dimension assumption. If the system has $k:=rank(S)$ functionally independent, $G$-invariant $D$-momenta such that $S\cap \tilde S^\perp = \{0\}$, then
    \begin{enumerate}
        \item[$(i)$] the reduced manifold $\mathcal{F}^{-1}(0)/G \subset \M/G$ is given by the common 0-level sets of the reduced $D$-momenta $\bar{F}_i$, where $\rho_0^*\bar{F}_i = F_i$ for $i=1,\dots,k$.
        
        \item[$(ii)$] There is an almost symplectic 2-form $\widetilde{\omega}$ on $\mathcal{F}^{-1}(0)/G$ so that $\rho_0^*\widetilde{\omega} |_{\mathcal{C}_{\mbox{\tiny{$\widetilde{H}$}}}} = \Omega_\subM|_{\mathcal{C}_{\mbox{\tiny{$\widetilde{H}$}}}}$, and for which the restriction $X_\emph\red^0$ of the reduced nonholonomic vector field $X_\emph\red$ to $\mathcal{F}^{-1}(0)/G$ is a Hamiltonian vector field with respect to restricted reduced Hamiltonian function $H_{\emph\red}^0$. That is, if $X_{\emph\red}^0 = X_{\emph\red}|_{\mathcal{F}^{-1}(0)/G}$ then
        $$
            {\bf i}_{X_{\emph\red}^0} \widetilde \omega = dH_{\emph\red}^0.
        $$
        
        \item[$(iii)$] There is a diffeomorphism $\phi: \mathcal{F}^{-1}(0)/G \to T^*(Q/G)$ such that \  $\widetilde\omega = \phi^{\ast} \omega_{\mbox{\tiny{\emph{can}}}}  - B_{\mbox{\tiny{$\langle J\!,\! \!\widetilde{\mathcal{K}}\rangle$}}}.$
    \end{enumerate}
\end{theorem}
\begin{proof}
    If we consider the nonholonomic system given by the constraint distribution $\widetilde{H}$ and the Lagrangian $L$, we have a $G$-Chaplygin system since $TQ=\widetilde{H}\oplus V$ (due to the condition $\widetilde{S}^\perp\cap S =\{0\}$).  Since $\mathcal{F}^{-1}(0) = \kappa^\flat(\widetilde{H})$, according to \cite{Koiller92}, the submanifold $\mathcal{F}^{-1}(0)/G$ is diffeomorphic to the cotangent manifold $T^{\ast}(Q/G)$, and the  reduced nonholonomic vector field restricted to $\mathcal{F}^{-1}(0)/G$, denoted by $\tilde{X}_{\red}^{0}$, satisfies that ${\bf i}_{\tilde{X}_\red^0} ( \omega_{\mbox{\tiny{can}}} - \hat{B}_{\mbox{\tiny{$\langle J\!,\! \!\widetilde{\mathcal{K}}\rangle$}}} ) = dH_\red^0,$ where $H_\red^0$ is the restricted reduced Hamiltonian and $\hat{B}_{\mbox{\tiny{$\langle J\!,\! \!\widetilde{\mathcal{K}}\rangle$}}}$ is the 2-form on $T^*(Q/G)$ such that $\phi^*\hat{B}_{\mbox{\tiny{$\langle J\!,\! \!\widetilde{\mathcal{K}}\rangle$}}}= B_{\mbox{\tiny{$\langle J\!,\! \!\widetilde{\mathcal{K}}\rangle$}}}$.
\end{proof}

In this case, we obtain that, on the zero-level set of the conserved quantities, the dynamics is described by a almost symplectic manifold $(\mathcal{F}^{-1}(0)/G , \widetilde\omega)$ which is diffeomorphic to 
$$
    (T^*(Q/G),  \omega_{\mbox{\tiny{can}}} - \hat{B}_{\mbox{\tiny{$\langle J\!,\! \!\widetilde{\mathcal{K}}\rangle$}}} ),
$$
showing that we have also a Chaplygin leaf.  Moreover, the term $\hat{B}_{\mbox{\tiny{$\langle J\!,\! \!\widetilde{\mathcal{K}}\rangle$}}} $ is exactly the {\it gyroscopic term} defined in \cite{JCLuis2020} to describe Chaplygin systems.  

\begin{remark} 
    When the first integrals are horizontal gauge momenta (i.e., $\tilde S = S$), then the horizontal space $\widetilde{H}$ is just $S^\perp \cap D$ and then $\kappa^\flat(\widetilde{H}) = J_\nh^{-1}(0)$ recovering Sec.~\ref{Ss:IdZeroLevel}.
\end{remark}

\section{Examples} \label{S:Examples}

\subsection{Examples Revisited} 

In this section we revisit --very briefly-- two examples that were studied in \cite{PaulaGarciaToriZuccalli} in the context of Chaplygin systems with an extra symmetry.  Now, we deal with these examples using a one step reduction, that is, our starting point is the almost Poisson manifold $(\M, \pi_\subB)$ and using the nonholonomic momentum map associated to $\mathfrak{B}_{\mbox{\tiny{HGS}}}$ we perform a reduction illustrating Sec.~\ref{S:NHSystems}.

\medskip

\subsubsection*{Chaplygin ball} 
The Chaplygin ball is the classical example describing a ball rolling without sliding on a plane. We follow the notation of \cite{GN2009,Balseiro2021,BalseiroSansonetto22}.  The configuration manifold is $Q=SO(3)\times \mathbb{R}^2$ with coordinates $(g,(x,y))$ where $g$ is a rotation matrix (with $\alpha=(\alpha_1,\alpha_2,\alpha_3)$, $\beta=(\beta_1,\beta_2,\beta_3)$ and $\gamma=(\gamma_1,\gamma_2,\gamma_3)$ its rows) and $(x,y)\in \mathbb{R}^2$ represents the center of mass of the ball (that coincides with the geometric center).  The  Lagrangian is
$L((g,x,y), (\vecOm, \dot x, \dot y))= \tfrac{1}{2} \langle \mathbb{I}\vecOm, \vecOm \rangle +   \tfrac{m}{2}(\dot x^2 + \dot y^2)$, where $\vecOm = (\Omega_1, \Omega_2, \Omega_3)$ is the angular velocity in body coordinates, $\mathbb{I}$ the inertia tensor represented as a diagonal matrix and $m$ is the mass of the ball.  If  $\lambda = (\lambda_1, \lambda_2, \lambda_3)$ are the left-invariant Maurer Cartan 1-forms on $SO(3)$, then the constrains 1-forms are $\epsilon^x=dx - r\langle\beta,\lambda\rangle$ and $\epsilon^y=dy+r\langle\alpha,\lambda\rangle$ and therefore $D=\textup{span}\{\mathbb{X}_j:=X_j^L +R\beta_j\partial_x-R\alpha_j\partial_y\}$. The system is invariant by the action of the Lie group $G=SO(2)\times \mathbb{R}^2$.

Now following \eqref{Eq:splittingTQ}, we consider the basis $\mathfrak{B}_{\mbox{\tiny{$TQ$}}}$ adapted to the splitting $TQ=H\oplus S\oplus W$ given by $\mathfrak{B}_{\mbox{\tiny{$TQ$}}} = \textup{span}\{X_1,X_2, Y, \partial_x,\partial_y\}$ where $X_i = \mathbb{X}_i-\gamma_i Y$ for $i=1,2$, $Y=\langle \gamma, \mathbb{X}\rangle$ for $\mathbb{X}=(\mathbb{X}_1,\mathbb{X}_2,\mathbb{X}_3)$. Consider the coordinates $(M_1,M_2,M_3,p_x,p_y)$  on $T^*_qQ$ associated to the dual basis $\mathfrak{B}_{\mbox{\tiny{$T^*Q$}}} = \{\sigma^1, \sigma^2,\sigma^Y, \epsilon^x,\epsilon^y\}$ of $T^*Q$, then the constraint manifold $\M$ is determined by the coordinates $(g,x,y,M_1,M_2,M_3)$ (since $(p_x,p_y)$ depend linearly on $M_i$).
Since $\g\simeq \mathbb{R}\times \mathbb{R}^2$, then the section $\xi(g,x,y)= (1;-x,y)$ of the action bundle $Q\times \g\to Q$ generates the subbundle $\g_S\to Q$ and it defines the horizontal gauge momenta $J_\xi = {\bf i}_{\xi_\subM} \Theta_\subM= M_3$ (observe that $\xi_\subQ = Y$).

Following \cite{Balseiro2021}, and the notation in \eqref{Eq:splittingTM}, \eqref{Eq:localCoordTM}, the 2-form $B_1 = \langle J,\mathcal{K}_{\mbox{\tiny{$\mathcal{W}$}}}\rangle - J_\xi d\widetilde \epsilon^Y$ and we obtain that 
$\displaystyle{\pi_1= \mathcal{X}_1\wedge\partial_{M_1} + \mathcal{X}_2\wedge \partial_{M_2} - \tfrac{1}{\gamma_3}(M_1\gamma_1-M_2\gamma_2)\partial_{M_1}\wedge \partial_{M_2}. }$
The nonholonomic momentum map, given by $\langle J_\nh(m),\xi(q)\rangle = M_3$ is a momentum bundle map for $\pi_1$ associated to the basis $\mathfrak{B}_{\mbox{\tiny{HGS}}} = \{ \xi\}$ (since $\pi_1^\sharp(dM_3)=-\xi_\subM$). The $G$-invariant submanifold $J_\nh^{-1}(\mu)\subset \M$ (for $\mu = c{\bf e}^1$) is defined by $M_3=c$, and the almost Dirac structure $L_c^1= \iota_\mu^*(L_1) \subset \mathbb{T}(J_\nh^{-1}(\mu))$ is not the graph of a 2-form or a bivector field since the elements $(0,\epsilon^x)$ and $(-\xi_\subM,0)$ belong to $L_c^1$.  The push-forward of $L_c^1$ to the quotient manifold $J_\nh^{-1}(\mu)/G$ (with coordinates $(\gamma, M_1,M_2)$) gives the Dirac structure defined by the symplectic 2-form $\omega_c^1 = \sigma^1\wedge dM_1 + \sigma^2\wedge dM_2 + \tfrac{1}{\gamma_3}(M_1\gamma_1-M_2\gamma_2)\sigma^1\wedge \sigma^2$, illustrating Prop.\ref{Prop:L1red_simplectic}. 
Moreover, using these coordinates, it is straightforward to see that  $(J_\nh^{-1}(\mu)/G, \omega_c^1)$ is symplectomorphic to $(T^*S^2, \omega_{\mbox{\tiny{can}}} =  -d(M_1\sigma^1+ M_2\sigma^2))$.

Finally, since $\mathcal{B} = (mr^2\langle \gamma,\Omega\rangle + M_3)\gamma_3\sigma^1\wedge \sigma^2,$ from Theorems \ref{T:NH_momentum_red} and \ref{T:PB-mu-identification}, we conclude that restricted reduced nonholonomic dynamics $X_\red^c$ is a Hamiltonian vector field on the almost symplectic manifold
\begin{equation}\label{Eq:Ex:ChapSphere}
(T^*S^2, \omega_{\mbox{\tiny{can}}} =  -d(M_1\sigma^1+ M_2\sigma^2) + \hat{\mathcal{B}}_c),
\end{equation}
where $\hat{\mathcal{B}}_c = \left( \tfrac{mr^2}{Y(\gamma)}\langle A^{-1} M, \gamma\rangle + \tfrac{c}{Y(\gamma)}\right) \gamma_3 \, \sigma^1\wedge \sigma^2,$ for $Y(\gamma)= 1-mr^2\langle A^{-1}\gamma, \gamma\rangle$, $A=\mathbb{I} +mr^2Id$ and $M= (M_1, M_2, \tfrac{-M_1\gamma_1 -M_2\gamma_2}{\gamma_3})$. We stress the Chaplygin spirit of \eqref{Eq:Ex:ChapSphere} and the fact that these are the leaves of the bivector field $\pi_\red^\subB$ that describes the reduced dynamics.

\medskip

\subsubsection*{Solids of revolution rolling on a plane} 
Consider a solid of revolution that is invariant around the vertical axis and therefore the inertia matrix $\mathbb{I}$ has principal moment of inertia $I_1=I_2$ and $I_3$. Keeping the notation of the previous example, the orientation of the solid is represented by the rotational matrix $g\in SO(3)$  and by $(x,y,z)\in\mathbb{R}^3$ the position of the center of mass.  
Since the solid rolls on a plane, the configuration manifold is $Q=\{(g,x,y,z) : z=-\langle \gamma,s\rangle\}$ where $s:S^2\to \mathfrak{S}$ is the map from $S^2$ to the surface of the body $\mathfrak{S}$ given by $s(\gamma) = (\varrho(\gamma_3)\gamma_1, \varrho(\gamma_3)\gamma_2,\upsilon(\gamma_3))$ for $\varrho,\upsilon$ smooth functions depending on the shape of the body, see \cite[Sec.~6.7.1]{Cushman2010}.  
The nonsliding constraints are given by the annihilator of the 1-forms $\epsilon^x= dx - \langle \alpha, s\times \lambda\rangle$ and $\epsilon^y = dy - \langle \beta,s\times \lambda\rangle$ and $D$ is generated by $\mathbb{X}_i := X_i^L +(\alpha\times s)_i\partial_x+(\beta\times s)_i\partial_y + (\gamma\times s)_i\partial_z$, for $i=1,2,3$.  
The Lie group $G=S^1\times SE(2)$ is a symmetry of the nonholonomic system when $\gamma_3\neq 1$ (since we want free and proper actions), see \cite{Cushman2010}. 

From \cite{Balseiro2021}, we recall that $S= \textup{span}\{\tilde Y_1 = -\mathbb{X}_3, \tilde Y_2=\langle \gamma, \mathbb{X}\rangle\}$.
Having in mind that $\g\simeq \mathbb{R}\times \mathbb{R}\times \mathbb{R}^2$, we observe that $\zeta_1 = (1; 0,h_1)$ and $\zeta_2 = (0; 1,h_2)$ are sections generating the bundle $\g_S\to Q$, where $h_1=h_1(g,x,y) = (y+\varrho\beta_3,-x-\varrho\alpha_3)$ and $h_2=h_2(g,x,y) = (y-R\beta_3,-x+R\alpha_3)$, for $R=\varrho.\gamma_3 - \upsilon$, that is, $(\zeta_i)_\subQ = \tilde Y_i$ for $i=1,2.$
Associated to the splitting $TQ=H\oplus S \oplus W$, we consider the following basis of $TQ$ defined by $\mathfrak{B}_{\mbox{\tiny{$TQ$}}} = \{X_0 =\gamma_1\mathbb{X}_2-\gamma_2\mathbb{X}_1, \widetilde Y_1, \widetilde Y_2, \partial_x,\partial_y\}$ with dual basis denoted by  $\mathfrak{B}_{\mbox{\tiny{$T^*Q$}}} = \{\sigma^0, \widetilde{Y}^1, \widetilde{Y}^2, \epsilon^x,\epsilon^y\}$. If $(p_0,p_1,p_2, p_x,p_y)$ are the coordinates of $T_q^*Q$ associated to our chosen basis, then  $(g,(x,y), p_0,p_1,p_2)$ determine the coordinates on $\M$. 
The symmetries also induce two horizontal gauge momenta $J_1, J_2$ that are given by $J_i = f_i^1 p_1 + f_i^2p_2$ for $i=1,2$ and where $f_i^j$ are smooth functions on $Q$ determined by a ODE (linear) system (see \cite{BalseiroSansonetto22,Cushman2010}). Then the horizontal gauge symmetries are given by sections $\xi_1$, $\xi_2$ so that $\xi_i = f_i^1 \zeta_1 + f_i^2\zeta_2$ and then $\mathfrak{B}_{\mbox{\tiny{HGS}}} = \{\xi_1, \xi_2\}$.  

Since $\textup{rank}(H)=1$ then $B=B_1$ and then we observe that $\pi_1=\pi_\subB = \mathcal{X}_0\wedge \partial_{p_0} + \mathcal{Y}_i\wedge \partial_{J_i}$, where $\mathcal{X}_0\in \Gamma(\mathcal{H})$ as in \eqref{Eq:localCoordTM} and $\mathcal{Y}_i = (\xi_i)_\subM$. 
As Prop.~\ref{Prop:Jnh-MomMap:MBM} guarantees, the nonholonomic momentum map $J_\nh:\M\to \g_S^*$ given by $J_\nh(m) = p_1\nu^1 + p_2\nu^2$ --where $\nu^1 = (1;0,0)$ and $\nu^2=(0;1,0)$ are sections on $\g_S^*\to Q$ dual to the basis $\{\zeta_1,\zeta_2\}$-- is a momentum bundle map for $\pi_\subB$ associated to $\mathfrak{B}_{\mbox{\tiny{HGS}}}$,  since we can check that $\pi_\subB^\sharp(dJ_i) = - (\xi_i)_\subM$.  
Now, if $\mu\in\g^*_S$ given by $\mu=c_1\mu^1+c_2\mu^2$ for $\{\mu^1, \mu^2\}$ the dual basis of $\mathfrak{B}_{\mbox{\tiny{HGS}}}$ then $J^{-1}_\nh(\mu) = \{(g,(x,y), p_0,p_1,p_2)\in\M : f_i^1p_1+f_i^2p_2 = c_i\}$ is a $G$-invariant submanifold of $\M$.  Moreover, observe now that $L_\mu^1 = L_\mu^\subB = \textup{span} \{(\partial_{p_0}, \sigma^0),(\mathcal{X}_0,dp_0),(\mathcal{Y}_i,0), (0,\epsilon^x),(0,\epsilon^y)\}$ and therefore $(L_\mu^\subB)_\red = \textup{span} \{(\partial_{p_0}, \sigma^0),(\mathcal{X}_0,dp_0)\}$ which is clearly the graph of the symplectic 2-form $\omega_\mu^\subB = \sigma^0\wedge dp_0$. 

Moreover,  we can also observe that $J^{-1}_\nh(\mu)/G$ has coordinates $(\gamma_3,p_0)$ and therefore it is diffeomorphic to $T^*S^1$ with the canonical symplectic form:  $\omega_\mu^\subB = \omega_{\mbox{\tiny{can}}} = -d(p_0\sigma^0)$. 
From here we conclude that that the reduced dynamics is described on each leaf by the canonical symplectic manifold $(T^*S^1,\omega_{\mbox{\tiny{can}}} )$ which are also the leaves of the symplectic foliation associated to the Poisson bracket $\pi_\red^\subB$ that describes the reduced dynamics.

\subsection{The homogeneous ball in a convex surface of revolution}

In this section, we study the momentum bundle map reduction process of the nonholonomic system consisting of a homogeneous sphere rolling without slipping on a convex surface of revolution, \cite{Balseiro2021, BorisovKilinMamaev2002,FGS2005,Hermans_1995}. In this section, we follow the approach in \cite{Balseiro2021}, where after a gauge transformation, the reduction by symmetries of the nonholonomic system admits a Poisson description $\pi_\red^\subB$ with 2-dimensional leaves given by the common level sets of the two horizontal gauge momenta. As a consequence of Theorem \ref{T:PB-mu-identification}, we will give a new perspective of these leaves bringing a canonical symplectic description.  This system satisfies the conserved quantity assumption and following \cite{Balseiro2021}, we consider the gauge transformation of a 2-form $B$ and show that the nonholonomic momentum map is a momentum map for $\pi_\subB$.
Observe that this example cannot be studied using the theory developed in \cite{PaulaGarciaToriZuccalli} since its group of symmetries can not be decomposed so that the system is Chaplygin with an extra symmetry.  

More precisely, consider the nonholonomic system given by a balanced and dynamically symmetric (homogeneous) sphere of mass $m$ and radius $R$ rolling without slipping and under the influence of gravity inside a convex surface of revolution $\Sigma \subset \mathbb{R}^{3}$ (whose vertical axis of symmetry is parallel to the gravitational force). The sphere being balanced means that its center of mass $C$ coincides with its geometrical center, and being dynamically symmetric means that its inertia tensor can be written in the form $\mathbb{I} = I\cdot \textup{Id}$ in an adapted frame, where $I > 0$
is constant and $\textup{Id}$ is the $3\times3$ identity matrix. 

The configuration manifold $Q$ is $\mathbb{R}^{2}\times SO(3)$ with coordinates $(x,y,g)$, where $(x,y)$ denotes the coordinates of the center of mass $C$ projected onto the plane $z=0$, and $g$ denotes the orthogonal matrix relating the {\it space frame} (with origin at $O$) and the {\it body frame} (a moving orthonormal frame attached to the ball with origin at $C$). See Figure \ref{fig:HomoConvex}. 

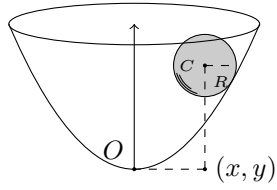
\begin{figure}[h!]
    \centering
    \begin{tikzpicture}
        \begin{scope}[shift={(0,0)}, scale=0.55]
            \fill[black!20] (1.69,0.5) circle (0.75cm);
            \draw (-3,1.5) parabola bend (0,-2) (3,1.5);
            \draw (0,1.5) ellipse (3cm and 0.5cm);
            \draw (1.69,0.5) circle (0.75cm); 
            \draw[fill] (1.69,0.5) circle (1pt) node[left] {\tiny$C$}; 
            \draw[fill] (0,-2) circle (1pt) node[above left] {\small$O$}; 
            \draw[dashed] (1.69,0.5) -- node[below] {\tiny$R$} (2.44,0.5); 
            \draw[->] (0,-2) -- (0,1.5); 
            \draw[rotate around={15:(1.69,0.5)}](0.99,0.5) arc (180:240:0.7cm);
            \draw[rotate around={20:(1.69,0.5)}](1.04,0.5) arc (180:220:0.65cm);
            \draw[dashed] (1.69,0.5) -- (1.69,-2); 
            \draw[dashed] (1.69,-2) -- (0,-2); 
            \draw[fill] (1.69,-2) circle (1pt) node[right] {\small$(x,y)$}; 
        \end{scope}
    \end{tikzpicture}
    \vspace{-0.5cm}
    \caption{The homogeneous ball in a convex surface of revolution.}
    \label{fig:HomoConvex}
\end{figure}

The surface $\Sigma$ is generated by rotating --around the $z$-axis-- the graph of a convex and smooth function $\phi : \mathbb{R}^{+}\rightarrow\mathbb{R}$. Then the equation characterizing $\Sigma$ is given by $z = \phi(x^{2}+y^{2})$. To guarantee smoothness and convexity of the surface, we assume that $\phi$ verifies that $\phi^{\prime}(0^{+}) = 0$, $\phi^{\prime}(s) > 0$ and $\phi^{\prime\prime}(s) > 0$, when $s >0$. Furthermore, to ensure the ball has only one contact point with the surface, the curvature of $\phi(s)$ must not exceed $1/R$. Denote by $\mathbf{n} = \mathbf{n}(x,y)$ the exterior unit normal vector to $\Sigma$ at the contact point with coordinates $(n_1,n_2,n_3)$. 
Let us denote by $\omega = (\omega_{1},\omega_{2},\omega_{3})$ the angular velocity associated to the right invariant frame $\{X_{1}^{R},X_{2}^{R},X_{3}^{R}\}$ of $TSO(3)$. Then, if  $a_{g}$ denotes the gravity acceleration, the Lagrangian of the system is
    $\displaystyle{L(x,y,g,\dot{x},\dot{y},\omega) = \tfrac{I}{2}\langle \omega,\omega\rangle+\tfrac{m}{2n_{3}^{2}}\left\{(1-n_{2}^{2})\dot{x}^{2}+(1-n_{1}^{2})\dot{y}^{2}+2n_{1}n_{2}\dot{x}\dot{y}\right\}-ma_{g}\phi.}$

\medskip

\noindent{\bf The nonholonomic constraints.} The nonholonomic constraint equations given by the non-sliding condition are 
$\dot{x} = -R(\omega_{2}n_{3}-\omega_{3}n_{2})$ and $\dot{y} = -R(\omega_{3}n_{1}-\omega_{1}n_{3}).$
Thus the constraint 1-forms are given by \begin{equation}\label{Eq:homocvxConsteq} 
\epsilon^{1} = dx-R(n_{2}\rho_{3}-n_{3}\rho_{2})\quad\textup{and}\quad \epsilon^{2} = dy-R(n_{3}\rho_{1}-n_{1}\rho_{3}), 
\end{equation}
where we denote by $\{\rho_{1},\rho_{2},\rho_{3}\}$ the right Maurer-Cartan 1-forms on $SO(3)$ and by $\{X_{1}^{R},X_{2}^{R},X_{3}^{R}\}$ its dual basis of right invariant vector fields.  Then the constraint distribution $D$, defined by the annihilator of $\epsilon^{1}$ and $\epsilon^{2}$, is given by 
$$
    D = \textup{span}\left\{Y_{x} := \partial_{x}+\tfrac{1}{Rn_{3}}(n_{2}Y_{n}-X_{2}^{R}),\, Y_{y} := \partial_{y}-\tfrac{1}{Rn_{3}}(n_{1}Y_{n}-X_{1}^{R}),\,Y_{n}\right\}, 
$$
where $Y_{n} := \langle {\bf n}, {\bf X}^R\rangle = n_1X_1^R +n_2X_2^R + n_3X_3^R$ for ${\bf X}^R = (X_{1}^{R},X_{2}^{R},X_{3}^{R})$.

\medskip

\noindent{\bf The symmetries and the splitting of $TQ$.} Consider the action of the Lie group $G = SO(2)\times SO(3)$ on $Q$ given, at each $(x,y,g)\in Q$, by $(h_{\theta},h)\cdot(x,y,g) := (h_{\theta}(x,y)^{T},\tilde{h}_{\theta}gh^{T}),$
where $h_\theta$ is the $2\times 2$ rotation matrix of angle $\theta$ and $\tilde{h}_{\theta}$ denotes the $3\times3$ rotation matrix of angle $\theta$ with respect to the $z$-axis. The Lagrangian and the constraints equations are invariant with respect to the lifted action on $TQ$. 
The Lie algebra $\mathfrak{g}$ of $G$ is isomorphic to $\mathbb{R}\times\mathbb{R}^{3}$ with infinitesimal generators given by
$$
    (1,\mathbf{0})_{Q} = -y\partial_{x}+x\partial_{y}+X_{3}^{R} \quad\textup{and}\quad (0,\mathbf{e}_{i})_{Q} = \alpha_{i}X_{1}^{R}+\beta_{i}X_{2}^{R}+\gamma_{i}X_{3}^{R},\quad\textup{for}\quad  i= 1,2,3,
$$
where $\mathbf{e}_{i}$ denotes the $i$-th element of the canonical basis of $\mathbb{R}^{3}$ and $\alpha = (\alpha_{1},\alpha_{2},\alpha_{3})$, $\beta = (\beta_{1},\beta_{2},\beta_{3})$ and $\gamma = (\gamma_{1},\gamma_{2},\gamma_{3})$ the rows of the matrix $g\in SO(3)$. 
In this section we will always assume that  $(x,y) \neq (0,0)$ so that we have a free and proper action (observe that the action is free whenever $(x,y) \neq (0,0)$).

Following Sec.~\ref{S:Identification}, we consider the splitting $TQ=H\oplus S\oplus W$ given by 
\begin{equation}\label{Eq:ExBall-splitting}
    \begin{split}
        H & =\textup{span} \{X_{0} := xY_{x}+yY_{y}\}, \qquad S = \textup{span} \{ Y := -yY_{x}+xY_{y},\, Y_{n}\},\\ 
        W & = \textup{span}\{ Z_{1} := \tfrac{1}{Rn_{3}}(X_{2}^{R}-n_{2}Y_{n}), Z_{2} := -\tfrac{1}{Rn_{3}}(X_{1}^{R}-n_{1}Y_{n})\}.
    \end{split}
\end{equation}
Associated to this adapted basis, we consider the dual basis $\mathfrak{B}_{\mbox{\tiny{$T^{\ast}Q$}}} = \left\{X^{0}, \alpha_Y, \rho_{n}, \epsilon^{1},\,\epsilon^{2}\right\}$, where 
$X^{0} := \frac{xdx+ydy}{x^{2}+y^{2}}$, $\alpha_Y := \frac{xdy-ydx}{x^{2}+y^{2}}$ and $\rho_n = \langle {\bf n} , \rho\rangle$ for $\rho = (\rho_1,\rho_2,\rho_3)$. 
We denote by $(p_{0},p_{Y},p_{n},M_{1},M_{2})$ the associated coordinates in $T^{\ast}Q$.  The constraint manifold $\M\subset T^*Q$ is given by the relations
$$
    M_{1} = \tfrac{-I}{I+mR^{2}}\left(\tfrac{xp_{0}-yp_{Y}}{x^{2}+y^{2}}\right)\quad\textup{and}\quad M_{2} = \tfrac{-I}{I+mR^{2}}\left(\tfrac{xp_{Y}+yp_{0}}{x^{2}+y^{2}}\right). 
$$

\noindent {\bf The horizontal gauge momenta.}  From \eqref{Eq:ExBall-splitting}, we may observe that $H=D\cap S^\perp$ but the chosen basis of $S$ does not generate the conserved quantities.  
In fact, this example has two horizontal gauge momenta (see e.g., \cite{BalseiroSansonetto22, Fasso2007}) that cannot be explicitly written. 
More precisely, it was proven in \cite{BalseiroSansonetto22, Fasso2007} that the horizontal gauge momenta of the system are the functions $J_1$ and $J_2$ defined on $\M$ given, respectively, by
$$
    J_1 = f_1^Y p_Y + f_1^np_n \qquad \mbox{and} \qquad J_2 = f_2^Y p_Y + f_2^np_n,
$$
where $f^Y_i, f_i^n$, for $=1,2$ are $G$-invariant functions on $Q$ which are defined as the (two independent) solutions of a linear system of ordinary differential equations: 
\begin{equation}\label{Eq:BallSurface:ODE}
    R_{12} f^n_i = X_0(f_i^Y) \qquad \mbox{and} \qquad R_{21} f^Y_i = X_0(f_i^n),
\end{equation}
where $R_{12}$ and $R_{21}$ are $G$-invariant functions on $Q$. 
We consider the basis of sections $\mathfrak{B}_{\g_S} = \{ \eta, \eta_n\} $ of the bundle $\g_S\to Q$ where $\eta_\subQ = Y$ and $(\eta_n)_\subQ = Y_n$.   
Following Def.~\ref{Def:HGM} and Lemma~\ref{L:Ad-basis}, the associated basis of horizontal gauge symmetries is given by 
$$
    \mathfrak{B}_{\mbox{\tiny{HGS}}} = \left\{\xi_{1} :=f_1^Y \eta + f_1^n\, \eta_n,\, \,  \xi_{2} := f_2^Y \eta + f_2^n\, \eta_n\right\}.
$$
For $\{ \nu, \nu^n\}$ the basis of section of the bundle $\g_S^*\to Q$  dual of $\mathfrak{B}_{\g_S}$, the nonholonomic momentum bundle map $J_\nh:\M\to \g_S^*$ is written, for $m=(x,y,g,p_0,p_Y,p_n)\in \M$, as 
$$
    J_\nh(m) = p_Y\nu + p_n\nu^n. 
$$
Since  $\pi_\nh^\sharp(dJ_i) \neq -(\xi_i)_\subM$ for $i=1,2,$ we conclude that $J_\nh:\M\to \g^*_S$ is not a momentum map for $\pi_\nh$ associated to the basis $\mathfrak{B}_{\mbox{\tiny{HGS}}}$.

\medskip

\noindent {\bf The gauge transformation and the nonholonomic momentum bundle map.}
Following \cite{Balseiro2021}, and the formulation \eqref{Eq:B1} we obtain the 2-form $B = B_1$ on $\M$ ($\mathcal{B}=0$ since $\textup{rank}(H)=1$)
$$
    B =  \langle J,  {\mathcal K}_{\mbox{\tiny{${\mathcal W}$}}} \rangle - p_Y R_{12}  {\mathcal X}^0\wedge\rho_n - p_n R_{21}  {\mathcal X}^0\wedge \alpha_{\mathcal Y} + p_n d\widetilde\rho_n,
$$
where $\mathfrak{B}_{\mbox{\tiny{$T^*\!\M$}}} = \{{\mathcal X}^0, \alpha_{\mathcal Y}, \widetilde\rho_n, \widetilde\epsilon^1, \widetilde\epsilon^2,dp_0,dp_Y, dp_n\}$ is a basis of $T^*\M$ given by ${\mathcal X}^0 =\tau_\subM^*X^0$, $\alpha_{\mathcal Y} = \tau_\subM^*\alpha_Y$, $\widetilde\rho_n= \tau_\subM^* \rho_n$ and $\widetilde\epsilon^a = \tau_\subM^*\epsilon^a$. Now, if $\mathfrak{B}_{\mbox{\tiny{$T\M$}}} = \{{\mathcal X}_0, {\mathcal Y}, \mathcal{Y}_n, \mathcal{Z}_1, \mathcal{Z}_2, \partial_{p_0},\partial_{p_Y}, \partial_{p_n}\}$ is the dual basis of $\mathfrak{B}_{\mbox{\tiny{$T^*\M$}}}$, we can write the gauge related bivector field $\pi_\subB$ as 
\begin{equation} \label{Eq:BallSurface:PiB}
   \pi_\subB = {\mathcal X}_0\wedge \partial_{p_0} + \mathcal{Y}\wedge \partial_{p_Y} + \mathcal{Y}_n \wedge \partial_{p_n} + p_n R_{21} \partial_{p_0}\wedge \partial_{p_Y} + p_Y R_{12}\partial_{p_0}\wedge \partial_{p_n}. 
\end{equation}

The bivector field $\pi_\subB$ not only describes the nonholonomic dynamics (as it was proven in \cite{Balseiro2021}), but also it is straightforward to check, using \eqref{Eq:BallSurface:ODE}, that the nonholonomic momentum bundle map $J_\nh:\M\to \g_S^*$ is the momentum map of $\pi_\subB$ associated to $\mathfrak{B}_{\mbox{\tiny{HGS}}}$, since $\pi_\subB^\sharp(dJ_i) = -(\xi_i)_\subM$ for $i=1,2$ (in concordance with Prop.~\ref{Prop:Jnh-MomMap:MBM}).   

\medskip

\noindent {\bf The momentum reduction and the Chaplygin foliation.}
Consider the basis $\mathfrak{B}_{\mbox{\tiny{HGS}}}^{\ast} := \{\mu^{1},\mu^{2}\}$ of $\g_S^*\to Q$ dual basis to $\mathfrak{B}_{\mbox{\tiny{HGS}}}$ and let $\mu := c_{1}\mu^{1}+c_{2}\mu^{2}$ where $c_1, c_2\in \mathbb{R}$. From \eqref{Eq:NHLevelSet}, we can foliate ${\mathcal{M}}$ by the connected components of the $G$-invariant submanifolds ${J}_{\nh}^{-1}(\mu)\subset{\mathcal{M}}$ given by ${J}_{\nh}^{-1}(\mu) = {J}_{1}^{-1}(c_{1})\cap{J}_{2}^{-1}(c_{2}) = \left\{(x,y,g,p_{0},p_{Y},p_{n})\in {\mathcal{M}};\, f_i^Y p_Y + f_i^np_n = c_{i},\,i=1,2\right\}.$
Following Section~\ref{Ss:Nonholonomic_momentum_reduction}, the backward image of $\pi_\subB$ by the inclusion $\iota_\mu:J^{-1}_\nh(\mu)\to \M$ defines the almost Dirac structure 
$$
    L_\mu^\subB = \textup{span}\{(\partial_{p_0}, {\mathcal X}^0), (-\mathcal{X}_0   , d p_0), (-(\xi_1)_\subM, 0 ) , (-(\xi_2)_\subM, 0) , (0 , \epsilon^a)\}  \subseteq \mathbb{T}(J^{-1}_\nh(\mu)),
$$
where we are changing the coordinates so that the bivector field $\pi_\subB$ is written as $\pi_\subB={\mathcal X}_0\wedge \partial_{p_0} + (\xi_i)_\subM\wedge \partial_{J_i}$.

The reduction to the orbit manifold ${\mathcal{M}}/G$ has orbit projection $\rho_{\mbox{\tiny{${\mathcal{M}}$}}} : {\mathcal{M}}\rightarrow {\mathcal{M}}/G$ given by 
$\rho_{\mbox{\tiny{${\mathcal{M}}$}}}(x,y,g,p_{0},p_{Y},p_{n}) = (\tau,{p}_{0},{p}_{Y},{p}_{n}),$
where $\rho_{\mbox{\tiny{${\mathcal{M}}$}}}^{\ast}\tau = x^{2}+y^{2}$.  Let us observe that $T\rho_{\mbox{\tiny{${\mathcal{M}}$}}}(\mathcal{X}_0)= 2 \tau \partial_\tau$ and also $\rho^*_{\mbox{\tiny{${\mathcal{M}}$}}} \tfrac{1}{2\tau} d\tau = \mathcal{X}^0$.   Making the change of coordinates $\tilde \tau= -\tfrac{1}{2}\ln(\tau)$, we explicitly show the diffeomorphism between the reduced manifolds $J_\nh^{-1}(\mu)/G$ and $T^*(\mathbb{R})$. Moreover, the forward image of $L_\mu^\subB$ by the projection $\rho_\mu : J_\nh^{-1}(\mu) \to J_\nh^{-1}(\mu)/G$ defines the reduced Dirac structure $(L^\subB_\mu)_{\red} = \textup{span}\{(-\partial_{p_0},  d\tilde\tau), ( \partial_{\tilde\tau}, dp_{0})\} \subseteq \mathbb{T}(J^{-1}_\nh(\mu)/G),$ which in agreement with Theorem \ref{T:NH_momentum_red} (in particular Prop.~\ref{Prop:L1red_simplectic}), it is a symplectic 2-form.   That is $(J_{\nh}^{-1}(\mu)/G, (L^\subB_\mu)_{\red})$ is diffeomorphic to $(T^*(\mathbb{R}), \omega_{\mbox{\tiny{can}}})$ where 
$$
    \omega_{\mbox{\tiny{can}}} = d\tilde \tau \wedge dp_0.
$$

We finally conclude that the canonical symplectic manifolds $(T^*(\mathbb{R}), \omega_{\mbox{\tiny{can}}})$ are the symplectic leaves of the reduced bivector field $\pi_\red^\subB$, induced by \eqref{Eq:BallSurface:PiB}. Therefore, we conclude that the reduced dynamics of the example presented is described by ``canonical symplectic leaves''.  

\subsection{Chaplygin ball over a fixed sphere}
In this section, we will study the momentum map reduction process of the example that inspired Sec.~\ref{S:general-0level}: the Chaplygin ball rolling over a sphere \cite{BorisovFedorov1995}.  More precisely, consider a sphere with a inhomogeneous distribution of mass of radius $r_0$ whose geometric center coincides with its center of mass.
Following  \cite{Fedorov2008,BorisovMamaev2002, BorisovMamaev2008, Tsiganov2011}, the ball
is allowed to roll without sliding over a fixed sphere of radius $R_0$ in three possible configurations:

\begin{enumerate}
    \item[$(a)$] The moving ball, with center $C$, is inside the fixed sphere, with center $O$, and necessarily $r_{0} < R_{0}$. This means the contact point is on the exterior of the moving ball but inside the fixed sphere; 
    \item[$(b)$] The moving ball is on the exterior of the fixed sphere, meaning the contact point is on the exterior of both spheres; 
    \item[$(c)$] The fixed sphere is inside the moving ball (and necessarily $R_{0} < r_{0}$), meaning the contact point is on the interior of the moving sphere and on the exterior of the fixed sphere. 
\end{enumerate}
\vspace{-0.7cm}
\begin{figure}[h!]
    \centering
    \begin{tikzpicture}
    \begin{scope}[shift={(0,0)}, scale=0.7]
        \fill[black!20] (0,1.25) circle (0.75cm);
        \draw (0,0) circle (2cm); 
        \draw[fill] (0,0) circle (1pt) node[left] {\small$O$}; 
        \draw[dashed] (0,0) -- node[below] {\small$R_{0}$} (2,0); 
        
        \draw (0,1.25) circle (0.75cm); 
        \draw[fill] (0,1.25) circle (1pt) node[above] {\tiny$C$}; 
        \draw[dashed] (0,1.25) -- node[right] {\tiny$\quad r_{0}$} (0.75,1.25); 
        \draw[rotate around={15:(0,1.25)}] (-0.7,1.25) arc (180:240:0.7cm);
        \draw[rotate around={20:(0,1.25)}] (-0.65,1.25) arc (180:220:0.65cm);

        \draw[fill] (0,2) circle (1pt) node[above] {\small$P$}; 
        \draw[fill] (0,-2) circle  node[below] {$(a)\,k>1$}; 
    
    \end{scope}

    \begin{scope}[shift={(5,0)}, scale=0.7]
        \fill[black!20] (0,1.5) circle (0.5cm);
        \draw (0,-0.5) circle (1.5cm); 
        \draw[fill] (0,-0.5) circle (1pt) node[left] {\small$O$}; 
        \draw[dashed] (0,-0.5) -- node[below] {\small$R_{0}$} (1.5,-0.5); 
        
        \draw (0,1.5) circle (0.5cm); 
        \draw[fill] (0,1.5) circle (1pt) node[above] {\tiny$C$}; 
        \draw[dashed] (0,1.5) -- node[right] {\tiny$\;\;r_{0}$} (0.5,1.5); 
        \draw[rotate around={15:(0,1.5)}] (-0.45,1.5) arc (180:240:0.45cm);
        \draw[rotate around={20:(0,1.5)}] (-0.4,1.5) arc (180:220:0.40cm);

        \draw[fill] (0,1) circle (1pt) node[below] {\small$P$}; 
        \draw[fill] (0,-2) circle  node[below] {$(b)\,0<k<1$}; 
    \end{scope}

    \begin{scope}[shift={(10,0)}, scale=0.7]
        \fill[black!20] (0,0) circle (2.1cm);
        \fill[white] (0,0) circle (2cm);
        
        \draw (0,0) circle (2cm); 
        \draw[fill] (0,0) circle (1pt) node[left] {\small$C$}; 
        \draw[dashed] (0,0) -- node[below] {\small$r_{0}$} (2,0); 
        \draw[rotate around={15:(0,0)}] (-1.95,0) arc (180:240:1.95cm);
        \draw[rotate around={20:(0,0)}] (-1.90,0) arc (180:220:1.90cm);
        
        \draw (0,1.25) circle (0.75cm); 
        \draw[fill] (0,1.25) circle (1pt) node[above] {\tiny$O$}; 
        \draw[dashed] (0,1.25) -- node[right] {\tiny$\quad R_{0}$} (0.75,1.25); 

        \draw[fill] (0,2) circle (1pt) node[above] {\small$P$}; 

    \draw[fill] (0,-2) circle  node[below] {$(c)\,k<0$}; 
    \end{scope}
    \end{tikzpicture}
    \vspace{-0.5cm}
    \caption{Configurations of Chaplygin ball over a fixed sphere.}
    \label{fig:Chaplyginballinball}
\end{figure}
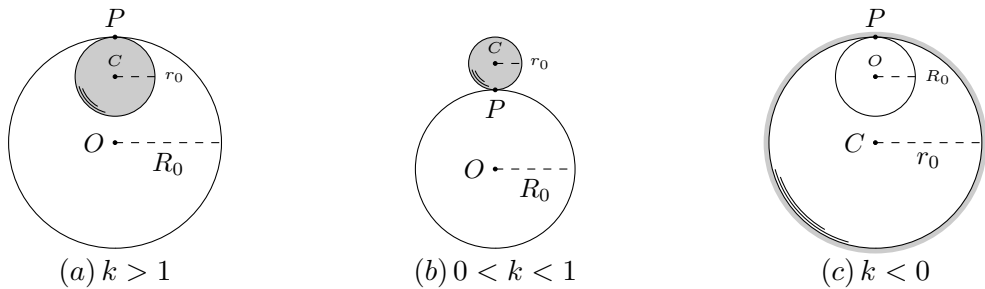
Let us define 
$$
    k := \tfrac{R_{0}}{R_{0} - r_{0}} \textup{ (for cases $(a)$ and $(c)$)} \quad \textup{and} \quad k := \tfrac{R_{0}}{R_{0} + r_{0}} \textup{ (for the case $(b)$)},
$$ 
and we observe that for $(a)$ we have that $k > 1$; for $(b)$ $k$ satisfies that $0<k<1$; and finally $k<0$ for $(c)$. The limit case where $k$ becomes $1$ corresponds to the classical Chaplygin ball on a plane, see \cite{BMChaplygin2002, Chaplygin2002}. 

According to \cite{Fedorov2008, BorisovMamaev2002, BorisovMamaev2008}, for an arbitrary value of $k$,  it is not possible to guarantee (Euler-Jacobi) integrability because, even having an invariant measure, the integrability is achieved by presenting an extra first integral (see a more detailed discussion in \cite{BorisovMamaev2008}).  
For the particular case $k=-1$ (i.e., case $(c)$ with $r_{0} = 2 R_{0}$), it was shown in 
\cite{BorisovFedorov1995, BorisovMamaevMarikhin2008} the explicit existence of an extra first integral and hence obtaining integrability.

The case $k=-1$ was originally studied in \cite{BorisovFedorov1995, BorisovMamaevMarikhin2008} and it became known as the Borisov-Mamaev-Fedorov system. This case is noteworthy because the system is hamiltonizable on the zero level set of the aforementioned first integral. In the present paper we show --from an intrinsic view point and using the techniques form Section \ref{S:general-0level}-- that the reduced 0-level set admits an almost symplectic structure recovering a result in \cite{Jovanovic2019} done in coordinates. Moreover, we put in evidence the Chaplygin character of the (reduced) 0-leaf and we see also how the hamiltonization is a consequence of this fact. 

\paragraph{The dynamics of the system.} This system has two types of constraints: the holonomic constraint which is imposed by the fact that the ball is attached to the surface of the fixed sphere, and the nonholonomic one which is defined by the non-slipping condition. 

As usual we keep the notation of the previous example and in this case the inertia tensor $\mathbb{I}$ has positive entries $I_1, I_2,I_3$. 
Denote by $g\in SO(3)$ the orthogonal matrix relating both frames and by $\mathbf{r} := \overrightarrow{OC}$ the vector that points the center of mass of the moving sphere.  The fact that the ball is balanced and its motion is attached to the surface of the fixed sphere implies that the configuration space $Q$ is given by: 
$$
    Q = \left\{(g,\mathbf{r})\in SO(3)\times\mathbb{R}^{3}, \langle \mathbf{r} , \mathbf{r} \rangle = (R_{0}/k)^{2}\,\right\}\simeq SO(3)\times S^{2}.
$$
In what follows, we will use $(g,\mathbf{r})$ as redundant coordinates of $Q$.
The angular velocity of the ball (in body frame) 
$\Omega=(\Omega_1, \Omega_2,\Omega_3)$ 
are the coordinates associated to the  left invariant frame of vector fields $\{X_{1}^{L}, X_{2}^{L}, X_{3}^{L}\}$ on $SO(3)$. 
The nonholonomic constraint representing the non-slipping condition at the contact point $P$ is described as $\dot{\mathbf{r}} = (1-k) (g\Omega)\times \mathbf{r},$
and the Lagrangian is just the kinetic energy (even for the case of a spherical potential) $L = \frac{1}{2}\langle\mathbb{I}\Omega,\Omega\rangle + \frac{1}{2}m\langle\dot{\mathbf{r}},\dot{\mathbf{r}}\rangle,$
where $m$ is the total mass of the moving sphere, see e.g.,  \cite{Jovanovic2018}. 
The normal vector to the surface of the fixed sphere at the contact point $P$ is given by $\mathbf{n} = \frac{k}{R_{0}}g^{T}\mathbf{r}$. 
Following \cite{BorisovMamaev2008, Jovanovic2018}, the equations describing the motion of the system are
$$
    E_{\bf n}\dot{\Omega} = E_{\bf n}\Omega\times\Omega - mr_{0}^{2}\mathbf{n}\times(\Omega\times\dot{\mathbf{n}})\quad\textup{and}\quad \dot{\mathbf{n}} = -k \Omega \times \mathbf{n},
$$
where $E_{\bf n}:= \mathbb{I}+mr_{0}^{2}(\textup{Id}-\mathbf{n}\otimes \mathbf{n})$.

\paragraph{The geometric setting of the system.} 
Recall that $\{\lambda_1, \lambda_2,\lambda_3\}$ are the left invariant Maurer-Cartan 1-forms, and the constraint 1-forms can be written, for $i=1,2,3$, as 
$$
    \epsilon^{i} := dr_{i}-(1-k)(\mathbf{c}_{j}\times \mathbf{r})_{i}\lambda_{j} = dr_{i}-(1-k)\langle g^{T}e_{i}, \lambda\times g^{T}\mathbf{r}\rangle,
$$
where $c_1, c_2, c_3$ are the columns of the matrix $g$. 
The constraint distribution $D$ on $Q$ is given by 
$$
    D = \textup{span}\left\{X_{i} := X_{i}^{L}+(1-k)\langle \mathbf{c}_{i}\times \mathbf{r} , \partial_{\mathbf{r}}\rangle,\, i = 1,2,3\right\}.
$$
Consider the following basis of $TQ$ and its dual of $T^{\ast}Q$ given, respectively, by $
    \mathfrak{B}_{\mbox{\tiny{$TQ$}}} = \{X_{i},\partial_{r_{i}}\}$ 
    and $\mathfrak{B}_{\mbox{\tiny{$T^*Q$}}} =  \{\lambda_{i},\epsilon^{i}\}$,
 with coordinates $(\Omega, \dot{r})$ of $T_qQ$ and $({\bf M}, {\bf p})=(M_1,M_2,M_3;p_1,p_2,p_3)$ of $T_q^*Q$.
Then the constraint manifold $\mathcal{M} := \kappa^\flat(D)\subset T^{\ast}Q$ is given by
\begin{equation*}
    \mathcal{M} = \left\{(g,\mathbf{r},\mathbf{M},\mathbf{p}) \in T^{\ast}Q;\,p_{j} = m(1-k)(g\Omega\times\mathbf{r})_{j} \right\},
\end{equation*}
with  $M_{j} = (\mathbb{I}\Omega)_{j}+mr_{0}^{2}(\mathbf{n}\times(\Omega\times\mathbf{n}))_{j}$.

Let us denote by $\mathfrak{B}_{\mbox{\tiny{$T^*\M$}}} = \{\tilde \lambda_i, \tilde \epsilon^i, d{M_i}\}$ for $\tau_\subM^*\lambda_i=\tilde \lambda_i$ and $\tau_\subM^*\epsilon^i=\tilde \epsilon^i$, and its dual basis of vector fields $\mathfrak{B}_{\mbox{\tiny{$T\M$}}} = \{\mathcal{X}_i, \partial_{r_i}, \partial_{M_i}\}$.  
Since $\Omega_{\mathcal{C}} := \left.\Omega_\subM \right|_{\mathcal{C}} = \tilde{\lambda}_{i}\wedge dM_{i} - \langle \mathbf{M}+km\,r_{0}^{2}\mathbf{n}\times(\Omega\times \mathbf{n}),d\tilde{\lambda}\rangle|_{\mathcal{C}}$, the nonholonomic bivector field $\pi_{\nh}$ on $\mathcal{M}$ is given by
$$
    \pi_{\nh} = \mathcal{X}_{i}\wedge \partial_{M_{i}} - \textup{cyclic}\left[\,  (\mathbf{M}+kmr_{0}^{2}(\Omega-\langle\mathbf{n},\Omega\rangle\mathbf{n}))_{1}\partial_{M_{2}}\wedge\partial_{M_{3}} \, \right]. 
$$
The Hamiltonian function $
H_{\subM} := \left.H\right|_{\subM} = \frac{1}{2}\langle\mathbf{M},\Omega\rangle,$
determines the nonholonomic vector field 
$$
    X_{\nh} := -\pi_{\nh}^{\sharp}(dH_{\subM}) = \langle\Omega,\mathcal{X}\rangle+\langle\mathbf{M}\times\Omega,\partial_{\mathbf{M}}\rangle,
$$
where $\mathcal{X} = (\mathcal{X}_1, \mathcal{X}_2, \mathcal{X}_3)$.

\subsection*{The BMF-case ($k=-1$)}
Recall that the Borisov-Mamaev-Fedorov case represents the situation $(c)$ with $r_0=2R_0$.  

\paragraph{Reduction by symmetries.} Following \cite{Jovanovic2018,Jovanovic2019}, consider the left diagonal action of the Lie group $G := SO(3)$ on $Q$ given, at each $h\in SO(3)$ and $(g,{\bf r})\in Q$, by $h\cdot (g,\mathbf{r}) = (hg,h\mathbf{r})$.
The action is proper and free and its tangent lift is given by $h\cdot (g,\mathbf{r},\Omega,\dot{\mathbf{r}}) = (hg,h\mathbf{r},\Omega,h\dot{\mathbf{r}})$. It is straightforward to see that the Lagrangian $L$ and $D$ are both $G$-invariant. The Lie algebra $\mathfrak{g} :=  Lie(SO(3))$ is isomorphic to $\mathbb{R}^{3}$, and the infinitesimal generators are given by 
\begin{equation}\label{Eq:BMFinfGen}
    V_{i} := ({\bf e}_{i})_{Q} = -(\mathbf{r}\times\partial_{\mathbf{r}})_{i}+\langle g^{T}e_{i},X\rangle\quad i = 1,2,3,
\end{equation}
for $X=(X_1, X_2,X_3)$, showing that the dimension assumption is satisfied.  Furthermore, note that the distribution ${S} := {D}\cap {V}$ is generated by the vector field $Y:=\langle\mathbf{n},X\rangle$.
The cotangent lift leaves the constraint manifold $\mathcal{M}$ invariant and the corresponding orbit projection 
$\rho : \mathcal{M} \rightarrow \mathcal{M}/G$ is given by 
$\rho(g,\mathbf{r},\mathbf{M}) = (\mathbf{n},\mathbf{M})$ and we also observe that 
$\mathcal{M}/G \simeq S^{2}\times\mathbb{R}^{3}$. 

\paragraph{The conserved quantity.} 
Following \cite{Fedorov2008, BorisovMamaev2002, BorisovMamaev2008, Tsiganov2011}, we define the function $F : \mathcal{M}\rightarrow \mathbb{R}$ given, at $m=(g, {\bf r}, {\bf M}) \in \M$, by
$$
    F(m) := \langle A\mathbf{M},\mathbf{n}\rangle,
$$
where $A$ is the $3\times 3$ matrix given by $A:= tr(E)\textup{Id}-2E$  and $E := \mathbb{I}+mr_{0}^{2}\textup{Id}$. We can check that $F$ is a $G$-invariant conserved quantity of $X_{\nh}$. Moreover $F$ is a $D$-momentum since $F=  {\bf i}_{\langle A\mathbf{n},\mathcal{X}\rangle}\Theta_{\subM}$ with $\langle A\mathbf{n},\mathcal{X}\rangle \in {\mathcal C}_m$, see Def.~\ref{Def:D-momentum}. But observe that $F$ is not a horizontal gauge momentum since $\langle A\mathbf{n},\mathcal{X}\rangle\not\in \mathcal{S}_m$. 

\begin{remark}
    Following \cite{Fasso2008}, a {\it weakly Noetherian function} is a function linear on the fibers that is a first integral of $X_\nh= -\pi_\nh^\sharp(dH_\subM)$ for any potential energy considered in the Hamiltonian function $H_\subM$. Moreover, in \cite{Fasso2008} it is proven that every $G$-invariant linear  weakly Noetherian function is, in fact, a horizontal gauge momentum.  In our case, $F$ is not a horizontal gauge momentum, but we also checked that it is not a weakly Noetherian function.  
\end{remark}

\paragraph{The almost symplectic structure of the zero-level set of $F$.} Even if we  are not under the conditions of the  momentum map reduction procedure described in Sec.~\ref{Ss:Nonholonomic_momentum_reduction}, we can still study the zero-level set of the conserved quantity $F$ as done in Sec.~\ref{S:general-0level}. Using Theorem \ref{T:general0-level}, we recover the almost symplectic structure obtained in \cite{Fedorov2008} from a geometrically intrinsic view-point. 

More precisely, we check that $\widetilde{S}= \textup{span}\{\langle A{\bf n}, {X}\rangle\}$ satisfies the hypothesis of Lemma \ref{L:tildeH}, i.e., $\textup{rank}(\widetilde{S})=\textup{rank}(S)=1$ and $S \cap \tilde S^\perp = \{0\}$.  Therefore the distribution $\widetilde{H}:= \widetilde{S}^\perp \cap D  \subset D$ satisfies $\kappa^{\flat}(\widetilde{H}) = F^{-1}(0) \subset \mathcal{M}$ and we conclude that $TQ = \widetilde{H}\oplus V$.

Consider the $G$-Chaplygin system given by the same Lagrangian $L$ but with constraint distribution given by $\widetilde{H}$. According to \cite{Koiller92}, the manifold $F^{-1}{(0)}/G$ is diffeomorphic to the cotangent manifold $T^{\ast}(Q/G) = T^{\ast}(S^{2})$, and the  reduced dynamics restricted to $F^{-1}{(0)}/G\simeq T^{\ast}(S^{2})$, denoted by $X_{\red}^{0}$, is Hamiltonian with respect to the almost symplectic $2$-form, given by
$$
    \widetilde{\omega} = \omega_{\mbox{\tiny{can}}} - \hat{B}_{\mbox{\tiny{$\langle J\!,\! \!\widetilde{\mathcal{K}}\rangle$}}},
$$
where $\omega_{\mbox{\tiny{can}}}$ is the canonical symplectic form on $T^{\ast}S^{2}$ and $\hat{B}_{\mbox{\tiny{$\langle J\!,\! \!\widetilde{\mathcal{K}}\rangle$}}}$ is defined in Theorem~\ref{T:general0-level}$(ii)$ (see also \cite{Koiller92}) as being $\phi^*\hat{B}_{\mbox{\tiny{$\langle J\!,\! \!\widetilde{\mathcal{K}}\rangle$}}} = B_{\mbox{\tiny{$\langle J\!,\! \!\widetilde{\mathcal{K}}\rangle$}}}$.  Therefore we conclude here that the 0-level set of the $D$-momentum $F$ has also a structure of a Chaplygin leaf.  

According to \cite{Yaroshchuk1992}, the system has an invariant measure, and since $Q/G$ has dimension 2, by the Chaplygin's reducing multiplier Theorem \cite{ChaplyginRedMult2008}, the form $\widetilde{\omega}$ admits a conformal factor. Then, the nonholonomic system on $F^{-1}{(0)}/G$ is symplectic after a time reparametrization. 

For completeness, the reduced nonholonomic bivector $\pi_{\red}$ on $\mathcal{M}/G$ is 
$$
    \pi_{\red} = (\mathbf{n}\times\partial_{\mathbf{n}})_{i}\wedge\partial_{M_{i}}-\textup{cyclic}\left[ \, (\mathbb{I}\Omega)_{1}\partial_{M_{2}}\wedge\partial_{M_{3}} \, \right],
$$
and, moreover, since ${H}_{\red} = \frac{1}{2}\langle\mathbf{M},\Omega\rangle$, the reduced nonholonomic vector field $X_{\red}$ is given by 
$X_{\red} = -\langle\mathbf{n}\times\Omega,\partial_{\mathbf{n}}\rangle+\langle\mathbf{M}\times\Omega,\partial_{\mathbf{M}}\rangle. $

\paragraph{Local coordinate approach.} Now, we use suitable coordinates to formulate explicitly the identification of the reduced zero level $F^{-1}(0)/G$ with the canonical symplectic manifold $T^{\ast}(S^{2})$ plus a magnetic term. First, define the following functions 
$$
    g(\mathbf{n}) := \frac{\langle \mathbb{I}A\mathbf{n},\mathbf{n} \rangle}{\langle\mathbb{I}A\mathbf{n} + mr_{0}^{2}\mathbf{n}\times(A\mathbf{n}\times \mathbf{n}),A\mathbf{n}\rangle}\quad\textup{and}\quad h_{j}(\mathbf{n}) := ((\mathbb{I}\mathbf{n}) \times (\mathbb{I}A\mathbf{n} + mr_{0}^{2}\mathbf{n}\times(A\mathbf{n}\times \mathbf{n})))_{j}, 
$$
and the vector fields $U_{1} := \langle \mathbf{n}-g(\mathbf{n})A\mathbf{n},X\rangle\quad\textup{and}\quad U_{2} := \langle \mathbf{h},X\rangle,$
where $\mathbf{h} := (h_{1},h_{2},h_{3})$.  Second, observe that $\widetilde{H} = \textup{span}\{U_1, U_2\}$ and $D=\widetilde{H}\oplus S = \textup{span}\{U_{1},U_{2},Y\}$.

Locally, if we consider $r_{3}\not=0$, we determine the distribution $W\subset V$, generated by the vector fields $Z_{1} := -V_{2}/r_{3}$ and $Z_{2} := V_{1}/r_{3}$ where $V_{1}$ and $V_{2}$ are given in (\ref{Eq:BMFinfGen}). Then $\widetilde{\mathfrak{B}}_{\mbox{\tiny{$TQ$}}} = \{U_{1},U_{2},Y,Z_{1},Z_{2}\}$ is an adapted basis of the splitting  $TQ = \widetilde{H}\oplus S\oplus W$, with dual basis $\widetilde{\mathfrak{B}}_{\mbox{\tiny{$T^*Q$}}} = \left\{\sigma^{1} ,\sigma^{2}, \sigma^{Y},\epsilon^{1},\epsilon^{2}\right\}$, where 
$$
    \sigma^{1} \!:= -\frac{\langle \mathbf{h},d\mathbf{n}\rangle}{g(\mathbf{n})\langle \mathbf{h},A\mathbf{n}\times \mathbf{n}\rangle}, 
\sigma^{2}\! := -\frac{\langle A\mathbf{n},d\mathbf{n}\rangle}{\langle \mathbf{h},A\mathbf{n}\times\mathbf{n}\rangle}, 
\sigma^{Y}\!\! := \frac{\langle \mathbb{I}A\mathbf{n} + mr_{0}^{2}\mathbf{n}\times(A\mathbf{n}\times \mathbf{n}),r_{3}\lambda+\gamma\times g^{T}\epsilon\rangle}{r_{3}\langle\mathbb{I}A\mathbf{n},\mathbf{n}\rangle}.
$$
We denote by $(v^{1},v^{2},v^{Y},v^{z_{1}},v^{z_{2}})$ and $(p_{1},p_{2},p_{Y},p_{z_{1}},p_{z_{2}})$ the respective coordinates in $TQ$ and $T^{\ast}Q$ associated to the bases $\widetilde{\mathfrak{B}}_{\mbox{\tiny{$TQ$}}}$ and $\widetilde{\mathfrak{B}}_{\mbox{\tiny{$T^*Q$}}}$.

The constraint manifold $\M$ is determined by the coordinates $(g,\mathbf{n}, p_{1},p_{2},p_{Y})$ where $p_{z_{a}} = \frac{\kappa_{1a}}{\kappa_{11}}p_{1}+\frac{\kappa_{2a}}{\kappa_{22}}p_{2}+\frac{(\kappa_{ya}-\kappa_{1a})}{(\kappa_{yy}-\kappa_{11})}(p_{Y}-p_{1})$, for $\kappa_{\alpha a} = \kappa(U_\alpha, Z_a)$ and $\kappa_{ya} = \kappa(Y, Z_a)$; and the Hamiltonian takes the form
$$
    H_{\subM}(g,\mathbf{n}, p_{1},p_{2},p_{Y}) = \frac{1}{2}\left( \frac{p_{1}^{2}}{\kappa_{11}}+\frac{p_{2}^{2}}{\kappa_{22}}+\frac{(p_{Y}-p_{1})^{2}}{\kappa_{yy}-\kappa_{11}}\right).
$$

The conserved quantity is written as $F(g,\mathbf{n}, p_{1},p_{2},p_{Y}) = \frac{(p_{Y}-p_{1})}{g(\mathbf{n})}$ defining a foliation given by the ($G$-invariant) leaves 
$\M_{c} := F^{-1}(c) = \{(g,\mathbf{n}, p_{1},p_{2},p_{Y})\in \mathcal{M};\, p_{Y} = p_{1}+cg(\mathbf{n})\}.$
In particular, for $c=0$  the reduced manifold $F^{-1}(0)/G$ is identified with the coordinates $({\bf n}, p_1,p_2)$ for ${\bf n}\in S^2$ showing that $F^{-1}(0)/G \simeq T^*(S^2)$ in agreement with Sec.~\ref{S:general-0level}. 

Moreover, let us denote by $\widetilde \sigma^i$ the 2-forms on $F^{-1}(0)/G$ so that $\rho_\subQ^*\widetilde \sigma^i = \sigma^i$ for $i=1,2$ and $\rho_\subQ: Q\to Q/G$, and by $C_{12}^Y, C_{12}^{a}$ (for $a=1,2$), the functions on $S^2$ defined by $C_{12}^Y :=\sigma^Y([U_1,U_2])$ and $C_{12}^{a} :=\epsilon^a([U_1,U_2])$. Therefore, the almost symplectic form $\widetilde \omega = \omega_{\mbox{\tiny{can}}} - \hat{B}_{\mbox{\tiny{$\langle J\!,\! \!\widetilde{\mathcal{K}}\rangle$}}}$ defining the reduced dynamics is given, in our chosen basis, by 
$$
    \omega_{\mbox{\tiny{can}}} := \widetilde \sigma^{j}\wedge dp_{j}-p_{j}d\widetilde\sigma^{j} \quad \textup{and}\quad \hat{B}_{\mbox{\tiny{$\langle J\!,\! \!\widetilde{\mathcal{K}}\rangle$}}} := (p_YC_{12}^Y  + p_{a}C_{12}^{a})\widetilde\sigma^{1}\wedge\widetilde\sigma^{2}, \quad \textup{where} \quad a=1,2,
$$
with Hamiltonian function $H_{\red}^0(\mathbf{n}, p_{1},p_{2}) = \frac{1}{2}\left( \frac{p_{1}^{2}}{\kappa_{11}}+\frac{p_{2}^{2}}{\kappa_{22}}\right).$

\begin{small}
\setlength{\itemsep}{-2cm}
\bibliographystyle{acm}
\bibliography{refs}
\end{small}

\end{document}